\documentclass[reqno,oneside,11pt,letterpaper]{amsart}

\usepackage[margin=3cm]{geometry}
\usepackage[T1]{fontenc}

\usepackage{latexsym}
\usepackage{amssymb}
\usepackage{esint}
\usepackage{bbm}
\usepackage{tensor}

\usepackage{enumitem}
\usepackage{xcolor}
\usepackage{array}
\usepackage{longtable}
\usepackage[hidelinks]{hyperref}
\usepackage{xurl}

\allowdisplaybreaks[4] % allow page breaks in multiline equations
\newcounter{smallarabics}
\newenvironment{arabicenumerate}
  {\begin{list}{{\normalfont\textrm{(\arabic{smallarabics})}}}
    {\usecounter{smallarabics}\setlength{\itemindent}{0cm}
     \setlength{\leftmargin}{5ex}\setlength{\labelwidth}{4ex}
     \setlength{\topsep}{0.75\parsep}\setlength{\partopsep}{0ex}
     \setlength{\itemsep}{0ex}}}
  {\end{list}}
\newcommand{\ben}{\begin{arabicenumerate}}
\newcommand{\een}{\end{arabicenumerate}}

\def\init{\setcounter{equation}{0}}

\newtheorem{theoreme}{Theorem}[section]
\newtheorem{proposition}[theoreme]{Proposition}
\newtheorem{lemma}[theoreme]{Lemma}
\newtheorem{definition}[theoreme]{Definition}
\theoremstyle{remark}
\newtheorem{remark}[theoreme]{Remark}
\theoremstyle{plain}
\newtheorem{corollary}[theoreme]{Corollary}
\newtheorem{hypothesis}[theoreme]{Hypothesis}

\newcommand{\beq}{\begin{equation}}
\newcommand{\eeq}{\end{equation}}
\newcommand{\defpage}[1]{p.~\pageref{#1}}
\newcommand{\appdefpage}[1]{p.~\pageref{#1}}

\def\rr{{\mathbb R}}
\def\cc{{\mathbb C}}
\def\nn{{\mathbb N}}

\renewcommand{\Re}{\operatorname{Re}}
\def\bar{\overline}
\def\reg{{\rm reg}}
\def\cinf{C^\infty}
\def\coinf{C_{\rm c}^\infty}
\def\proof{\noindent{\bf Proof.}\ \ }
\def\qed{$\square$\medskip}
\def\i{{\rm i}}
\def\p{ \partial}
\def\12{\frac{1}{2}}
\def\e{{\rm e}}
\def\s{{\rm s}}
\def\vol{{\rm vol}}

\def\dVol{\mathop{}\!d{\rm vol}}
\def\ad{{\rm ad}}
\def\sol{{\rm sol}}
\def\cl{{\rm cl}}
\def\cyl{{\rm cyl}}
\def\sc{\rm{sc}}
\def\c{\rm{c}}
\def\b{{\rm b}}
\def\norm{{\bf 0}}
\def\rq{{\bf q}}

\def\cA{\mathcal{A}}
\def\cD{{\mathcal D}}
\def\cE{{\mathcal E}}
\def\cF{{\mathcal F}}
\def\cH{{\mathcal H}}
\def\cN{{\mathcal N}}
\def\cS{{\mathcal S}}
\def\cU{{\mathcal U}}
\def\cV{{\mathcal V}}
\newcommand{\one}{\mathbbm{1}}
\def\bS{\mathbb{S}}
\def\Ric{{\bf Ric}}
\def\rg{{\bf g}}
\def\Riem{{\bf Riem}}
\def\rh{{\bf h}}
\def\hrg{{\bf{\mathsf{g}}}}

\def\Ran{{\rm Ran}}
\newcommand{\Ker}{{\rm Ker}}
\newcommand{\Kerop}{{\rm Ker}\,}
\DeclareMathOperator{\Dom}{Dom}
\DeclareMathOperator{\supp}{supp}
\def\dual{\!\cdot \!}
\def\CCR{{\rm CCR}}
\def\zero{{\mskip-4mu{\rm\textit{o}}}}
\def\Diff{{\rm Diff}}
\newcommand{\mat}[4]{\begin{pmatrix}#1 &#2  \\ #3 &#4 \end{pmatrix}}

\newcommand{\col}[2]{\begin{pmatrix}#1 \\#2\end{pmatrix}}
\newcommand{\open}[1]{\mathopen{}\mathclose{\left]#1 \right[}}

\newcommand{\traa}[1]{\mskip-6mu\upharpoonright_{#1}}
\def\tosim{\overset{\sim}{\to}}
\newcommand\moyal{\mathbin{\raisebox{-0.05em}{\text{\ttfamily\#}}}}
\makeatletter
\newcommand*{\defeq}{\mathrel{\rlap{%
                     \raisebox{0.3ex}{$\m@th\cdot$}}%
                     \raisebox{-0.3ex}{$\m@th\cdot$}}%
                     =}
\newcommand*{\eqdef}{=\mathrel{\rlap{%
                     \raisebox{0.3ex}{$\m@th\cdot$}}%
                     \raisebox{-0.3ex}{$\m@th\cdot$}}%
                     }
\makeatother

\def\WF{{\rm WF}}
\def\Op{{\rm Op}}

\def\calde{Calder\'{o}n }
\def\tc{\tilde{c}}
\def\rx{{\rm x}}
\def\jx{\langle x\rangle}
\def\tb{\tilde{b}}

\def\tD{\widetilde{D}}
\def\tK{\widetilde{K}}
\def\tM{\widetilde{M}}
\def\tV{\widetilde{V}}
\def\trho{\widetilde{\varrho}}

\newcommand\be{{}^{\rm b}}
\newcommand{\oper}[1]{\Psi^{#1}_{{\rm s}}(M)}
\newcommand{\operw}[2]{\Psi^{#1}_{{\rm s}\: #2}(M)}
\newcommand{\opern}[1]{\Psi^{#1}_{{\rm s}\: \norm}(M_{\norm})}

\newcommand{\operbun}[1]{\Psi^{#1}_{{\rm s}}(M; V)}
\newcommand{\opersig}[1]{\Psi^{#1}_{ {\rm s}}(\Sigma; V)}
\newcommand{\operwsig}[2]{\Psi^{#1}_{ {\rm s}\:#2}(\Sigma; V)}
\def\hV{\mathsf{V}}

\def\reta{\boldsymbol{\eta}}
\def\treta{\widetilde{\reta}}
\def\tdelta{\widetilde{\boldsymbol{\delta}}}

\newcommand{\cov}[1]{\nabla_{e_{#1}}}

\begin{document}
\title[Hadamard states for linearized gravity]{Hadamard states for linearized gravity on asymptotically flat spacetimes}
\author{Christian G\'erard}
\address{Laboratoire de Math\'ematiques d'Orsay, Universit\'e Paris-Saclay, France}
\email{christian.gerard@universite-paris-saclay.fr} 
\author{Mikhail Molodyk}
\address{Department of Mathematics,  Department of Physics and Astronomy, Texas A\&M University,United States}
\email{mam765@tamu.edu}
\author{Micha{\l} Wrochna} 
\address{Mathematical Institute, Universiteit Utrecht, The Netherlands\newline\hspace*{\parindent}%
Mathematics \& Data Science, Vrije Universiteit Brussel, Belgium}
\email{m.wrochna@uu.nl}
\email{michal.wrochna@vub.be}
 \keywords{linearized Einstein equations, microlocal analysis, Quantum Field Theory on curved spacetimes, asymptotically flat spacetimes, Hadamard states, $\b$-calculus}
 \subjclass[2020]{81T20, 83C05, 58J47, 58J45, 58J32}
\begin{abstract}
 We consider the quantization of linearized Einstein equations on $4$-dimensional asymptotically flat spacetimes.  We prove the existence of  fully gauge invariant Hadamard states in the harmonic gauge for any Einstein  metric $\rg$  asymptotic to the Minkowski metric in a conic neighborhood of spatial infinity.
\end{abstract}
\maketitle
\section{Introduction and summary}\label{sec0}\init
\subsection{Introduction}
 The quantization of linearized gravity has long remained a challenging open problem at the intersection of algebraic quantum field theory, partial differential equations, and Lorentzian geometry. While the classical framework of linearized gravity is  well understood \cite{FH, HS, BFR, BDM}, and various candidate physical states have been explored in the physics literature for highly symmetric background spacetimes, a rigorous mathematical quantization has remained elusive. In particular, a mathematically satisfactory answer has been lacking even in the fundamentally important setting of perturbations of Minkowski spacetime.

Within algebraic quantum field theory, the primary criterion for a state to be physical is the {\em Hadamard condition}. This microlocal regularity condition is indispensable for the rigorous renormalization of non-linear observables \cite{FNW, AFO, FV} and for the mathematical formulation of perturbative quantum gravity \cite{BFR, BFH, rejzner}. 

Compared to scalar fields or Dirac fermions, tensor gauge theories present severe analytical obstacles and 
 standard techniques for proving the existence of Hadamard states fail when applied directly to linearized gravity: 

The deformation argument \cite{FNW} is not applicable  because the background spacetime $(M, \rg)$ must  satisfy the non-linear Einstein equations, precluding arbitrary background deformations.

Pseudodifferential calculus methods \cite{junker, GW, GOW} typically violate gauge invariance if implemented naively.
 A pseudodifferential construction of Hadamard states is known in the case of Yang--Mills fields linearized around a possibly non-zero solution~\cite{GW1}. However  it uses  a spacetime deformation argument and various arguments  specific to differential forms.

For the same reason, another method introduced recently in \cite{MS} for Maxwell fields relying on a new Hodge decomposition seems also presently difficult to extend to linearized gravity. 

Conformal scattering approaches \cite{DMP} suffer from divergences at conformal infinity \cite{BDM} that limit the quantization to a subspace of the classical degrees of freedom.

For Maxwell, Proca, Yang-Mills, and Teukolsky fields \cite{furlani, FP, hollands, DS, FS, GW1, Moretti2023, Iuliano2023, MS, WrochnaZahn2017},  there exist several constructions of Hadamard states relying on different methods. Very recently, H\"afner and Klein~\cite{HaefnerKlein2026} constructed the Unruh state for bosonic Teukolsky fields of spins $0,\pm1,\pm2$ on subextreme Kerr spacetimes and proved its Hadamard property throughout the black hole exterior and the interior up to the inner horizon. 

For linearized gravity on highly symmetric backgrounds, most notably on the de Sitter spacetime $dS^{4}$, there exists a large body of works in the physical literature. A popular approach is via {\em Wick rotation}.  The remarkable feature is that the gauge-fixed linearized Einstein operator on $dS^{4}$  becomes  the canonical Laplacian on  the 4-sphere $\bS^{4}$ after Wick rotation. 

Its Euclidean Green's function provides a natural candidate for the two-point function of a state after Wick rotating back to Lorentzian signature. The so-obtained ``state'', if well-defined,  is called  the \emph{Euclidean} or \emph{Bunch--Davies} \emph{vacuum}. 
Let us mention,  among others, the works \cite{Allen1986,AFO,Faizal2012,Higuchi2001} and, for a recent rigorous analysis in the spirit of the present work, \cite{GW3}.

If the background spacetime $(M, \rg)$ is {\em analytic}, the Wick rotation in a Gaussian time coordinate with respect to some reference Cauchy surface seems a promising approach. In fact, one can construct from the elliptic operator obtained by Wick rotation, so-called {\em \calde projectors}, which are  well-known objects  in elliptic boundary value problems, to try to construct a quasi-free state for the Lorentzian theory. 

States obtained from \calde projectors depend on fewer arbitrary choices than those constructed by pseudodifferential calculus. Therefore, it is hoped that the crucial gauge invariance property will be automatically satisfied. This approach was followed in \cite{GMW}, providing a quasi-free state for linearized gravity but with gauge invariance and positivity holding only modulo the addition of some smooth corrections. This was due to the fact that the Dirichlet boundary conditions, whose introduction is needed for the invertibility of the Wick-rotated operators, have the effect of destroying the gauge invariance; cf.~\cite{Witten2021,CMS} for related issues concerning boundary conditions in Euclidean gravity. 

 Let us note that \calde projectors will actually play an essential role in our construction, even if no analyticity of the metric is assumed.

The existence of Hadamard states for linearized gravity has only recently been established on spacetimes possessing a {\em compact} Cauchy surface \cite{G2} by a method called {\em full gauge fixing}. 

In this paper, we extend the result of \cite{G2} to the case of {\em asymptotically flat spacetimes},   by providing a rigorous construction of gauge-invariant Hadamard states for linearized gravity (see Thm.~\ref{mainthmmain}). To the best of our knowledge, this work provides the first mathematical proof of the existence of such states for perturbations of Minkowski spacetime, overcoming the analytical roadblocks previously imposed by gauge freedom and background constraints.

\subsubsection{Linearized Einstein equations}
\leavevmode{}Let us first explain the general setting.

Let $(M, \rg)$ be a globally hyperbolic spacetime with $\dim M= 4$ solving the Einstein equations
\[
\Ric= \Lambda\rg,
\]
where $\Lambda\in \rr$ is the cosmological constant. Let $V_{k}\defeq\cc\otimes^{k}_{\rm s}T^{*}M$, $k=1,2$  be the \emph{complex} bundle of symmetric $(0, k)$-tensors on $M$. Note that we consider here {\em complex} tensor fields, which in our opinion makes the whole framework somewhat more elegant. The link with the more common real framework is explained for example in \cite[Sect.~2]{GW}. We consider the two differential operators
\[
P= - \square_{2}-  I \circ d \circ \delta + 2\,\Riem_{\rg},\quad K= I\circ d,
\]
where $\square_{2}$ is the  d'Alembertian  associated to $\rg$, $I$ is the trace reversal, $d$ and $\delta$ are the symmetric differential and co-differential, see Subsects.~\ref{sec1.1}, \ref{sec1.2}.
The {\em linearized Einstein equations} around $\rg$ are 
\begin{equation}
\label{e0.1}
Pu=0,
\end{equation}
where $u$ is a (symmetric) $(0,2)$-tensor. 
The identity $P\circ K=0$ implies that $\Ker_{\rm sc}P$ is invariant under 
 {\em linearized gauge transformations} given by $u\mapsto u+ K w$, where $w$ is a $(0, 1)$-tensor. Therefore the natural 'on-shell' phase space is  the quotient space:
\[
\dfrac{\Ker_{\sc}P}{\Ran_{\sc}K},
\]
of 'solutions modulo gauge'. Here  and below the subscripts $\sc$ resp.~$\c$ refer to 'space compact', resp.~'compact' for example $\Ker_{\sc}P$ is the space of (smooth) space compact solutions of \eqref{e0.1}.

\subsubsection{The phase space for linearized gravity}
It is customary to work in   the {\em de Donder} or {\em harmonic gauge}:
\[
K^{\star}u=0,
\]
where $K^{\star}= \delta$ is the adjoint of $K$ for a Hermitian form $(\cdot| \cdot)_{I, V_{2}}$ involving the trace reversal $I$, see \ref{not:physical-hermitian-forms}, for which $P$ is formally selfadjoint. The  de Donder gauge condition can be imposed by a gauge transformation, and the quotient space $\dfrac{\Ker_{\sc}P}{\Ran_{\sc}K}$ is isomorphic to
\[
\dfrac{\Ker_{\rm sc}D_{2}\cap \Ker_{\rm sc}K^{\star}}{K\Ker_{\rm sc}D_{1}},
\]
where
\[
\begin{aligned}
D_{1}&= K^{\star}\circ K= - \Box_{1}-\Lambda,\\
D_{2}&= P+ K\circ K^{\star}=- \Box_{2}+ 2\Riem_{\rg},
\end{aligned}
\]
are hyperbolic operators acting respectively on $(0,1)$- and $(0, 2)$-tensors and $K \Ker_{\sc}D_{1}$ represents the residual gauge freedom.  An important  consequence of $P\circ K=0$ is the identity
\begin{equation}
\label{e.bli3}
D_{2}\circ K= K\circ D_{1}.
\end{equation}
One can then introduce   the 'off shell' phase space:
\[
\cV_{P}= \dfrac{\Ker_{\c}K^{\star}}{\Ran_{\c}P}
\]
equipped with the Hermitian form
\[
\bar{[u]}\dual {\bf Q}_{P}[u]\defeq  \bar{u}\dual {\bf Q}_{I, 2}u, 
\]
where 
\[
\bar{u}\dual {\bf Q}_{I, 2}u= \i (u|I G_{2} u)_{V_{2}}, 
\]
and $G_{2}= G_{2\,{\rm ret}}- G_{2\,{\rm adv}}$ is the commutator function for $D_{2}$.

The \emph{algebraic quantization} of linearized gravity simply consists  in constructing the  $\CCR$ $*$-algebra $\CCR(\cV_{P},{\bf Q}_{P})$. 

\subsubsection{Hadamard states for linearized gravity}
A quasi-free state on  $\CCR(\cV_{P},{\bf Q}_{P})$  is defined by a pair of Hermitian forms $\Lambda^{\pm}_{P}$ on $\cV_{P}$ called {\em covariances} such that 
 \beq\label{e0.2}
\begin{aligned}
(1)&\quad \boldsymbol{\Lambda}^{\pm}_{P}=\boldsymbol{ \Lambda}_{P}^{\pm*},\quad \boldsymbol{\Lambda}_{P}^{\pm}\geq 0,\\
(2)&\quad \boldsymbol{\Lambda}_{P}^{+}- \boldsymbol{\Lambda}_{P}^{-}={\bf Q}_{P}.
\end{aligned}
\eeq
We will consider  covariances obtained from 
a pair of Hermitian  forms on $\coinf(M; V_{2})$ by
\beq\label{e0.3}
\bar{[u]}\dual \boldsymbol{\Lambda}^{\pm}_{P}[u]\defeq \bar{u}\dual \boldsymbol{\Lambda}_{2}^{\pm}u,
\eeq
The Hermitian  forms  $\boldsymbol{\Lambda}_{2}^{\pm}$ have to satisfy the following set of conditions:
\beq\label{e.bli1}
\begin{aligned}
(1)&\quad D_{2}^{*}\circ\boldsymbol{\Lambda}_{2}^{\pm}= \boldsymbol{\Lambda}_{2}^{\pm}\circ D_{2}=0,\\
(2)&\quad \boldsymbol{\Lambda}_{2}^{\pm}=0 \hbox{ on  } \Ker_{\c}K^{\star}\times \Ran_{\c}K,\\
(3)&\quad \boldsymbol{ \Lambda}_{2}^{\pm}= \boldsymbol{\Lambda}_{2}^{\pm*},\quad \boldsymbol{\Lambda}_{2}^{\pm}\geq 0, \hbox{ on } \Ker_{\c}K^{\star},\\
(4)&\quad  \boldsymbol{\Lambda}_{2}^{+}-\boldsymbol{\Lambda}_{2}^{-}={\bf Q}_{I, 2}, \hbox{ on } \Ker_{\c}K^{\star},\\
(5)&\quad \WF(\Lambda_{2}^{\pm})'\subset \cN^{\pm}\times \cN^{\pm},
\end{aligned}
\eeq
where $\Lambda_{2}^{\pm}\in\cD'(M\times M; L(V_{2}))$ are the distributional kernels of $\boldsymbol{\Lambda}_{2}^{\pm}$ and  $\cN^{\pm}$ are the two connected components of the characteristic manifold $\cN= \{(x, \xi)\in T^{*}M\setminus\zero: \xi\dual \rg^{-1}(x)\xi= 0\}$.

Condition (1) corresponds to  the  {\em field equations}, familiar from quantization of matter fields, while condition (2), specific to gauge fields, is the {\em gauge invariance}. Together they imply that $\boldsymbol{\Lambda}_{2}^{\pm}$ 'pass to quotient'  and define $\boldsymbol{\Lambda}^{\pm}_{P}$ by  \eqref{e0.3}.

Condition (3) is the {\em positivity}, while condition (4) corresponds to the {\em CCR}. Together they imply  conditions \eqref{e0.2} on $\boldsymbol{\Lambda}^{\pm}_{P}$.

The last condition (5) is  the {\em Hadamard condition}, which singles out Hadamard states, considered as the physically meaningful states with the correct short distance behavior of two-point functions.

\begin{definition}\label{def1}
 A quasi-free state $\omega_{P}$ on $\CCR(\cV_{P},{\bf Q}_{P})$ with covariances $\boldsymbol{\Lambda}^{\pm}_{P}$ generated by $\boldsymbol{\Lambda}_{2}^{\pm}$ satisfying \eqref{e.bli1} is called a \emph{ Hadamard state for linearized gravity} on $(M, \rg)$.
\end{definition}
Conditions (1), (4) and (5) are rather easy to satisfy. In fact the construction of Hadamard states for scalar fields via pseudodifferential calculus initiated in \cite{junker, junkererratum} and further developed in \cite{GW1, GOW} can be adapted to the tensor case and produces a wealth of covariances satisfying (1), (4) and (5), see \cite[Sect.~5]{GMW}. These covariances are unique modulo smoothing error terms. 

Conditions (2)  and (3) in contrary, are sensitive to smoothing error terms and very difficult to impose. 

\subsection{Asymptotically flat spacetimes}
In this paper we consider linearized gravity with cosmological constant $\Lambda=0$ and construct covariances $\boldsymbol{\Lambda}_{2}^{\pm}$ satisfying \eqref{e.bli1} on \emph{asymptotically flat spacetimes}. 
More precisely, denoting by $\reta$ the Minkowski metric on $\rr^{4}$ equipped with its canonical coordinates $(x_{0}, \rx)$, we consider Ricci-flat metrics $\rg$ on $\rr^{4}$ (i.e.~with $\Ric= 0$) satisfying the following two conditions:  firstly  for some  conical neighborhood $U$ of  the Cauchy surface $\Sigma= \{x_{0}=0\}$ the metric $\rg$ satisfies:
\beq\label{e.bli2}
\p_{x}^{\alpha}(\rg- \reta)\in O(\langle x\rangle)^{-1- |\alpha|},\ \ \alpha\in \nn_{0}^{4},\ x\in U.
\eeq
Secondly we assume that $\Sigma$ is a Cauchy surface for $(\rr^{4}, \rg)$.
It is known, see e.g.~\cite{CK, LR, HV} and the discussion in Remark~\ref{rem:bgmetric-existence}, that there exist Ricci-flat metrics satisfying \eqref{e.bli2}.

In this paper we will prove the following result.
\begin{theoreme}\label{mainthmmain}
For any Ricci-flat metric $\rg$ on $\rr^{4}$ satisfying the above assumptions, there exist Hadamard states for linearized gravity on $(\rr^{4}, \rg)$.
\end{theoreme}

\subsection{Summary of the construction}\label{sec0.3}
We now  give a summary of the construction of Hadamard states for linearized gravity done  in this paper.

\subsubsection{Hadamard projectors}

The construction of covariances satisfying the Hadamard condition relies on {\em Hadamard projectors} for $D_k$, which  are operators $c_{k}^{\pm}$, acting on Cauchy data for $D_k$ with: 
 \beq
\begin{aligned}
i)&\quad c_k^++c_k^-=\one,\quad (c_{k}^{\pm})^{2}= c_{k}^{\pm},\\
iii')&\quad \rq_{I,k}\circ  c_{k}^{\pm}= c_{k}^{\pm*}\circ  \rq_{I,k},\\
iv)&\quad \WF(\cU_{k}\circ c_{k}^{\pm})'\subset (\cN^{\pm}\cup\cF)\times T^{*}\Sigma,
\end{aligned}
\eeq
where  the  Cauchy solution operators $\cU_{k}$ and the charges $\rq_{I, k}$ are defined in \ref{sec1.6.2}.

One can then define a pair of Cauchy surface covariances for $D_{2}$ by:
\[
\boldsymbol{\lambda}_{2}^{\pm}= \pm \rq_{I, 2}\circ c_{2}^{\pm}.
\]
They satisfy the Hadamard condition, but are in general neither positive nor gauge invariant. One can, however, ensure their positivity for an auxiliary unphysical ``Euclidean charge''.

\subsubsection{Full gauge fixing}\label{sec0.3.1}
To  circumvent the difficulties  with conditions (2) and (3) we will use a method introduced in \cite{G2} relying on {\em full gauge fixing}. The main idea is to impose a gauge condition which fully fixes the gauge, making gauge invariance automatic, and at the same time restricts the Cauchy data of the metric perturbation to a subspace where the unphysical and physical charges match on the range of the Hadamard projectors.

%In order to explain it, it is convenient to fix a Cauchy surface $\Sigma$, taken as $\{x_{0}=0\}$, and to work with the equivalent Cauchy surface phase space equal to the quotient 
%$\dfrac{\Ker_{\c}K_{\Sigma}^{\dag}}{\Ran_{\c}K_{\Sigma}}$, where $K_{\Sigma}$, $K_{\Sigma}^{\dag}$ are Cauchy surface equivalents of $K, K^{\star}$, see \ref{sec1.6.2}.  
% 

Namely, we try to 
 find a supplementary space $E$ of $K\Ker_{\sc}D_{1}$ in $\Ker_{\sc}D_{2}\cap \Ker_{\sc}K^{\star}$. We can then identify the  phase space $\dfrac{\Ker_{\rm sc}D_{2}\cap \Ker_{\rm sc}K^{\star}}{K\Ker_{\rm sc}D_{1}}$ with $E$ using the associated projection
$\pi:\Ker_{\sc}D_{2}\cap \Ker_{\sc}K^{\star}\to E$. One then replaces $\boldsymbol{\lambda}_{2}^{\pm}$ by the {\em modified covariances}:
\[
\pi^{*}\circ \boldsymbol{\lambda}_{2}^{\pm}\circ \pi.
\] 
Applying the projection $\pi$ will automatically ensure gauge invariance.

The supplementary space $E$ has however to be chosen appropriately. Firstly  the Lorentzian  scalar product on  $(0, 2)$-tensor Cauchy data should be {\em positive} on $E$. This property will be needed to ensure the positivity of the state. 

  Secondly  $E$ should be adapted to   the microlocal Hadamard condition, i.e.~its definition should somehow distinguish between 'positive' and 'negative energy' solutions.
  
A gauge condition that satisfies the first requirement (but not the second) is the \emph{TT-synchronous gauge condition}. Working with solutions, this amounts, in  addition to the harmonic gauge condition $K^{\star}u=0$, to define the supplementary space $E$  by imposing the {\em traceless condition} ${\rm tr}_{\rg}u=0$ and the {\em synchronous condition} $u_{t\Sigma}=0$. 

To formulate the TT-synchronous gauge condition one needs of course to identify  $(M, \rg)$ with $(I_{t}\times \Sigma, - c^{2}dt^{2}+ \rh_{t})$, where $I$ is some interval, $\Sigma$ a Cauchy surface and $\rh_{t}$ a $t$-dependent family of Riemannian metrics on $\Sigma$.

It is sufficient to impose it to second order on $\Sigma$, i.e.~to require that
\[
{\rm tr}_{\rg}u\traa{\Sigma}=\nabla_{t}{\rm tr}_{\rg}u\traa{\Sigma}=0,\quad u_{t\Sigma}\traa{\Sigma}= \nabla_{t}u_{t\Sigma}\traa{\Sigma}=0.
\]

\subsubsection{The microlocal TT-synchronous gauge condition}\label{sec0.3.4}

An issue with  the TT-synchronous gauge condition is that the modified covariances obtained from it will in general {\em  not satisfy }the Hadamard condition anymore.
As in \cite{G2} we use instead a microlocal version of the synchronous gauge condition.  
In terms of Cauchy data $f=\col{f_{0}}{f_{1}} $  the {\em microlocal TT-synchronous gauge condition} is given by:  
\[
{\rm tr}_{\rg} \pi_{0}c_{2}^{\pm} f=0, \quad (\pi_{0}c_{2}^{\pm} f)_{t\Sigma}=0,
\]
where $\pi_{0}$ is the projection on the first component. 

The fact that the space $E$ defined by the microlocal TT-synchronous gauge condition is indeed a supplementary space to $\Ran_{\c}K_{\Sigma}$  is equivalent to the invertibility of some  {\em elliptic} system of PDE on $\Sigma$. This problem is uniquely solvable in appropriate Sobolev spaces modulo a smoothing error term. Since in the case of compact Cauchy surfaces a smoothing error is compact, this means the problem is \textit{Fredholm}, which is sufficient to define Hadamard covariances.

However, when $\Sigma$ is not compact (as in the spacetimes we consider), ellipticity alone is not sufficient to ensure invertibility or even the Fredholm property, because of  the \emph{infrared problem} caused by large distance behavior.

\subsubsection{Approach to the infrared problem}
We show that in asymptotically flat spacetimes, the construction of Hadamard projectors can still be done, but to similarly fix the gauge modulo compact errors one needs control at spatial infinity. We do this in two steps. First, we construct Hadamard projectors which are determined, as in previous work, only up to smoothing operators, but we ensure that they belong to a calculus of operators with good mapping properties between \textit{weighted b-Sobolev spaces}, in particular with well-defined \textit{normal operators}. In the second step, we revisit the Wick rotation method: we modify the Hadamard projectors near infinity, replacing their normal operators by those of Calder{\'o}n projectors of a Wick-rotated problem.

\subsubsection{The $\b$-framework}\label{sec0.3.2}
A standard method to adress the infrared problem is to use Melrose's $\b$-\emph{calculus}. In its simplest version, it amounts for example to replacing the flat Laplacian $- \Delta$ on $\rr^{n}$ by its $\b${\em -transform}  $-\Delta_{\b}= -\rho^{-1-n/2}\Delta \rho^{-1+ n/2}$, where $\rho= r^{-1}$ is the inverse distance to the origin. It is then easy to see that $-\Delta_{\b}$ is invertible  between $\b$-{\em Sobolev spaces} if $n\geq 3$, see \ref{sec9.2.3} for the definition, and similarly the $\b$-transform of a second order elliptic operator asymptotic to $- \Delta$ at infinity will be {\em Fredholm} between $\b$-Sobolev spaces.

We will perform  this transformation  at the level of the original operators $D_{1}, D_{2}$, replacing $D_{k}$ by $D_{k\b}= \rho^{-1-n/2}D_{k}\rho^{-1+ n/2}$, where $\rho= r^{-1}$ but now $r$ comes from perturbed {\em Beig-Schmidt coordinates} $(t, r, \omega)$, adapted to the background metric $\rg$, see Subsect.~\ref{sec2.2}.

We also transform $K$ into $K_{\b}= \rho^{1- n/2}K \rho^{-1+ n/2}$, so that  the key identity $D_{2\b}\circ K_{\b}= T_{\b}\circ D_{1\b}$, analogous to \eqref{e.bli3}, is satisfied with $T_{\b}=\rho^{-2}K_{\b}\rho^2$. All the constructions  will then be done  in the $\b$-framework. 

An essential property of $\b$-differential or $\b$-pseudodifferential operators is that they possess a {\em normal operator}, describing their behavior near $\rho= 0$. In particular $D_{k\b}, K_{\b}$ have normal operators $D_{k\b\norm}, K_{\b\norm}$. 

The operators $D_{k\b\norm}$  are actually $\b$-transforms of the  d'Alembertians for the flat metric $\reta$ written in the Beig-Schmidt coordinates. 

They are analytic in $t$ and hence setting $t= \i s$ can be Wick-rotated  to elliptic operators $\widetilde{D}_{k\b\norm}$, which are   $\b$-transforms  of the flat Laplacian on $\rr^{n}$ written in polar coordinates (for related results on elliptic and Fredholm theory on asymptotically Euclidean spaces, including applications to linearized gravity, see \cite{HiG,HiAE,VaLow}).

\subsubsection{Hadamard projectors in the \b-framework}\label{sec0.3.3}

We then follow the same strategy in the b-framework and construct {\em Hadamard projectors} for $D_{k\b}$, which  are operators $c_{k}^{\pm}$, acting on Cauchy data for $D_{k\b}$ with: 
 \beq\label{condit-hada-weak}
\begin{aligned}
i)&\quad c_k^++c_k^-=\one,\quad (c_{k}^{\pm})^{2}= c_{k}^{\pm},\\
iii')&\quad \rq_{I,k\b}\circ  c_{k}^{\pm}= c_{k}^{\pm*}\circ  \rq_{I,k\b},\\
iv)&\quad \WF(\cU_{k}\circ c_{k}^{\pm})'\subset (\cN^{\pm}\cup\cF)\times T^{*}\Sigma,
\end{aligned}
\eeq
where $\rq_{I, k\b}$ is the charge corresponding to $\rq_{I, k}$ in the b-framework. 

The construction of $c_{k}^{\pm}$ is equivalent to a factorization of $D_{k\b}$ and relies on the $\b$-pseudodifferential calculus. 
It is unique modulo smoothing $\b$-pseudodifferential operators. 

 To ensure that the Hadamard projectors $c_{k}^{\pm}$ have the correct behavior near $\rho= 0$, we perform a {\em transplant}, i.e.~we replace their normal operators by the {\em \calde projectors} for the Wick-rotated normal operators $\widetilde{D}_{k\b\norm}$. 
 
We only need to consider the normal operator at this stage, which comes from the exact Minkowski model at infinity and is therefore analytic, so its Wick rotation is simple to define. Moreover, the Wick-rotated problem is naturally stated in b-Sobolev spaces without the need for artificial boundary conditions (which were a significant obstacle in \cite{GMW}). Since the Calder{\'o}n projectors are naturally gauge-invariant, the Hadamard projectors obtained are gauge-invariant modulo an error which is not only smoothing but also decays at infinity.

This is possible thanks to the fact which originally motivated the use of Calder{\'o}n projectors in \cite{GW2,GMW}: the Calder{\'o}n projectors satisfy, up to Wick rotation, the same symbolic equations as the Hadamard projectors (in our case, both on the normal operator level), which implies that they differ only by a smoothing operator. This allows one to separately deal with the differential and decay aspects of the problem, adjusting the normal operator of the Hadamard projectors to have the desired asymptotic behavior while preserving the Hadamard property.

As before, we then define a pair of Cauchy surface covariances for $D_{2\b}$ by:
\[
\boldsymbol{\lambda}_{2\b}^{\pm}= \pm \rq_{I, 2\b}\circ c_{2}^{\pm}.
\]

\subsubsection{Gauge-fixing in the b-framework}

Once this is achieved, we follow \cite{G2} and impose the microlocal TT-synchronous gauge condition, see Subsect.~\ref{sec0.3.4}. 

The fact that the subspace $E$ is supplementary to $K_{\b}\Ker_{\sc}D_{1\b}$ is equivalent to the invertibility  of a system  of $\b$-pseudodifferential operators  acting on Cauchy data. This system can be analyzed via the $\b$-calculus. The fact that $c_{k}^{\pm}$ have the correct behavior at $\rho= 0$ ensures that this system is {\em Fredholm} between appropriate $\b$-Sobolev spaces. 

If the system is invertible (the so-called {\em regular case}, see Subsect.~\ref{sec11.3}), then the existence of a Hadamard state for linearized gravity follows rather easily.

If it is not invertible (the so-called {\em singular case}, see Subsect.~\ref{sec11.4}) 
then by Fredholm theory it still has finite dimensional kernel and cokernel. We then follow arguments in \cite{G2} and further alter the modified covariances by some finite rank smoothing operators to obtain a Hadamard state.

\subsection{Notation}\label{sec0.1}
We now collect various notation used throughout the paper. The detailed $\b$-pseudodifferential calculus, normal operators, and Fredholm criterion used later are recalled in Appendix~\ref{sec9}. An elementary flow lemma used in the metric reduction is recorded separately in Appendix~\ref{sec:app-ode}.
\subsubsection{Isomorphisms of vector spaces}
If $E, F$ are vector spaces and $a\in L(E, F)$ we write $a: E\tosim F$ if $a$ is an isomorphism. If $E, F$ are topological vector spaces, we use the same notation if $a$  is a homeomorphism.
\subsubsection{Sesquilinear forms}
If $E$ is a complex vector space, its anti-dual is denoted by $E^{*}$. 
A sesquilinear form ${\bf a}$ on $E$ is an element of $L(E, E^{*})$ and its action on elements of $E$ is denoted by $\bar{u}\dual {\bf a}v$. 
\subsubsection{Projections}
If $F\subset E$ are two vector spaces we say that $\pi: E\to F$ is a projection if  $\pi^{2}= \pi$ and  $\Ran \pi= F$.

\subsubsection{Operators on quotient spaces}\label{sec0.2.1}

Let $F_i\subset E_i$, $i=1,2$ be vector spaces and let $a\in L(E_1, E_2)$. Then the induced map
\[
[a]\in L( E_1/F_1, E_2/F_2),
\]
 is
\begin{equation}
\label{e0.01}
\begin{aligned}
1)&\quad \hbox{well-defined if }a E_1\subset E_2\hbox{ and }a F_1\subset F_2,\\
2)&\quad \hbox{injective iff }a^{-1}F_2=F_1,\\
3)&\quad \hbox{surjective iff }E_2=a E_1+F_2.
\end{aligned}
\end{equation}

\subsubsection{Sesquilinear forms on quotients}\label{sec0.2.2} 
Let now $F\subset E$ be vector spaces and let ${\bf c}\in L( E, E^*)$. We denote by $F^\circ\subset E^*$ the annihilator of $F$. Then the induced map
\[
[{\bf c}]\in L( E/F,(E/F)^*),
\]
 is
\begin{equation}
\label{e0.02}
\begin{aligned}
1)&\quad \hbox{well-defined if }{\bf c}E\subset F^{\circ},\quad F\subset\Kerop {\bf c},\\
2)&\quad \hbox{non-degenerate iff }F=\Kerop{\bf c}.
\end{aligned}
\end{equation}
If ${\bf c}$ is hermitian or anti-hermitian then the condition $F\subset\Kerop {\bf c}$ implies  ${\bf c}E\subset F^{\circ}$ (and vice versa).

\subsubsection{Vector bundles}\label{sec01.1.1}
Let  $V\overset{\pi}{\to}M$ be  a finite rank complex vector bundle  over a smooth manifold $M$.  

- If $\Sigma\subset M$ is a smooth manifold we denote by $V|_{\Sigma}\overset{\pi}{\to}\Sigma$ the restriction of $V$ to $\Sigma$.

- We denote by $\cinf(M; V)$, resp.~$\coinf(M; V)$ the space of smooth, resp.~compactly supported smooth sections of $V$.  
 
 - We denote by $\cD'(M; V)$, resp.~$\cE'(M; V)$ the space of distributional, resp.~compactly supported distributional  sections of $V$.

- We use the same notation if $V$ is a finite dimensional vector space, i.e.~we write simply $V$ instead of the trivial vector bundle $M\times V$.

- We denote by  $V_{k}\defeq\cc\otimes^{k}_{\rm s}T^{*}M$ 
 the complex bundle of symmetric $(0,k)$-tensors over $M$.

\subsubsection{Spacetimes}
We use the convention  $(-,+,\dots,+)$ for the Lo\-rent\-zian signature. 

- If $(M, \rg)$ is a spacetime, we denote by $J_{\pm}(K)$ the future/past causal shadow of  $K\subset M$.

- If $M$ is a globally hyperbolic spacetime  with $V\overset{\pi}{\to}M$  a finite rank vector bundle, we   denote by $\cinf_{\rm sc}(M;V)$ the space of space-compact sections, i.e.~sections in $\cinf(M;V)$ with compactly supported restriction to a Cauchy surface.

\subsubsection{Distributional kernels and wavefront sets}\ 

- If $u\in \cD'(M; V)$ we denote by $\WF (u)\subset T^{*}M\setminus\zero$ its {\em wavefront set}, which is invariantly defined using local trivializations of $V$. 

- If $V_{i}\overset{\pi}{\to}M_{i}$ are two vector bundles as above and $a: \coinf(M_{1}; V_{1})\to \cD'(M_{2}; V_{2})$ is linear continuous, then $a$ admits a distributional kernel, still denoted by $a\in \cD'(M_{2}\times M_{1}; V_{2}\boxtimes V_{1})$.  

- We denote by $\WF(a)'\subset (T^{*}M_{2}\times T^{*}M_{1})\setminus \zero$ its {\em primed} wavefront set, defined by 
  \[
  \Gamma'= \{((x_{2},\ \xi_{2}),\ (x_{1},\ -\xi_{1})):\ ((x_{2},\ \xi_{2}),\ (x_{1},\ \xi_{1}))\in \Gamma\}\hbox{ for }\Gamma\subset T^{*}M_{2}\times T^{*}M_{1}.
  \]

 \subsubsection{Symbol classes}\label{secnot.1.1}
 We denote by $S^{m}(\rr^{n})$ for $m\in \rr$ the space of functions in $\cinf(\rr^{n})$  such that
 \[
 \p_{x}^{\alpha}u(x)\in O(\jx^{m- |\alpha|}),\ \ \alpha\in \nn.
 \]
 Using the polar coordinates $r= |x|,\quad \omega= \dfrac{x}{|x|}$ in $\{|x|\geq 1\}$ this is equivalent to
 \[
 \p_{r}^{\alpha}\p_{\omega}^{\beta}u(r, \omega)\in O(r^{m- \alpha}),\ \ \alpha\in \nn\hbox{ in }\{|x|\geq 1\}.
 \]
A symbol $u\in S^{m}(\rr^{n})$ is classical if there exist $u_{k}\in S^{m-k}(\rr^{n})$ homogeneous of degree $m-k$ in $|r|\geq 1$ such that for each $N\in \nn$ $u-\sum_{j=0}^{N}u_{k}\in S^{m-N-1}(\rr^{n})$. We denote by $S^{m}_{\cl}(\rr^{n})$ the space of classical symbols. 

Finally for $r>0$ we denote by $S^{m}_{r}(\rr^{n})$ the space $S^{m}_{\cl}(\rr^{n})+ S^{m-r}(\rr^{n})$. Symbols in $S^{m}_{r}(\rr^{n})$ have a principal part which is a function on $\rr^{n}$ homogeneous of degree $m$. We will only use the case $r=1$ in the sequel.

The spaces $S^{m}_{\cl}(\rr^{n})$, $S^{m}(\rr^{n})$ resp.~$S^{m}_{r}(\rr^{n})$ correspond to the spaces $x^{-m}C^{\infty}(M)$, $x^{-m}\cA^{\infty}(M)$ resp.~$x^{-m}\cA^{\infty}_{r}(M)$ introduced below in \ref{not:conormal-classes} if we take $M=\overline{\rr^{n}}$ the radial compactification of $\rr^{n}$.

\subsubsection{Tensor valued symbol classes}\label{not:tensor-valued-classes}

We will need to consider not only spaces of scalar functions  but also spaces of sections of $V_{k}= \cc\otimes_{\rm s}^{k}T^{*}\rr^{n}$.  We set $d_{k}= {\rm dim}V_{k}$.

To this end  we identify $V_{k}$ with $\rr^{n}\times \cc^{d_{k}}$ using the frame obtained from $(dx^i)$ which form a  smooth frame of the \emph{scattering cotangent bundle}  ${}^{\rm sc}T^*\overline{\rr^n}$. 

Near $\partial\overline{\rr^n}$ we write $\rho=|x|^{-1}$. We use these scattering frames to define tensor-valued symbol classes on  $\rr^{n}$: if
\[
        u=\sum_{I=(i_{1}, \dots, i_{k})} u_I\,dx^{i_{1}}\otimes\cdots dx^{i_{k}},
\]
then membership of $u$  in $S^{m}(\overline{\rr^n}; V_{k})$ etc, means membership of all scalar coefficients $u_I$ in the corresponding scalar class.

   \subsubsection{Differential operators}\label{secnot.1.1b}
 
 We denote by $\Diff_{}^{m}(\rr^{n})$  resp.~$\Diff_{1}^{m}(\rr^{n})$ for $m\in \nn$ the space of differential operators $P$ of order $m$ on $\rr^{n}$ such that near $r=\infty$ 
 \[
 P= \sum_{|\alpha|+ |\beta|\leq m}a_{\alpha, \beta}(r, \omega)\p_{r}^{\alpha}(r^{-1}\p_{\omega})^{\beta},\ \ a_{\alpha, \beta}\in S^{|\alpha|+|\beta|-m}(\rr^{n})\hbox{ resp.~}
S_{1}^{|\alpha|+|\beta|-m}(\rr^{n}).
 \]
 The space $\Diff_{}^{m}(\rr^{n}; V_{k})$ resp.~$\Diff_{1}^{m}(\rr^{n}; V_{k})$ is defined similarly, by requiring that $a_{\alpha,\beta}\in S^{|\alpha|+|\beta|-m}(\rr^{n}; L(V_{k}))$ resp.~$S_{1}^{|\alpha|+|\beta|-m}(\rr^{n}; L(V_{k}))$, again using the frame obtained from $(dx^{i})$.

These orders reflect the fact that each radial or rescaled angular derivative lowers the growth order by one: every term of $P$ maps $S^{\ell}$ into $S^{\ell-m}$. In other words, $r^{m}P$ has bounded conormal coefficients when written in terms $r\p_r$ and $\p_\omega$ (as in the $\b$-framework introduced in the next subsections). For example, the Euclidean Laplacian belongs to $\Diff^2_1$:
	\[
	-\Delta=-\p_r^2-\frac{n-1}{r}\p_r-r^{-2}\Delta_{\bS^{n-1}}
	=r^{-2}\bigl(-(r\p_r)^2-(n-2)r\p_r-\Delta_{\bS^{n-1}}\bigr).
	\]
	The first order coefficient has order $-1$, as required when $m=2$.

 \subsubsection{Manifolds with boundary}\label{secnot.1.2}
Here and in the following subsections $M$ is a smooth $n$-dimensional manifold with compact boundary $\p M= Y$ and $x$ a boundary defining function. 

To keep the discussion simple we will fix the boundary definition function $x$, and will not try to study the invariance of various objects introduced below under a change of boundary defining function.

We identify a neighborhood of $\p M$ in $M$ with $\open{0, 2\epsilon}_{x}\times Y$ and set $M_{0}= M\setminus \open{0, \epsilon} \times Y$, $M_{\infty}=\open{0, 2\epsilon} \times Y$. We can view $M$ as the manifold obtained by gluing $M_{0}$ with $M_{\infty}$ over $\open{\epsilon/2, 3\epsilon} \times Y$.

We denote by $M_{\norm}$ the Cartesian product
\[
M_{\norm}= \open{0, +\infty}_{x}\times Y, 
\]
which serves as the  model space for $M$ near $\{x=0\}$. The dilation group $\{T_{\lambda}\}_{\lambda\geq 0}$ acts on functions on $M_{\norm}$ by $T_{\lambda}:u(x, y)\mapsto  u(\lambda x, y)$.

\subsubsection{Conormal function classes}\label{not:conormal-classes}

We denote by $\cinf(M)$ the space of functions on $M$ smooth up to $\{x=0\}$. We denote by $\cA^{\infty}(M)$ the space of smooth functions on $M^{\circ}$  such that 
\[
(x\p_{x})^{\alpha}\p_{y}^{\beta}u\in L^{\infty}(M_{\infty})\  (\alpha, \beta)\in\nn^{n}\hbox{ over }M_{\infty},
\]
i.e.~the space of functions {\em conormal} to $\{x=0\}$. We write $\cA^{\infty,r}(M)=x^r\cA^{\infty}(M)$.
  
We denote by $\cA_{r}^{\infty}(M)$ for $r>0$ the space of smooth functions on $M^{\circ}$ such that 
\[
u(x, y)= u_{\norm}(y)+ x^{r}v(x, y)\hbox{ over }M_{\infty},\ \ \hbox{ for }u_{\norm}\in \cinf(Y),\ v\in \cA^{\infty}(M).
\]
We have $\cinf(M)\subset \cA_1^{\infty}(M)$, and we will only use the case $r=1$ in the sequel.

\subsubsection{{\em b}-cotangent bundle}\label{not:b-tensor-bundles}
The space of sections of the {\em $\b$-tangent bundle} $\be TM$ is the $\cinf(M)$ module generated by the frame $(x\p_{x}, \p_{y})$, i.e.~the smooth vector fields on $M$ tangent to $\p M$. Similarly the space of sections of the {\em $\b$-cotangent bundle } $\be T^{*}M$ is the $\cinf(M)$ module generated by $(\dfrac{dx}{x}, dy)$.

We define the space $\cinf(M; \be T^{*}M)$ of sections of $\be T^{*}M$ with coefficients smooth up to the boundary, and similarly with $\cinf$ replaced by other function spaces, using the frames $(x\p_{x}, \p_{y})$, resp.~$( \dfrac{d x}{x}, dy)$.

Over $M_{\infty}$, the coordinates on $\be T^{*}M$ associated   to the frame $(\dfrac{dx}{x}, dy)$ are denoted by $(x, y, \zeta)$,   $\zeta= (\xi, \eta)$, corresponding to writing $\omega\in \be T^{*}M$ as $\xi\frac{dx}{x}+ \eta dy$.

\subsubsection{{\em b}-differential operators}
\label{not:b-diff}

We denote by  $\Diff^{m}_{\b}(M)$  the space of $\b$-differential operators. We recall that $P\in \Diff^{m}_{\b}(M)$ if $P\in\Diff^m(M^{\circ})$ and over $M_{\infty}$   
\[P= \sum_{k+ |\alpha|\leq m}a_{k \alpha}(x, y)D_{y}^{\alpha}(xD_{x})^{k},\ \ a_{k, \alpha}\in \cinf(M).
\]
Similarly we say that $P\in \Diff^{m}_{\b 1}(M)$ if over $M_{\infty}$ it has the same form but with $a_{k, \alpha}\in \cA^{\infty}_{1}(M)$.

In both cases the {\em normal operator} of $P$ denoted by $P_{\norm}$ is the dilation invariant differential operator acting on $M_{\norm}$:
\[
P_{\norm}= \sum_{k+ |\alpha|\leq m}a_{k \alpha\norm}(y)D_{y}^{\alpha}(xD_{x})^{k},
\]
where we recall that $a_{\norm}(y)= a(0, y)$.

\subsubsection{The cylinder picture}
It is convenient to set $x= \e^{-s}$ over $M_{\infty}$, so that $ \frac{dx}{x}= -ds$ and $x\p_{x}= - \p_{s}$. 

We set then  $M_{\infty\:{\rm cyl}}=\open{{\rm -}\ln(2\epsilon), +\infty}\times Y_{y}$ and we denote by $M_{\rm cyl}$ the manifold obtained by gluing $M_{0}$ with $M_{\infty\:{\rm cyl }}$. The corresponding identification $M\tosim M_{\cyl}$ will be loosely called the {\em cylinder map}. 

The image of $M_{\norm}$ under the cylinder map is $M_{\norm\:\cyl}= \rr_{s}\times Y$. Under the cylinder map  the dilation group $\{T_{\lambda}\}_{\lambda\geq 0}$ becomes  the translation group  in $s$ on $M_{\norm\:\cyl}$.

Under the cylinder map $\cA^{\infty}(M)$   is identified with  the space $\cinf_{\rm u}(M_{\cyl})$ of smooth functions on $M_{\rm cyl}$ uniformly bounded with all derivatives.

\subsubsection{Densities }\label{sec9.2.1}
We equip $M$ with a smooth density $d\mu$ such that over $M_{\infty}$
\[
d\mu= m(x, y)\dfrac{dx}{x}d^{n-1}y,\ \ m(x, y)\in \cA^{\infty}_{1}(M).
\]
The corresponding density on $M_{\rm cyl}$ is denoted by $d\mu_{\cyl}$. 

$M_{\norm}$ is equipped with the dilation invariant density
\[
d\mu_{\norm}= m_{\norm}(y)\dfrac{dx}{x}d^{n-1}y.
\]

\subsubsection{$\b$-Sobolev spaces}\label{sec9.2.3}
For $m\in \rr$ we denote by  $H^{m}(M_{\rm cyl})$  the usual Sobolev spaces on $(M, d\mu_{\rm cyl})$ and by $H^{m}_{\b}(M)$
its image under the cylinder map. For $m\in \nn$ we have \[
u\in H^{m}_{\b}(M)\hbox{ iff }u_{|M_{0}}\in H^{m}(M_{0}),\quad (x\p_{x})^{\alpha}\p_{y}^{\beta}u_{|M_{\infty}}\in L^{2}(M_{\infty}, d\mu),\ \ \alpha+ |\beta|\leq m.
\]

For $m, r\in \rr$  we set $H^{m, r}_{\b}(M)= x^{r}H^{m}_{\b}(M)$ equipped with its natural  norm.  The space $H^{0}_{\b}(M)$ is often denoted simply by $L^{2}_{\b}(M)$. We also write $H^{\infty,r}_{\b}(M)=\bigcap_{m\in\rr}H^{m,r}_{\b}(M)$.

\subsubsection{Radial compactification of $\rr^n$}
\label{not:rad-compact}
We will mostly  apply these notions when $M=\overline{\rr^n}$, the radial compactification of $\rr^n$ (or a subset of it).  
Again we will  need to consider also  spaces of sections of $V_{k}$ which we identify with $M\times \cc^{d_{k}}$ as in  \ref{not:tensor-valued-classes}.

 We use this identification to  
define the spaces $\cinf(M; V_{k})$, $\cA^\infty(M; V_{k})$, $\cA^{\infty,r}(M; V_{k})$,
\hfill\linebreak $\rho^r\cinf(M; V_{k})$ or $H^{m, r}_{\b}(M; V_{k})$.

\section{Linearized gravity}\label{sec1}\init
In this section we review the quantization of linearized gravity, following \cite{HS} and \cite{GW1}. 

\subsection{Notation} \label{sec1.1}
 Let $(M,\rg)$ be an $n$-dimensional Lorentzian spacetime satisfying the Einstein equation, i.e.~$\Ric=0$ in the notation discussed later. Often  $(M,\rg)$ will be  obtained by restricting a globally hyperbolic space-time to a neighborhood of a Cauchy surface $\Sigma$.
 
\subsubsection{Hermitian forms on tensors}\label{sec1.1.2}
 We denote by  $V_{k}\defeq\cc\otimes^{k}_{\rm s}T^{*}M$ 
 the complex bundle of symmetric $(0,k)$-tensors   equipped with the non-degenerate Hermitian form
\begin{equation}
\label{eq:uvk}
(u| v)_{V_{k}}\defeq k! \,\bar{u}\dual (\rg^{\otimes k})^{-1}v.
\end{equation}
The $k!$ normalization differs from the most common convention, but  has the advantage that various expressions involving adjoints are more symmetric.
If $U\subset M$ is open, we set
\begin{equation}
\label{uvkU}
(u|v)_{V_{k}(U)}=  \int_{U}(u(x)| v(x))_{V_{k}}\dVol_{\rg}, \ \ u,\ v\in \coinf(U;V_{k}),
\end{equation} 
and if  $\Sigma\subset M$ is a Cauchy surface, we set
\beq\label{uvksig}
(u|v)_{V_{k}(\Sigma)}=  \int_{\Sigma}(u(x)| v(x))_{V_{k}}\dVol_{\rh}, \ \ u,\ v\in \coinf(\Sigma;V_{k}),
\eeq
where $\dVol_{\rh}$ is the induced density on $\Sigma$.
\subsubsection{Trace reversal}\label{blurck}
 We set:
$I\defeq \one - \frac{1}{4}|\rg)(\rg|$ 
where
\[
(\rg| : u_{2}\mapsto (\rg| u_{2})_{V_{2}}, \quad |\rg): u_{0}\mapsto u_{0}\rg.
\]
Thus, on two-tensors,
\[
        (Iu)_{\mu\nu}=u_{\mu\nu}-\frac12\rg_{\mu\nu}u^\rho{}_{\rho}.
\]
We have $I= I^{*}$ for $(\cdot| \cdot)_{V_{2}}$, $I\rg=(1-\frac{n}{2})\rg$, $I$ is invertible for $n\neq 2$. In the physical case $n=4$ we have $I^{2}=\one$ and $I$ is the orthogonal symmetry with respect to $\cc\rg$ and called the  \emph{trace reversal}.

\subsubsection{Physical Hermitian forms}\label{not:physical-hermitian-forms}
Using $I$ we can define the {\em physical Hermitian forms} on the fibers of $V_{k}$ as
\[
(u| v)_{I,V_{1}} \defeq (u|v)_{V_{1}},\quad (u| v)_{I,V_{2}} \defeq (u| I v)_{V_{2}},
\]
and denote by $(\cdot| \cdot)_{I, V_{k}(U)}$ the corresponding analogs of $(\cdot| \cdot)_{V_{k}(U)}$. We denote by $A^{\star}$ the formal adjoint of $A$ for $(\cdot| \cdot)_{I, V_{k}(M)}$ to distinguish it from the formal adjoint $A^{*}$ for $(\cdot| \cdot)_{V_{k}(M)}$.

 \subsection{Operators on tensors}\label{sec1.2}
We denote by $d$ the symmetric differential and by $\delta$ its formal adjoint for $(\cdot| \cdot)_{V_{k}(M)}$, see e.g.~\cite[Subsect.~2.2]{G1}. Indices are raised with $\rg^{-1}$ and repeated indices are contracted. With the normalization of the Hermitian forms in \eqref{eq:uvk}, the basic tensorial formulas are, for $f\in\cinf(M)$, $w\in\cinf(M;V_1)$ and $u\in\cinf(M;V_2)$,
\[
\begin{aligned}
        (df)_{\mu}&=\nabla_\mu f,\quad
        (dw)_{\mu\nu}=\frac12(\nabla_{\mu}w_{\nu}+\nabla_{\nu}w_{\mu}),\\
        \delta w&=-\nabla^\mu w_\mu,\quad
        (\delta u)_\mu=-2\nabla^\nu u_{\nu\mu}.
\end{aligned}
\]

We set
 \beq\label{e0.3bip}
 D_{1}= - \square_{1},\quad
D_{2}= - \square_{2}+ 2 \Riem_{\rg},
 \eeq
 where $-\square_{k}$ is the connection d'Alembertian and $\Riem_{\rg}$ denotes the curvature action on symmetric two-tensors, with convention $(\Riem_{\rg}u)_{\mu\nu}=\mathbf{R}_{\mu}{}^{\rho}{}_{\nu}{}^{\sigma}u_{\rho\sigma}$, see e.g.~\cite[2.3.2]{G1}. Explicitly,
\[
\begin{aligned}
        (D_1w)_\mu&=-\nabla^\rho\nabla_\rho w_\mu,\\
        (D_2u)_{\mu\nu}&=-\nabla^\rho\nabla_\rho u_{\mu\nu}+2 \mathbf{R}_\mu{}^\rho{}_{\nu}{}^\sigma u_{\rho\sigma}.
\end{aligned}
\]
$D_{k}$ are hyperbolic operators with principal symbol $\xi\dual \rg^{-1}\xi\one_{V_{k}}$.  We will denote by $G_{k}=G_{k{\rm ret}}-G_{k{\rm adv}}$ their causal propagators.
The \emph{linearized Einstein equations} are
 \[
 Pu=0,\quad P\defeq D_{2}- I\circ d\circ \delta.
 \]
The operator producing pure gauge solutions is $K\defeq I\circ d$ and $K^{\star}=\delta$. In coordinates,
\[
\begin{aligned}
        (Kw)_{\mu\nu}
        &=\frac12(\nabla_\mu w_\nu+\nabla_\nu w_\mu)-\frac12\rg_{\mu\nu}\nabla^\rho w_\rho,\\
        (K^\star u)_\mu&=-2\nabla^\nu u_{\nu\mu},\\
        (Pu)_{\mu\nu}
        &=-\nabla^\rho\nabla_\rho u_{\mu\nu}
          +2\mathbf{R}_\mu{}^\rho{}_{\nu}{}^\sigma u_{\rho\sigma}
          +\nabla_\mu\nabla^\rho u_{\rho\nu}
          +\nabla_\nu\nabla^\rho u_{\rho\mu}
          -\rg_{\mu\nu}\nabla^\rho\nabla^\sigma u_{\rho\sigma}.
\end{aligned}
\]
Since by assumption $\Ric=0$, we have the following identities:
\[
\begin{aligned}
D_{2}&=P+K\circ K^{\star},\quad D_{1}=K^{\star}\circ K,\\
P&=P^{\star}\neq P^{*},\quad P\circ K=0,\\
D_{1}&=D_{1}^{*},\quad D_{2}=D_{2}^{*}=D_{2}^{\star},\\
D_{2}\circ K&=K\circ D_{1},\quad K^{\star}\circ D_{2}=D_{1}\circ K^{\star},\quad [D_{2}, I]=0.
\end{aligned}
\]
The condition $K^\star u=0$ is the linearized \emph{de Donder} or \emph{harmonic gauge}. 

\subsection{Phase spaces}\label{sec1.6}
If $A$ is a differential operator acting between sections of some vector bundles,  $\Ran_{\c}A$,  $\Ker_{\c}A$, resp.~$\Ran_{\sc}A$, $\Ker_{\sc}A$, denote the range, kernel of $A$ restricted to  sections with compact, resp.~space-compact supports.
\subsubsection{On-shell phase spaces}
The natural  phase space for linearized Einstein equations, corresponding to `solutions modulo gauge' is the quotient space
\[
\dfrac{\Ker_{\sc}P}{\Ran_{\sc}K}.
\]
Let
\[
\bar{u}\dual \rq_{P\,\sol}u\defeq (u| [\one_{J_{+}(\Sigma)}, \i P]u)_{I, V_{2}(M)},\ \ u\in \Ker_{\sc}P.
\]
where $\Sigma$ is a Cauchy surface in $(M, \rg)$. It can be shown, see e.g.~\cite[Sect.~5]{HS} that $\rq_{P\,\sol}$ is Hermitian, independent on the choice of $\Sigma$ and  that  $[\rq_{P\,\sol}]$ is well defined on $\dfrac{\Ker_{\sc}P}{\Ran_{\sc}K}$.

The corresponding phase space of `gauge-fixed' solutions is
\beq\label{eq:gf}
\dfrac{\Ker_{\sc}D_{2}\cap \Ker_{\sc}K^{\star}}{K\Ker_{\sc}D_{1}}.
\eeq
Let
\[
\bar{u}\dual \rq_{I, 2\, \sol}u\defeq (u| [\one_{J_{+}(\Sigma)}, \i D_{2}]u)_{I, V_{2}(M)},\ \ u\in \Ker_{\sc}D_{2}.
\]
$\rq_{I, 2\,\sol}$ is Hermitian, independent on the choice of $\Sigma$ and   $[\rq_{I, 2\,\sol}]$ is well defined on \eqref{eq:gf}. One can show, see e.g.~\cite[Subsect.~2.4]{GW1} that if $\iota:\Ker_{\sc}D_{2}\cap \Ker_{\sc}K^{\star}\to \Ker_{\sc}P$ is the canonical injection then
\[
[\iota]: \Big(\dfrac{\Ker_{\sc}D_{2}\cap \Ker_{\sc}K^{\star}}{K\Ker_{\sc}D_{1}}, [\rq_{I, 2 \,\sol}]\Big)\to \Big(\dfrac{\Ker_{\sc}P}{\Ran_{\sc}K}, [\rq_{P\,\sol}]\Big)
\]
is  well defined and  an isomorphism of Hermitian spaces.

\subsubsection{Off-shell phase spaces}
One can equivalently work with off-shell phase spaces consisting of compactly supported sections. In fact, see e.g.~\cite{HS}, one can show that if $G_{2}$ is the causal propagator for $D_{2}$ then
\[
\bar{u}\dual {\bf Q}_{I,2}u\defeq \i(u| G_{2}u)_{I, V_{2}(M)}, 
\]
is Hermitian on $\coinf(M; V_{2})$ and that $[{\bf Q}_{I, 2}]$ is well defined on $\dfrac{\Ker_{\c}K^{\star}}{\Ran_{\c}P}$. 

Moreover the maps
\[
 \Big(\dfrac{\Ker_{\c} K^{\star}}{\Ran_{\c} P}, [{\bf Q}_{I,2}]\Big)\overset{[G_{2}]}{\to} \Big(\dfrac{\Ker_{\sc} P}{\Ran_{\sc} K}, [\rq_{P\,\sol}]\Big),
 \]
\[
\Big(\dfrac{\Ker_{\c} K^{\star}}{\Ran_{\c} P}, [{\bf Q}_{I,2}]\Big)\overset{[G_{2}]}{\to}\Big(\dfrac{\Ker_{\sc} D_{2}\cap \Ker_{\sc}K^{\star}}{K\Ker_{\sc}D_{1}}, [\rq_{I, 2\,\sol}]\Big)
 \]
are isomorphisms of Hermitian spaces.

\subsubsection{Cauchy data phase spaces}\label{sec1.6.2}
Let $\Sigma\subset M$ a smooth space-like Cauchy surface.  
We set
\beq\label{e00.3a}
\varrho_{k}u= \col{u\traa{\Sigma}}{\i^{-1}\nabla_{\nu}u\traa{\Sigma}}= \col{f_{0}}{f_{1}},\ \ u\in \cinf_{\sc}(M; V_{k}),
\eeq
where $\nu$ is the future directed unit normal to $\Sigma$. 

Since $D_{k}$ are hyperbolic operators, the Cauchy problem for $D_{k}$ is well posed, and we define $\cU_{k}: \coinf(\Sigma; V_{k}\otimes \cc^{2})\to \Ker_{\sc}D_{k}$ by
\beq\label{e00.3}
\begin{cases}
D_{k}\cU_{k}= 0, \\
\varrho_{k}\cU_{k}= \one.
\end{cases}
\eeq
We associate to $K, K^{\star}$ and $I$ the corresponding operators acting on Cauchy data by
\[
K_{\Sigma}\defeq \varrho_{2}K \cU_{1},\quad K_{\Sigma}^{\dag}\defeq \varrho_{1}K^{\star}\cU_{2},\quad I_{\Sigma}\defeq \varrho_{2}I \cU_{2}.
\]
One defines Hermitian forms $\rq_{k}$ on $\coinf(\Sigma; V_{k}\otimes\cc^{2})$ by
\beq\label{lub}
(u_{k}| D_{k}v_{k})_{V_{k}(J_{\pm}(\Sigma))}- (D_{k}u_{k}| v_{k})_{V_{k}(J_{\pm}(\Sigma))}
\eqdef  \pm \i^{-1}\overline{\varrho_{k}u_{k}}\dual \rq_{k}\varrho_{k}v_{k}
\eeq
for all $u_{k}, v_{k}\in \coinf(M; V_{k})$. We set
\[
\rq_{I,1}= \rq_{1},  \quad \rq_{I,2}= \rq_{2}\circ I_{\Sigma},
\]
which corresponds to using the physical Hermitian forms in the lhs of \eqref{lub}. 
One can show, see \cite[Lem.~2.9]{GW1}, that $K_{\Sigma}^{\dag}\circ K_{\Sigma}=0$, and that $K_{\Sigma}^{\dag}$ is the adjoint of $K_{\Sigma}$ for $\rq_{I, k}$, i.e.~
\[
\bar{f}_{2}\dual \rq_{I, 2}K_{\Sigma}f_{1}=\bar{K_{\Sigma}^{\dag}f_{2}}\dual \rq_{I,1}f_{1},\ \ f_{k}\in \coinf(\Sigma; V_{k}\otimes \cc^{2}). 
\]
The crucial fact for us, see  \cite{GW1}, is  that the maps
\[
 \Big(\dfrac{\Ker_{\c} K^{\star}}{\Ran_{\c} P}, [{\bf Q}_{I,2}]\Big)\overset{[\varrho_{2}\circ G_{2}]}{\to} \Big(\dfrac{\Ker_{\c} K_{\Sigma}^{\dag}}{\Ran_{\c} K_{\Sigma}}, [\rq_{I,2}]\Big),
 \]
\[
\Big(\dfrac{\Ker_{\sc} D_{2}\cap \Ker_{\sc}K^{\star}}{K\Ker_{\sc}D_{1}}, [\rq_{I, 2\,\sol}]\Big)
\overset{[\varrho_{2}]}{\to}\Big(\dfrac{\Ker_{\c} K_{\Sigma}^{\dag}}{\Ran_{\c} K_{\Sigma}}, [\rq_{I,2}]\Big). 
\]
are isomorphisms of Hermitian spaces.

\subsection{Quantization}\label{sec1.7}

 The algebraic quantization of linear gauge theories is discussed in detail in \cite[Sect.~3]{GW1}. The algebraic framework reduces the quantization problem to showing the existence of physically relevant quantum states on the  CCR $*$-algebra $\CCR(\cV_{P},{\bf Q}_{P})$ where we set
 \[
 (\cV_{P},{\bf Q}_{P})= \Big(\dfrac{\Ker_{\c} K^{\star}}{\Ran_{\c} P}, [{\bf Q}_{I, 2}]\Big).
 \]
    The notions of CCR $*$-algebras, quasi-free states and covariances  are explained for example in \cite[Chap.~4]{G1}.

\subsubsection{Covariances}\label{sec1.7.1}
A quasi-free state on $\CCR(\cV_{P},{\bf Q}_{P})$ is determined by a pair $\boldsymbol{\Lambda}_{P}^{\pm}$ of {\em covariances}, i.e.~of Hermitian forms on $\cV_{P}$  such that \[
\begin{aligned}
i)&\quad \boldsymbol{\Lambda}^{\pm}_{P}= \boldsymbol{\Lambda}_{P}^{\pm*},\quad \boldsymbol{\Lambda}_{P}^{\pm}\geq 0,\\
ii)&\quad 
\boldsymbol{\Lambda}_{P}^{+}- \boldsymbol{\Lambda}_{P}^{-}={\bf Q}_{P}.
\end{aligned}
\]
We will consider quasi-free states $ \omega$ on  $\CCR(\cV_{P},{\bf Q}_{P})$ with covariances obtained  from a  pair of continuous Hermitian forms $\boldsymbol{\Lambda}_{2}^{\pm}$ on $\coinf(M; V_{2})$ (called  the {\em spacetime covariances} of $\omega$) by:
 \beq\label{e00.1}
 [\bar{u}]\dual \boldsymbol{\Lambda}_{P}^{\pm}[u]= \bar{u}\dual \boldsymbol{\Lambda}_{2}^{\pm}u,\ \ [u]\in \frac{\Ker_{\c} K^\star}{\Ran_{\c} P}.
 \eeq
\begin{lemma}\label{lemma.had}
 Suppose  that $\boldsymbol{\Lambda}_{2}^{\pm}\in L(  \coinf(M; V_{2}), \coinf(M; V_{2})^{*})$ are such that:
  \beq\label{defodefo}
\begin{aligned}
i)&\quad  D_{2}^{*} \circ \boldsymbol{\Lambda}^\pm_{2}=\boldsymbol{\Lambda}^\pm_{2}  \circ D_{2} =0 , \\ 
ii)&\quad \boldsymbol{\Lambda}_{2}^{+}-\boldsymbol{\Lambda}^{-}_{2} = {\bf Q}_{I,2} \hbox{ on }\Ker_{\c} K^{\star},\\
iii)&\quad \boldsymbol{\Lambda}_{2}^{\pm}=0\hbox{ on }\Ker_{\c} K^{\star}\times \Ran_{\c}K,\\
iv) & \quad \boldsymbol{\Lambda}_{2}^{\pm}= \boldsymbol{\Lambda}_{2}^{\pm*},\quad \boldsymbol{\Lambda}^\pm_{2}  \geq 0  \hbox{ on } \Ker_{\c} K^{\star}.
\end{aligned}
\eeq
Then $\boldsymbol{\Lambda}_{2}^{\pm}$ are the covariances  of a quasi-free state on $\CCR(\cV_{P},{\bf Q}_{P})$.
\end{lemma}
\begin{remark}\label{rem1.0}
 Condition \eqref{defodefo} {\it iii)} is the {\em gauge invariance condition}. It turns out, see \cite[Subsect.~2.7]{GW3} that it can be replaced by the weaker condition
\[
iii')\quad \boldsymbol{\Lambda}_{2}^{\pm}=0\hbox{ on }\Ker_{\c} K^{\star}\times \Ran_{\c}KK^{\star},
\]
although this fact  will not be needed in this paper. Condition \eqref{defodefo} {\it iv)} is the {\em positivity condition}, implying the positivity of the state on  
$\CCR(\cV_{P},{\bf Q}_{P})$.

\end{remark}
\subsection{Hadamard condition} \label{sec1.7b}
The general consensus is that the  \emph{Hadamard condition}  singles out the physically meaningful states.   We recall that
 $$
 \cN=\{(x,\ \xi)\in T^*M\setminus\zero\,:\, \xi\cdot \rg^{-1}(x)\xi=0\}
 $$
 is  the  characteristic set of the wave operator on $(M,\rg)$, and 
$$
\cN^{\pm}=\cN\cap \{(x,\ \xi)\in T^*M\setminus\zero\, :\, \pm v\dual \xi>0\,\ \forall v\in T_xM\hbox{ future-directed time-like}\}
$$
are its  two connected components, corresponding to the upper/lower energy shells.

To formulate the Hadamard condition, we identify  the Hermitian forms $\boldsymbol{\Lambda}_{2}^{\pm}$  with distributional kernels  $\Lambda_{2}^{\pm}(\cdot, \cdot)\in \cD'(M\times M; L(V_{2}))$, called  {\em two-point functions}.  

This identification is   defined by the formal identity
\[
\bar{u}\dual \boldsymbol{\Lambda}^{\pm}_{2}u\eqdef\int_{M\times M}(u(x)| \Lambda_{2}^{\pm}(x, y)u(y))_{V_{2}}\dVol_{\rg}(x)\dVol_{\rg}(y),\ \ u\in \coinf(M; V_{2}).
\]
One can of course use other Hermitian forms on  the fibers of $V_{2}$ to do this identification, like for example $(\cdot| \cdot)_{I, V_{2}}$. This change amounts  to composing  $\Lambda_{2}^{\pm}(x, y)$ with smooth linear operators acting on the fibers of $V_{2}$ over $x$ and $y$ and does not change the Hadamard condition \eqref{defidefi} below.

\begin{definition}\label{defhadama}
  A quasi-free state $\omega$ on $\CCR(\cV_{P},{\bf Q}_{P})$ given by  covariances $\boldsymbol{\Lambda}_{2}^{\pm}$ as in Lemma~\ref{lemma.had} is {\em Hadamard} if in addition to \eqref{defodefo} it satisfies:
\beq\label{defidefi}
 \WF(\Lambda_{2}^{\pm})'\subset \cN^{\pm}\times \cN^{\pm}.
 \eeq
\end{definition}

\subsubsection{Hadamard condition on a Cauchy surface}\label{sec1.7.3}
One can equivalently consider Hermitian forms  $\boldsymbol{\lambda}^{\pm}_{2}$  on  the space of Cauchy data $\coinf(\Sigma; V_{2}\otimes \cc^{2})$ called {\em Cauchy surface covariances}.
Namely   assume that we have a pair of Hermitian forms $\boldsymbol{\lambda}_{2}^{\pm}$ on $\coinf(\Sigma; V_{2}\otimes \cc^{2})$ and  set
 \beq\label{e00.2b}
\boldsymbol{\Lambda}_{2}^{\pm}= (\varrho_{2}G_{2})^{*}\boldsymbol{\lambda}_{2}^{\pm}(\varrho_{2}G_{2}).
 \eeq
 The conditions on $\boldsymbol{\lambda}_{2}^{\pm}$ corresponding to \eqref{defodefo} are
 \beq
 \label{defodefo-cauchy}
 \begin{aligned}
i)&\quad \boldsymbol{\lambda}_{2}^{+}-\boldsymbol{\lambda}_{2}^{-} = \rq_{I,2} \hbox{ on }\Ker_{\c} K_{\Sigma}^{\dag},\\
ii)&\quad \boldsymbol{\lambda}_{2}^{\pm}= 0\hbox{ on }\Ker_{\c} K_{\Sigma}^{\dag}\times \Ran_{\c}K_{\Sigma},\\
iii) & \quad \boldsymbol{\lambda}_{2}^{\pm}= \boldsymbol{\lambda}_{2}^{\pm*},\quad \boldsymbol{\lambda}^\pm_{2}  \geq 0 \mbox{ on } \Ker_{\c} K_{\Sigma}^{\dag}.
\end{aligned}
 \eeq
 
Since $\rq_{I,2}$ is non-degenerate, we can set
\beq\label{e00.4}
\boldsymbol{\lambda}^{\pm}_{2}\eqdef \pm \rq_{I,2}\circ  c_{2}^{\pm}.
\eeq
\begin{proposition}\label{prop:states1} Suppose $c_2^\pm: \coinf(\Sigma;V_{2}\otimes \cc^{2})\to \cinf(\Sigma;V_{2}\otimes \cc^{2})$ is a pair of operators such that:
\beq\label{condit-hada}
\begin{aligned}
i)&\quad c_2^++c_2^-=\one,\\
ii)&\quad c_2^\pm : \Ran_{\c}K_{\Sigma}\to \Ran K_{\Sigma},\\
iii)&\quad \rq_{I,2}\circ  c_{2}^{\pm}= c_{2}^{\pm*}\circ  \rq_{I,2},\quad \pm \rq_{I,2} \circ  c^\pm_2 \geq 0,\quad \hbox{ on } \Ker_{\c} K_{\Sigma}^\dag.
\end{aligned}
\eeq

Then $\boldsymbol{\Lambda}_{2}^{\pm}$ given by \eqref{e00.2b} and \eqref{e00.4} are the covariances of a quasi-free state on $\CCR(\cV_{P},{\bf Q}_{P})$. Furthermore if for some neighborhood $\mathcal O$ of $\Sigma$ in $M$ we have: 
\[
{\it iv)}\ \WF(\cU_{2}\circ c_{2}^{\pm})'\subset (\cN^{\pm}\cup\cF)\times T^{*}\Sigma
\] over $\mathcal O\times \Sigma$,
where $\cF\subset T^{*}M$ is a conic set with $\cF\cap \cN= \emptyset$, 
then the associated state  is Hadamard. 
\end{proposition}
The proof of \eqref{condit-hada} is analogous to the one found in \cite[Sect.~3.4]{GW1}. The proof of the statement on the Hadamard condition can be found in \cite[Sect.~11.1]{G1}. 

In practice $c_{2}^{\pm}$ are constructed as {\em projections} and therefore called {\em Hadamard projectors}.

 \subsection{Conformal transformations}\label{sec1.8}
 The operators $D_{k}$ have principal symbol $\xi\dual\rg^{-1}\xi\one_{V_{k}}$ and it will be very useful to perform a conformal change of the metric, replacing $\rg$ by $\widehat{\rg}= c^{2}\rg$.
This subsection is only an algebraic reminder of how conformal changes act on the phase space.
%Another less important transformation, which can be skipped on a first reading,  consists in changing the Hermitian scalar products on the fibers of $V_{k}$.
%To this end  we fix a bundle isomorphism  $W_{k}: \coinf(M; V_{k})\to \coinf(M; V_{k})$ (to be chosen more specifically later on) and we denote by $\widehat{V}_{k}$ the bundle $V_{k}$ equipped with the fiber scalar product $(\cdot| \cdot)_{\widehat{V}_{k}}$ defined by
%\[
%(W_{k}u_{k}| W_{k}v_{k})_{\widehat{V}_{k}}\defeq (u_{k}|v_{k})_{V_{k}}.
%\]

Setting 
\[
U_{k}u\defeq c^{-n/2}u,\quad \
 (u_{k}| u_{k})_{\widehat{V}_{k}(M)}= \int_{M}(u_{k}|u_{k})_{V_{k}}d\vol_{\widehat{\rg}},
 \]
 
 we have
 \[
 (U_{k}u_{k}| U_{k}v_{k})_{\widehat{V}_{k}(M)}= (u_{k}| v_{k})_{V_{k}(M)},\ \ u_{k},\ v_{k}\in \coinf(M; V_{k}).
 \]
 We set
 \beq\label{e2.12r}
 \widehat{D}_{k}\defeq c^{-1}U_{k} D_{k} U_{k}^{-1}c^{-1},
 \eeq
 which is hyperbolic with principal symbol $\xi\dual\widehat{\rg}^{-1}\xi\one_{V_{k}}$, and formally self-adjoint for 
 $(\cdot| \cdot)_{\widehat{V}_{k}(M)}$.  We set also:
 \beq\label{e2.12}
 \widehat{K}\defeq c U_{2}K U_{1}^{-1}c^{-1},\quad \widehat{T}\defeq c^{-1}U_{2}KU_{1}^{-1}c, 
  \eeq
 and we have:
 \beq\label{e2.12m}
  \widehat{D}_{2}\circ \widehat{K}= \widehat{T}\circ \widehat{D}_{1}.
 \eeq

Let us define the Hermitian forms   $\widehat{\rq}_{\widehat{I},k}$ as in \ref{sec1.6.2}, replacing $ D_{k},I,  V_{k}$ by $\widehat{D}_{k}, \widehat{I}, \widehat{V}_{k}$.

We  fix  analogs of the Cauchy data maps $\varrho_{k}$ for $\widehat{D}_{k}$, denote them by $\widehat{\varrho}_{k}$, and define the maps $\mathcal{T}_{k}$ on $\coinf(\Sigma; V_{k}\otimes \cc^{2})$ by:
\beq\label{e2.12p}
\widehat{\varrho}_{k}cU_{k}u_{k}\eqdef \mathcal{T}_{k}\varrho_{k}u_{k},\ \ u_{k}\in C^{\infty}_{\sc}(M; V_{k}).
\eeq
If $\widehat{\cU}_{k}: \coinf(\Sigma; \widehat{V}_{k}\otimes \cc^{2})\to \Ker \widehat{D}_{k}$ is the solution of the associated Cauchy problem, we have
\begin{equation}
\label{e2.12w}
\widehat{\cU}_{k}= (cU_{k})\circ \cU_{k}\circ \mathcal{T}_{k}^{-1},
\end{equation}
and
\begin{equation}
\label{e2.12n}
\widehat{\rq}_{\widehat{I}, 2}= (\mathcal{T}_{2}^{-1})^{*}\rq_{I, 2}\mathcal{T}_{2}^{-1}.
\end{equation}

From \eqref{e2.12m} we see that $\widehat{K}: \Ker_{\sc}\widehat{D}_{1}\to \Ker_{\sc}\widehat{D}_{2}$ and we denote by $\widehat{K}_{\Sigma}$ the Cauchy operator associated to $\widehat{K}$ and by $\widehat{K}_{\Sigma}^{\dag}$ its adjoint for $\widehat{\rq}_{\widehat{I},k}$.
One has
\begin{equation}
\label{e2.12x}
\widehat{K}_{\Sigma}= \mathcal{T}_{2}\circ K_{\Sigma}\circ \mathcal{T}_{1}^{-1},\quad \widehat{K}_{\Sigma}^{\dag}= \mathcal{T}_{1}\circ K_{\Sigma}^{\dag}\circ \mathcal{T}_{2}^{-1}
\end{equation}

Let  $\boldsymbol{\lambda}_{2}^{\pm}$ be a pair of Hermitian forms on $\coinf(\Sigma; V_{2}\otimes \cc^{2})$ and let 
\beq\label{e2.12o}
\widehat{\boldsymbol{\lambda}}^{\pm}_{2}\defeq (\mathcal{T}_{2}^{-1})^{*}\circ \boldsymbol{\lambda}_{2}^{\pm}\mathcal{T}_{2}^{-1}.
\eeq
Using \eqref{e2.12n} we have:
\[
\widehat{c}_{2}^{\pm}= \mathcal{T}_{2}c_{2}^{\pm}\mathcal{T}_{2}^{-1}.
\]

Using \eqref{e2.12x} we see that the conditions in \eqref{defodefo-cauchy} for $\boldsymbol{\lambda}_{2}^{\pm}$ and for $\widehat{\boldsymbol{\lambda}}^{\pm}_{2}$ are equivalent. 

Similarly using \eqref{e2.12w} the  conditions  {\it iv)} in Prop.~\ref{prop:states1} for $c_{2}^{\pm}$ and $\widehat{c}_{2}^{\pm}$ are equivalent.

  \section{The background metric and its reduction}\label{sec2}\init
In this section we state our hypotheses on the background metric $\rg$ and put it in a form useful later on.
 \subsection{Hypotheses on the background metric}\label{sec2.1}
We consider a spacetime $(\rr^{n}, \rg)$, where $\rg$ is a Ricci-flat perturbation of the Minkowski metric $\boldsymbol{\eta}$ on $\rr^n$. We are mainly interested in the physical case $n=4$, but it is convenient to keep $n$ arbitrary in the sequel.

Let $x=(x_{0},\rx)$ be the canonical coordinates on $\rr^{n}$ and let $\Sigma=\{x_{0}=0\}$. A \emph{conic neighborhood} of $\Sigma$ is an open set $U$ containing
$$
\{x\in \rr^{n}: |x|<\delta\}\cup\{x\in \rr^{n}: |x_{0}|<\delta |x|\}
$$
for some $\delta>0$.

\begin{hypothesis}\label{hyp:bgmetric}
There exists a conic neighborhood $U$ of $\Sigma$ such that:
\begin{enumerate}[label={\rm(H\arabic*)},leftmargin=7ex]
\item The metric is asymptotically Minkowski of order $-1$ in $U$:
\begin{equation}
\label{e2.hyp1}
\rg=\boldsymbol{\eta}+\rg_{-1}\hbox{ in }U,\ \ \rg_{-1}\in S^{-1}(\rr^{n};V_{2}).
\end{equation}
\item The hypersurface $\Sigma$ is a spacelike Cauchy surface in $(\rr^{n},\rg)$.
\end{enumerate}
\end{hypothesis}

\begin{remark}\label{rem:bgmetric-existence}
Einstein metrics satisfying Hypothesis~\ref{hyp:bgmetric} do exist. In particular, the results of \cite{HV} imply that for $n=4$, for any $\delta>0$, there exist solutions of the Einstein equations which in a sufficiently small conic neighborhood of $\Sigma$ are of the form $\rg=\rg^{m}_{S}+\rg_{-1-\delta}$, where $\rg^{m}_{S}$ is, outside a compact region, the Schwarzschild metric of mass $m$. In a conic neighborhood of $\Sigma$, $\rg^{m}_{S}$ differs from the Minkowski metric by terms in $S^{-1}(\rr^{n};V_{2})$, while $\rg_{-1-\delta}\in S^{-1-\delta}(\rr^{n};V_{2})$. These arise from Einstein initial data on $\Sigma$ differing from the Schwarzschild data by a symbolic perturbation which decays faster than $1/r$.

Concretely, \cite[Thm.~1.8]{HV} states that if initial data on $\Sigma=\rr^3$ are given such that, compared to $\rg_S^m$, the induced metric is perturbed by an error in $H_{\mathrm{b}}^{\infty,1+\delta}(\overline{\rr^3})$ and the second fundamental form by an error in $H_{\mathrm{b}}^{\infty,2+\delta}(\overline{\rr^3})$, then in a conic neighborhood of $\Sigma$ the solution $\rg$ differs from $\rg_S^m$ by an error in $H_{\mathrm{b}}^{\infty,1+\delta}({}^mM)$, where ${}^mM$ is a particular compactification of $\rr^4$. We now explain why this is compatible with (H1)--(H2) of Hypothesis~\ref{hyp:bgmetric}.

First, since we work in an arbitrarily small conic neighborhood of $\Sigma$, we use the radial compactification associated with the spacelike radial variable
$$
r=(|\rx|^2-x_0^2)^{\frac12},\quad
\rho=r^{-1},\quad
\frac{x_0}{|\rx|}=\tanh t,
\quad \omega=\frac{\rx}{|\rx|}.
$$
Equivalently, $x_0=r\sinh t$ and $\rx=r\cosh t\,\omega$. This differs from the boundary coordinate used in some statements of \cite{HV}, where one works with $x_0/|\rx|$ and with a logarithmically modified time coordinate adapted to the mass parameter. In the conic region considered here these changes are conormal changes of boundary coordinates. For example, if $\tau=x_0/|\rx|=\tanh t$, then $d\tau=\cosh^{-2}(t)dt$, and replacing $\tau$ by $\tau+2m\rho\ln\rho$ only changes $\b$-vector fields and scattering one-forms by conormal lower order terms. Thus the regularity information from \cite{HV} is compatible with the symbolic assumptions used below, after the same harmless arbitrarily small loss in decay order as in the original formulation.

Second, since \cite{HV} solves the Einstein equation for the unknown dynamical metric in $L^2$-based spaces, results on regularity of the metric and assumptions on regularity of initial data are stated in terms of $H_{\mathrm{b}}^{\infty}({}^mM)$, rather than the $L^{\infty}$-based $\mathcal{A}^{\infty}(M)$ spaces we use here. Since $H_{\mathrm{b}}^{\infty,\epsilon}(M)\subset \mathcal{A}^{\infty,\epsilon}(M)\subset H_{\mathrm{b}}^{\infty}(M)$ for any $\epsilon>0$, conversion between these spaces leads to an arbitrarily small loss in the decay order of the metric perturbation. The perturbation therefore still decays strictly faster than the Schwarzschild leading term when expressed in our symbol spaces.
\end{remark}

\begin{remark}Our construction also applies under the weaker assumption that $\rg_{-1}\in S^{-\epsilon}(\rr^{n};V_{2})$ for some $0<\epsilon<1$. The required modification is to allow remainders of order $r^{-\epsilon}$ rather than $r^{-1}$ in the coefficient classes and the extended $\b$-pseudodifferential calculus (see Subsect.~\ref{sec9.4b}). The normal operators discussed in later sections are then unchanged, and the remainders still improve decay by a positive amount, which suffices for the compactness arguments in our proofs.\end{remark}
 
\subsection{Reduction of the background metric}\label{sec2.2}

\subsubsection{Beig--Schmidt coordinates near spacelike infinity}
In a conic neighborhood of $\Sigma$ we use the so-called Beig--Schmidt coordinates:
$$
        x_0=r\sinh t,\quad
        \rx=r\cosh t\,\omega,
        \quad \omega\in \bS^{n-2},
        \quad \rho=r^{-1}.
$$
Thus
$$
        r=(|\rx|^2-x_0^2)^{\frac12},
        \quad \frac{x_0}{|\rx|}=\tanh t,
        \quad \omega=\frac{\rx}{|\rx|}.
$$
Note that the variable $r$ is the Minkowski radial variable and not the spatial radius $|\rx|$ (except on $\Sigma$, where $r=|\rx|$). We denote by $\rh_{\bS^{d}}$ the canonical metric on $\bS^{d}$.

In these coordinates the Minkowski metric reads:
\begin{equation}
\label{e2.mink-hyp}
\begin{aligned}
\boldsymbol{\eta}
  &= -r^2dt^2+dr^2+r^2\cosh^2(t)\rh_{\bS^{n-2}} \\
  &= -\left(\frac{dt}{\rho}\right)^2+\boldsymbol{\delta}(t),
\end{aligned}
\end{equation}
where
$$
        \boldsymbol{\delta}(t)
        =dr^2+r^2\cosh^2(t)\rh_{\bS^{n-2}}
        =\rho^{-4}d\rho^2+\rho^{-2}\cosh^2(t)\rh_{\bS^{n-2}}.
$$
Equivalently,
\begin{equation}
\label{e2.mink-b-hyp}
        \boldsymbol{\eta}_{\b}:=\rho^2\boldsymbol{\eta}
        =-dt^2+\boldsymbol{\delta}_{\b}(t),
        \quad
        \boldsymbol{\delta}_{\b}(t)
        =\left(\frac{d\rho}{\rho}\right)^2+
        \cosh^2(t)\rh_{\bS^{n-2}}.
\end{equation}

We replace $t$ by $\chi(r)t + (1-\chi(r))x_0$, where $\chi\in C^{\infty}(\mathbb{R})$ with $\chi=0$ near 0 and $\chi(r)=1$ for $r>R$ for some $R$; the cutoff and its product with the exterior coordinate $t$ are extended by zero outside that chart. After this modification, the regions $\{|t|<\delta\}$ are conic neighborhoods of $\Sigma$ and the level sets of $t$ for small enough $|t|$ are spacelike hypersurfaces for $\mathbf{g}$. The expressions for the Minkowski metric above continue to hold in $\{r>R\}$.

\subsubsection{Compactification}
\label{sec:compactification}
We fix a conic neighborhood $U_{\delta}=\{|t|<\delta\}$ of $\Sigma$. We denote by $M$ the manifold obtained by smoothly gluing $U_{\delta}$ with
$$
        M_{\infty}=\{(t,\ \rho,\ \omega):0\leq \rho<1/2,
        |t|<\delta,\quad \omega\in\bS^{n-2}\},
$$
using the coordinate changes above on the overlap. We smoothly extend $\rho$ to a globally defined boundary defining function on $M$.

The closure of $\Sigma$, and similarly of any level set of $t$, is identified with the radial compactification of $\rr^{n-1}$. On the boundary of these level sets the $\b$-frames are
$$
        \rho\p_{\rho},\ \p_{\omega_i},
        \quad
        \frac{d\rho}{\rho},\ d\omega_i.
$$
The hypothesis $\rg-\boldsymbol{\eta}\in S^{-1}$ is equivalent to:
\begin{equation}
\label{e2.metric-b-assumption}
\begin{aligned}
        \rho^2\rg
        &=\boldsymbol{\eta}_{\b}
          +\rho\cA^{\infty}(M;\otimes^2_{\s}\be T^*M),\\
        \rho^{-2}\rg^{-1}
        &=\boldsymbol{\eta}_{\b}^{-1}
          +\rho\cA^{\infty}(M;\otimes^2_{\s}\be TM).
\end{aligned}
\end{equation}
\subsubsection{Conformal reduction}\label{sec:blub}
We now reduce $\rg$  to a split form. 
\begin{proposition}
\label{prop2.1}
There exist $\varepsilon>0$, a diffeomorphism
$$
        \chi:U_{\varepsilon}\to U_{\varepsilon},
        \quad \chi(0,\cdot)=\one_{\Sigma},
        \quad t\circ\chi=t,
$$
and a conformal factor $c\in\cA^{\infty}(M)$ such that
\begin{equation}
\label{eq:metric-reduced}
        \chi^*\rg
        =c^{-2}(-r^2dt^2+\rh(t)),
       \end{equation}
where $\rh(t)$ is a family of Riemannian metrics on the level sets of $t$. Moreover,
\begin{equation}
\label{e2.metric-reduced-asymptotics}
\begin{aligned}
        c&=1+\rho\cA^{\infty}(M),\\
        \rh_{\b}(t):=\rho^2\rh(t)
          &=\boldsymbol{\delta}_{\b}(t)
            +\rho\cA^{\infty}(M;\otimes_{\s}^{2}{}^{\b}T^{*}\Sigma),\\
        \rh_{\b}(t)^{-1}
          &=\boldsymbol{\delta}_{\b}(t)^{-1}
            +\rho\cA^{\infty}(M;\otimes_{\s}^{2}{}^{\b}T\Sigma).
\end{aligned}
\end{equation}
\end{proposition}

\proof
Set
$$
        \kappa:=dt\dual\rg^{-1}dt.
$$
For the free metric $\boldsymbol{\eta}$ one has $dt\dual\boldsymbol{\eta}^{-1}dt=-\rho^2$ in $\{r>R\}$. Hence, by \eqref{e2.metric-b-assumption},
\begin{equation}
\label{e2.kappa}
        \kappa=-\rho^2(1+\rho a),
        \ \ a\in\cA^{\infty}(M).
\end{equation}
After possibly shrinking the conic neighborhood, $\kappa<0$.

Define
$$
        V:=\frac{\rg^{-1}dt}{dt\dual\rg^{-1}dt}.
$$
Then $dt(V)=1$. For $\boldsymbol{\eta}$ this vector field is $\partial_t$ in $\{r>R\}$, and \eqref{e2.metric-b-assumption} gives
\begin{equation}
\label{e2.V-b}
        V=\partial_t+v,
        \ \ v\in\rho\cA^{\infty}(M;{}^{\b}T\Sigma),
\end{equation}
where $v$ is tangent to the level sets of $t$.

Let $\widetilde\varphi_t$ be the time-dependent flow on $\Sigma$ generated by $v$, normalized by $\widetilde\varphi_0=\one_{\Sigma}$, and define
$$
        \chi(t,y)=(t,\widetilde\varphi_t(y)).
$$
Then $\chi_*\partial_t=V$. If $Y$ is tangent to a level set of $t$, then
$$
        \rg(V,Y)=\frac{dt(Y)}{\kappa}=0.
$$
Thus $\chi^*\rg$ has no $dt\,dy$ cross terms. Since
$$
        \chi^*\rg(\partial_t,\partial_t)=\rg(V,V)=\kappa^{-1},
$$
we choose
\begin{equation}
\label{e2.conformal-factor-new}
        c^2:=-\rho^{-2}\chi^*\kappa.
\end{equation}
This gives $c^2\chi^*\rg(\partial_t,\partial_t)=-\rho^{-2}$ and proves \eqref{eq:metric-reduced}.

The asymptotics are obtained by applying Corollary~\ref{cor:ODE-small} in the appendix to \eqref{e2.V-b}. In local product coordinates near $\rho=0$,
\begin{equation}
\label{e2.flow-error-new}
        \chi^*\rho=\rho+\rho^2\cA^{\infty}(M),
        \quad
        \chi^*\omega_i=
        \omega_i+\rho\cA^{\infty}(M).
\end{equation}
Combining \eqref{e2.kappa} and \eqref{e2.flow-error-new} gives $c=1+\rho\cA^{\infty}(M)$. Writing
$$
        (x_1,\ldots,x_n)=(t,\rho,\omega_1,
        \ldots,\omega_{n-2}),
        \quad
        (\tilde x_1,\ldots,\tilde x_n)=
        (t,\tilde\rho,\tilde\omega_1,
        \ldots,\tilde\omega_{n-2}),
$$
one has
$$
\chi^*\rg-\boldsymbol{\eta}
=
\sum_{i,j,k,l}
\left(
\rg_{kl}(\tilde x)
\frac{\partial\tilde x_k}{\partial x_i}
\frac{\partial\tilde x_l}{\partial x_j}
-
\boldsymbol{\eta}_{ij}(x)
\right)dx_i dx_j.
$$
Using \eqref{e2.flow-error-new} and part (4) of Lemma~\ref{thm:ODE}, this difference has $\rho\cA^{\infty}$ coefficients in the scattering coframe
$$
        \rho^{-2}d\rho,
        \quad \rho^{-1}dt,
        \quad \rho^{-1}d\omega_i.
$$
Equivalently, $\rho^2(\chi^*\rg-\boldsymbol{\eta})\in
\rho\cA^{\infty}(M;\otimes_{\s}^{2}{}^{\b}T^{*}M)$, and the spatial estimates in \eqref{e2.metric-reduced-asymptotics} follow. \qed

From this point on, we denote $\chi^*\rg$ simply by $\rg$. 
In the sequel we take the compactification $M$ to be defined with $\delta$ equal to the $\varepsilon$ in Proposition~\ref{prop2.1}. We set 
\beq\label{defdegb}
\rg_{\b}:=c^{2}\rho^2\rg.
\eeq 
The preceding proposition gives the following form, which is the one used below.

\begin{proposition}\label{prop2.2}
We have
\beq\label{e2.conformal-factor}
        c=1+\rho\cA^{\infty}(M)
\eeq
and
\beq\label{e2.bmetric-reduced}
\begin{aligned}
\rg_{\b}&=-dt^2+\rh_{\b}(t)
        =-dt^2+\boldsymbol{\delta}_{\b}(t)
        +\rho\cA^{\infty}(M;\otimes_{\s}^{2}{}^{\b}T^{*}\Sigma),\\
\rg_{\b}^{-1}&=-\p_t^2+\rh_{\b}(t)^{-1}
        =-\p_t^2+\boldsymbol{\delta}_{\b}(t)^{-1}
        +\rho\cA^{\infty}(M;\otimes_{\s}^{2}{}^{\b}T\Sigma).
\end{aligned}
\eeq
In particular, $(\rg_{\b})_{tt}=-1$ and the mixed $t\rho$, $t\omega$ coefficients vanish.
\end{proposition}

\subsection{Decomposition of tensors}\label{sec2.3.2}
\def\Sig{\Sigma}

$M$ can be identified with $I_t\times\Sigma$, where $I=\open{-\delta,\delta}$. Correspondingly, we set 
$V_{k\Sigma}= \cc\otimes_{\rm s}^{k}T^{*}\Sigma$, with $V_{0\Sigma}=\Sigma\times\cc$, and identify
%\beq\label{etiti.0}
% \begin{array}{l}
 $$
 \cinf(M; V_{1})\tosim\textstyle \bigoplus_{j=0}^{1} \cinf(I; \cinf(\Sigma; V_{j\Sigma}))
$$ 
 by
 $$
  w\mapsto (w_{t}, w_{\Sig}), \quad
   w \eqdef w_{t}rdt+ w_{\Sig},
$$   
and
$$   
   \cinf(M; V_{2})\tosim \textstyle \bigoplus_{j=0}^{2} \cinf(I; \cinf(\Sigma; V_{j\Sigma}))
$$
by
$$
 \begin{aligned} u\mapsto &(u_{tt}, u_{t\Sig},u_{\Sig\Sig}), \\
u\eqdef & \, u_{tt}rdt\otimes rdt+ u_{t\Sig}\otimes rdt+ rdt\otimes u_{t\Sig}+ u_{\Sig\Sig}.
\end{aligned}
$$
 
 These tensor decompositions are compatible with the definition of Sobolev spaces in \ref{sec9.2.3}. For example 
 $\| w\|_{H^{m, r}_{\b}(M; V_{1})}$ is equivalent to $\|w_{t}\|_{H^{m, r}_{\b}(M; V_{0\Sigma})}+ \| w_{\Sigma}\|_{H^{m, r}_{\b}(M; V_{1\Sigma})}$ 
 and similarly for $\| u\|_{H^{m, r}_{\b}(M; V_{2})}$.

\subsection{Hermitian  structures}\label{sec2.3}
In this subsection we fix  some notation for Hermitian    structures that we will use in the sequel. 

\def\hV{\mathsf{V}}
\subsubsection{Hermitian structures}\label{sec2.3.1}
On $\rr^{n}$ we will  use the density  $\dVol_{\rg}$. On the fibers of $V_{k}$ we will use the Hermitian form   $(\cdot| \cdot)_{V^{\rg}_{k}}$ induced by $\rg$,  see \eqref{eq:uvk}. The Hermitian form $(\cdot| \cdot)_{V_{k}(\rr^{n})}$ is  then defined as in \eqref{uvkU}.

On   the radial compactification $M$ of $\rr^{n}$ we will use the $\b$-density  $\dVol_{\hrg_{\b}}$ and   on the fibers of $V_{k}\overset{\pi}{\to}M $ the Hermitian form $(\cdot| \cdot)_{V^{\rg}_{k}}$. We set:
\beq\label{e2.10bd}
(u| u)_{V_{k}(M)}= \int_{M}(u| u)_{V^{\rg}_{k}}\dVol_{\hrg_{\b}},\ \ u\in \coinf(M; V_{k}).
\eeq

Under the identifications in Sect.~\ref{sec2.3.2}, the Hermitian forms $(\cdot| \cdot)_{V^{\rg}_{k}}$ become
 \beq\label{e2.10bb}
 \begin{aligned}
(w| w)_{V^{\rg}_{1}}&= c^{2}(- | w_{t}|^{2}+ (w_{\Sig}| w_{\Sig})_{V^{\rh(t)}_{1\Sigma}}),\\
(u|u)_{V^{\hrg}_{2}}&= c^{4}(2|u_{tt}|^{2}-  4(u_{t\Sig}| u_{t\Sig})_{V^{\rh(t)}_{1\Sigma}} +  (u_{\Sig\Sig}| u_{\Sig\Sig})_{V^{\rh(t)}_{2\Sigma}}),
\end{aligned}
 \eeq
 where $(\cdot| \cdot)_{V^{\rh(t)}_{k\Sigma}}$ are the natural analogs of $(\cdot| \cdot)_{V^{\rg}_{k}}$ with $\rg$ replaced by $\rh(t)$.

We denote by $(\cdot| \cdot)_{V^{\rg(0)}_{k}}$ the Hermitian form induced on the fibers of $V_{k}$ by  the metric \[
\rg(0)= c^{-2}(0)(-r^{2}dt^{2}+ \rh(0)), 
\]
obtained by freezing the coefficients of $\rg$  at $t=0$, and set \beq\label{e2.10cc}
(u| u)_{V_{k}(\Sigma)}= \int_{\Sigma}(u| u)_{V^{\rg(0)}_{k}}\dVol_{\rh_{\b}(0)},\ \ u\in \coinf(\Sigma; V_{k}).
\eeq

%For $t\in I= ]-\delta, \delta[$ we denote by $(\cdot| \cdot)_{V^{\hrg(t)}_{k}}$ the Hermitian form induced on the fibers of $V_{k}$ by $\hrg(t)$, where now $t$ is considered as a fixed parameter. We then set:\beq\label{e2.10cc}
%(u| u)_{V_{kt}(\Sigma)}= \int_{\Sigma}(u| u)_{V^{\hrg(t)}_{k}}\dVol_{\rh_{\b}(t)}, \ u\in \coinf(\Sigma; V_{k}).
%\eeq
%so that
%\beq\label{e2.10zz}
%(u| u)_{V_{k}(I\times \Sigma)}= \int_{I}(u| u)_{V_{kt}(\Sigma)}dt, \ u\in\coinf(M; V_{k}).
%\eeq
\subsubsection{Normal Hermitian structures}\label{sec2.3.1c}
Let  $M_{\norm}= \rr^{+}_{\rho}\times I_{t}\times \bS^{n-2}$, $\Sigma_{\norm}= \rr^{+}_{\rho}\times \bS^{n-2}$. We set similarly
\beq\label{e2.11}
(u| u)_{V_{k}(M_{\norm})}= \int_{M_{\norm}}(u(t)| u(t))_{V^{\reta}_{k}}\dVol_{\reta_{\b}}, 
\eeq
and
\beq\label{e2.11a}
(u| u)_{V_{k}(\Sigma_{\norm})}= \int_{\Sigma_{\norm}}(u| u)_{V^{\reta(0)}_{kt}}\dVol_{\boldsymbol{\delta}_{\b}(0)}, 
\eeq
%\beq\label{e2.11a}
%(u| u)_{V_{k t}(\Sigma_{\norm})}= \int_{\Sigma_{\norm}}(u| u)_{V^{\reta(t)}_{kt}}\dVol_{\boldsymbol{\delta}_{\b}(t)}, 
%\eeq
which correspond to the Hermitian structures in \ref{sec2.3.1} with $M$, $\Sigma$, $\hrg$, $\rh_{\b}$ replaced by $M_{\norm}$, $\Sigma_{\norm}$, $\reta$,$\boldsymbol{\delta}_{\b}$ respectively.

\subsection{Hilbertian structures}\label{sec2.3.1b}
We will also need to fix a reference {\em Hilbertian structure} on the fibers of $V_{k}$ over $\Sigma$ resp.~$\Sigma_{\norm}$ identified with $M\cap \{t=0\}$ resp.~$M_{\norm}\cap \{t=0\}$.

\subsubsection{Hilbertian structures}\label{sec2.3.1d1}
We denote by $(\cdot| \cdot)_{\tV_{k}^{\widetilde{\rg}(0)}}$ the Hibertian structure on the fibers of $V_{k}$ obtained from
\[
\widetilde{\rg}(0)= c^{-2}(0)(r^{2}dt^{2}+ \rh(0)).
\]
 The corresponding Hilbertian scalar product is:
 \[
(u| u)_{\tV_{k}(\Sigma)}= \int_{\Sigma}(u| u)_{\tV_{k}^{\widetilde{\rg}(0)}}\dVol_{\rh_{\b}(0)},\ \ u\in \coinf(\Sigma; V_{k}).
\]
We have
\begin{equation}
\label{e2.12a}
(u| u)_{V_{k}(\Sigma)}= (u | \tau_{k}u)_{\tV_{k}(\Sigma)},
\end{equation}
for 
\begin{equation}
\label{e2.12b}
\tau_{k}= \tau^{\otimes k},\quad \tau= \mat{-1}{0}{0}{\one},
\end{equation}
where we use the tensor decomposition introduced below in Subsect.~\ref{sec2.3.2}.
%\[
%(u| u)_{\tV_{k0}(\Sigma)}= \int_{\Sigma}(u| u)_{\tV_{k}^{\widetilde{\hrg}(0)}}\dVol_{\rh_{\b}(0)}, \ u\in \coinf(\Sigma; V_{k}).
%\]
%We have
%\begin{equation}
%\label{e2.12a}
%(u| u)_{V_{k0}(\Sigma)}= (u | \tau_{k}u)_{\tV_{k0}(\Sigma)},
%\end{equation}
%for 
%\begin{equation}
%\label{e2.12b}
%\tau_{k}= \tau^{\otimes k}, \ \tau= \mat{-1}{0}{0}{\one}.
%\end{equation}

\subsubsection{Normal Hilbertian structures}\label{sec2.3.1d}
We replace as in \ref{sec2.3.1c} $\Sigma$, $\hrg$, $\rh_{\b}$ by $\Sigma_{\norm}$, $\reta$, $\boldsymbol{\delta}_{\b}$, and denote by $(\cdot| \cdot)_{\tV_{k}^{\treta(0)}}$ the Hilbertian structure on the fibers of $V_{k}$ induced by
\[
\treta(0)= r^{2}dt^{2}+ \boldsymbol{\delta}(0).
\]
We set:
\beq\label{e2.11c}
(u| u)_{\tV_{k}(\Sigma_{\norm})}= \int_{\Sigma_{\norm}}(u| u)_{\tV_{k}^{\treta(0)}}\dVol_{\boldsymbol{\delta}_{\b}(0)},\ \ u\in \coinf(\Sigma_{\norm}; V_{k}).
\eeq
Similarly to \eqref{e2.12a} we have
\begin{equation}
\label{e2.12d}
(u| u)_{V_{k}(\Sigma_{\norm})}= (u | \tau_{k}u)_{\tV_{k}(\Sigma_{\norm})}.
\end{equation}
%We set:
%\beq\label{e2.11c}
%(u| u)_{\tV_{k 0}(\Sigma_{\norm})}= \int_{\Sigma_{\norm}}(u| u)_{\tV_{k}^{\treta(0)}}\dVol_{\boldsymbol{\delta}_{\b}(0)}, \ u\in \coinf(\Sigma_{\norm}; V_{k}).
%\eeq
%Similarly to \eqref{e2.12a} we have
%\begin{equation}
%\label{e2.12d}
%(u| u)_{V_{k0}(\Sigma_{\norm})}= (u | \tau_{k}u)_{\tV_{k0}(\Sigma_{\norm})}.
%\end{equation}

\subsection{Adjoints}\label{sec2.3.1e}
We will denote by $A^{\star}$ the formal adjoint of $A$ for the Hermitian structures $(\cdot|\cdot)_{V_{k}(M)}$, $(\cdot|\cdot)_{V_{k}(\Sigma)}$, $(\cdot|\cdot)_{V_{k}(M_{\norm})}$, $(\cdot|\cdot)_{V_{k}(\Sigma_{\norm})}$.
We will similarly denote by $A^{*}$ the formal adjoint of $A$ for the Hilbertian structures $(\cdot|\cdot)_{\tV_{k}(\Sigma)}$, $(\cdot|\cdot)_{\tV_{k}(\Sigma_{\norm})}$.

%We will denote by $A^{\star}$ the formal adjoint of $A$ for the Hermitian structures $(\cdot|\cdot)_{V_{k}(M)}$, $(\cdot|\cdot)_{V_{kt}(\Sigma)}$, $(\cdot|\cdot)_{V_{k}(M_{\norm})}$, $(\cdot|\cdot)_{V_{kt}(\Sigma_{\norm})}$.
%We will similarly denote by $A^{*}$ the formal adjoint of $A$ for the Hilbertian structures $(\cdot|\cdot)_{\tV_{k0}(\Sigma)}$, $(\cdot|\cdot)_{\tV_{k0}(\Sigma_{\norm})}$. 
% 
 For an operator $A$ acting on tensors over $\Sigma$ resp.~$\Sigma_{\norm}$ we can compute both $A^{\star}$ and $A^{*}$.  Using \eqref{e2.12a}, \eqref{e2.12d} we have:
 \begin{equation}
 \label{e2.12e}
 A^{\star}= \tau_{k}^{-1}A^{*}\tau_{k},\quad (A^{\star})^{*}= (A^{*})^{\star},
 \end{equation}
 since $\tau_{k}= \tau_{k}^{*}$ and $\tau_{k}^{2}=\one$.

 \subsection{Charges}
\label{sec2.3.3}
\subsubsection{Lorentzian charges}
The Lorentzian charge $\rq_{k\b}$ is the Hermitian form on $\coinf(\Sigma; V_{k}\otimes \cc^{2})$ defined as
\beq\label{e2.lorentzian-charge}
\bar{f}\dual \rq_{k\b}f=(f_{1}| f_{0})_{V_{k}(\Sigma)}+ (f_{0}| f_{1})_{V_{k}(\Sigma)}.
\eeq

%The Lorentzian charge $\rq_{k}$ is the Hermitian form on $\coinf(\Sigma; V_{k}\otimes \cc^{2})$ defined as
%\beq\label{e2.lorentzian-charge}
%\bar{f}\dual \rq_{k}f=(f_{1}| f_{0})_{V_{k0}(\Sigma)}+ (f_{0}| f_{1})_{V_{k0}(\Sigma)}.
%\eeq
% We have 
% \beq\label{e2.lorentzian-charge-matrix}
% \bar{f}\dual \rq_{k}f= (f| q_{k}f)_{\tV_{k}(\Sigma)\otimes \cc^{2}},\quad q_{k}= \mat{0}{\tau_{k}}{\tau_{k}}{0}.
% \eeq
We define the Cauchy data map  $\varrho_{k\b}: \Ker_{\sc}D_{k\b}\to \coinf(\Sigma; V_{k}\otimes \cc^{2})$  adapted to  $D_{k\b}$ by:
\begin{equation}
\label{e2.10bc}
  \varrho_{k\b}u
        \defeq
        \binom{u|_{t=0}}{\i^{-1}\nabla_{t}u|_{t=0}},
\end{equation}
where $\nabla_{t}$ is defined below in \eqref{e6.new15f}.
If $u\in \Ker_{\sc}D_{k\b}$ we have:
\[
\bar{\varrho_{k\b}u}\dual \rq_{k\b}\varrho_{k\b}u= \pm(u|[\one_{J^{\pm}(\Sigma)}, \i D_{k\b}]u)_{V_{k}(M)}, 
\]
where $(\cdot| \cdot)_{V_{k}(M)}$ is defined in \eqref{e2.10bd}.
 
 Sometimes we will need to express $\rq_{k}$ using the Hilbertian scalar product $(\cdot| \cdot)_{\tV_{k}}$ introduced in \eqref{sec2.3.1d1}. We set
 \begin{equation}
 \label{e2.10bee}
 \bar{f}\dual\rq_{k\b}f= (f| q_{k}f)_{\tV_{k}(\Sigma)\otimes \cc^{2}},\quad \hbox{ for }q_{k}= \mat{0}{\tau_{k}}{\tau_{k}}{0},
 \end{equation}
 where $\tau_{k}$ is defined in \eqref{e2.12b}.

% Sometimes we will need to express $\rq_{k}$ using the Hilbertian scalar product $(\cdot| \cdot)_{\tV_{k0}}$ introduced in \eqref{sec2.3.1d1}. We set
% \begin{equation}
% \label{e2.10bee}
% \bar{f}\dual\rq_{k}f\eqdef (f| q_{k}f)_{\tV_{k0}(\Sigma)\otimes \cc^{2}}, \hbox{ for }q_{k}= \mat{\tau_{k}}{0}{0}{\tau_{k}},
% \end{equation}
% where $\tau_{k}$ is defined in \eqref{e2.12b}.

 \subsubsection{Physical charge}
 For $k=2$ we will have to consider also the {\em physical charge}. We define $I_{\b\Sigma}: \coinf(\Sigma; V_{2}\otimes \cc^{2})\to \coinf(\Sigma; V_{2}\otimes \cc^{2})$ by:
 \begin{equation}
 \label{e2.10be}
 \varrho_{2\b} I_{\b}u\eqdef I_{\b\Sigma} \varrho_{2\b}u,
 \end{equation}
  where $I_{\b}$ is defined below in \ref{sec6.1.1}. Using   Prop.~\ref{prop3.normal-split} (vi) we have:
\begin{equation}
\label{e2.10bf}
I_{\b\Sigma}= \mat{I_{\b}\traa{\Sigma}}{0}{0}{I_{\b}\traa{\Sigma}}.
\end{equation}

We set then 
 \beq\label{e2.physical-charge}
 \rq_{I,2\b}= \rq_{2\b}\circ I_{\b\Sigma},
 \eeq
 We denote by $\rq_{k\b\norm}$ the normal version of $\rq_{k\b}$ with $D_{k\b}$ replaced by its normal operator $D_{k\b\norm}$. We have:
\begin{equation}
\label{e2.charge-b}
\bar{f}\dual \rq_{k\b\norm}f=(f_{1}| f_{0})_{V_{k}(\Sigma_{\norm})}+ (f_{0}| f_{1})_{V_{k}(\Sigma_{\norm})}.
\end{equation}
%We denote by $\rq_{k\norm}$ the normal version of $\rq_{k}$ with $D_{k\b}$ replaced by its normal operator $D_{k\b\norm}$. We have:
%\begin{equation}
%\label{e2.charge-b}
%\bar{f}\dual \rq_{k\norm}f=(f_{1}| f_{0})_{V_{k0}(\Sigma_{\norm})}+ (f_{0}| f_{1})_{V_{k0}(\Sigma_{\norm})}.
%\end{equation}
We define similarly the operators $I_{\b\norm}$ resp.~$I_{\b\Sigma\norm}$, which are the normal  operators of $I_{\b}$, $I_{\b\Sigma}$. They are obtained  as in \eqref{e6.-01} and \eqref{e2.10be} with $\rg$ replaced by $\reta$ in the definition   of $I$. Finally we define the normal version of $\rq_{2\b}$: 
\begin{equation}
\label{e2.charge-c}
\rq_{I,2\b\norm}= \rq_{2\b\norm}\circ I_{\b\Sigma\norm}.
\end{equation}

\subsubsection{Euclidean charges}\label{sec8.1.2}
 We will also need the Euclidean charges $\tilde{\rq}_{k\b}$:
 \beq\label{e2.euclidean-charge}
\bar{f}\dual \tilde{\rq}_{k\b}f=(f_{1}| f_{0})_{\tV_{k}(\Sigma)}+ (f_{0}| f_{1})_{\tV_{k}(\Sigma)},
\eeq
and their normal versions
\begin{equation}
\label{e2.euclideian-charge-normal}
\bar{f}\dual \tilde{\rq}_{k\b\norm}f=(f_{1}| f_{0})_{\tV_{k}(\Sigma_{\norm})}+ (f_{0}| f_{1})_{\tV_{k}(\Sigma_{\norm})}.
\end{equation}

% We will also need the Euclidean charges $\tilde{\rq}_{k}$:
% \beq\label{e2.euclidean-charge}
%\bar{f}\dual \tilde{\rq}_{k}f=(f_{1}| f_{0})_{\tV_{k0}(\Sigma)}+ (f_{0}| f_{1})_{\tV_{k0}(\Sigma)},
%\eeq
%and their normal versions
%\begin{equation}
%\label{e2.euclideian-charge-normal}
%\bar{f}\dual \tilde{\rq}_{k\norm}f=(f_{1}| f_{0})_{\tV_{k0}(\Sigma_{\norm})}+ (f_{0}| f_{1})_{\tV_{k0}(\Sigma_{\norm})}.
%\end{equation}

%which will appear  in Subsect.~\ref{sec10.2b} in connection with \calde projectors. 
%We have:
%\beq\label{e2.euclidean-charge-matrix}
%\bar{f}\dual\tilde{\rq}_{k}f= (f| \tilde{q}_{k}f)_{\tV_{k}(\Sigma)\otimes \cc^{2}},\quad \tilde{q}_{k}= \mat{0}{\one}{\one}{0}.
%\eeq

\section{Reduction of $D_{k}$}\init\label{sec6}
In this section we reduce the Lichnerowicz operators $D_{k}$, $k=1,2$, using the reduction of Sect.~\ref{sec2}.

\subsection{Reduced operators}\label{sec6.-1}
%The conformally reduced operator is denoted by $\hD_{k}$, and the final $\b$-operator $D_{k\b}$ is obtained directly from $\hD_{k}$ by  conjugation by $\rho^{-n/2}$.   Here $c$ is the conformal factor in \eqref{e2.conformal-factor-new}, and the reduced metric $\hrg$ was defined in \eqref{eq:metric-reduced}. Thus we use
We set
\beq\label{e6.-1}
 U_{\b}= (c\rho)^{-n/2},\quad
D_{k\b}\defeq (c\rho)^{-1}U_{\b}D_{k}U_{\b}^{-1}(c\rho)^{-1}.
\eeq
%\begin{equation}
%\label{e6.-0}
%U_{\b}= \rho^{-n/2},
%\end{equation}
%\beq\label{e6.-2}
%D_{k\b}\defeq \rho^{-1}U_{\b}\hD_{k}U_{\b}^{-1}\rho^{-1}.
%\eeq
  Sometimes to emphasize the dependence of $D_{k}$ on the background metric $\rg$ we will write $D_{k}^{\rg}$. 
  
   When no confusion is possible we suppress tensor degree $k$ and write $D, D_{\b}$ for $D_{k}, D_{k\b}$.
  
  \subsubsection{Trace reversal}\label{sec6.1.1}
Recall that  the trace reversal $I$ was defined in Subsect.~\ref{blurck}. We set
 \beq\label{e6.-01}
 I_{\b}\defeq (c\rho) U_{\b}IU_{\b}^{-1}(c\rho)^{-1}= I.
 \eeq
 Since $[D_{2}, I]=0$ we have also $[D_{2\b}, I_{\b}]=0$.
%
%
%In the $\b$-framework the trace reversal is simply the trace reversal in the $\b$-coframe associated with the reduced metric. On the Cauchy surface this gives
%\beq\label{e2.trace-reversal-b}
%I_{\b\Sigma}u= u - \frac{1}{4}(\hrg(0)| u)_{(\hrg(0)^{-1})^{\otimes 2}} \hrg(0).
%\eeq
%We will keep the shorthand $I_{\b}$ when no confusion with its restriction to $\Sigma$ is possible.
%

  \subsection{Preparations}\label{sec6.0}
  \subsubsection{Adapted frames on $V_{k}$}\label{sec6.0.1}
  In \ref{not:tensor-valued-classes} we identified  $V_{k}$ with $\rr^{n}\otimes \cc^{d_{k}}$ using the frame obtained from $(dx^{i})$.
  
  In this section  we will use  over $M_{\infty}$ the frame obtained from $( e^{a})=(r dt,dr,r d\omega)$.  The dual frame of $T\rr^{n}$ is $( e_{a})=(r^{-1}\p_{t},\p_{r},r^{-1}\p_{\omega})$ over $M_{\infty}$. 
   
  If $w$ is a $(0,1)$-tensor we have $w= w_{a} e^{a}$ for $w_{a}= w\cdot e_{a}$ and if $u$ is a $(0,2)$-tensor $u= u_{ab} e^{a}\otimes  e^{b}$ for $u_{ab}= e_{a}\dual u e_{b}$.
In particular $\rg_{ab}=  e_{a}\dual \rg  e_{b}$ and $\rg^{ab}=  e^{a}\dual \rg^{-1} e^{b}$.

This change of frame does not change the spaces $S^{m}_{1}(\rr^{n}; V_{k})$, $\Diff^{m}_{1}(\rr^{n}; V_{k})$, $\Diff^{m}_{\b, 1}(M; V_{k})$ or $H^{m, r}_{\b}(M; V_{k})$.
\subsubsection{Metric connection}\label{sec6.0.2}
Let $\nabla^{\rg}$ be the metric connection for $\rg$ and  $\Gamma\indices{^{c}_{ab}^{\rg}}$  the associated connection coefficients. 
Recall that
\beq\label{e6.new11}
(\nabla_{ e_{a}}^{\rg})^{\star}= -(\nabla_{ e_{a}}^{\rg}+ {\rm div}^{\rg} e_{a}),
\eeq
where $A^{\star}$ is the adjoint for  $(\cdot| \cdot)_{V_{k}^{\rg}(\rr^{n})}$.
Using that  $\rg=\reta+ S^{-1}(\rr^{n}; \otimes^{2}_{\rm s}T^{*}\rr^{n})$ it is straightforward to show that:
\begin{equation}
\label{e6.new1}
\begin{aligned}
\Gamma\indices{^{c}_{ab}^{\rg}}&= \Gamma\indices{^{c}_{ab}^{\reta}}+ S^{-2}(\rr^{n}),\quad {\rm div}^{\rg} e_{a}- {\rm div}^{\reta} e_{a}\in S^{-2}(\rr^{n}),\\
\Gamma\indices{^{c}_{ab}^{\reta}},\ {\rm div}^{\reta} e_{a}&\in S^{-1}_{\cl}(\rr^{n}),
\end{aligned}
\end{equation}
so
\begin{equation}
\label{e6.new16c}
\Gamma\indices{^{c}_{ab}^{\rg}}, {\rm div}^{\rg} e_{a}\in S^{-1}_{1}(\rr^{n}),
\end{equation}
see \ref{secnot.1.1} for notation. Therefore
\beq\label{e6.new-11}
\cov{a}^{\rg}= \p_{ e_{a}}+ R_{ka}^{\rg},\ \ R_{ka}^{\rg}\in S_{1}^{-1}(\rr^{n}; L(V_{k})),
\eeq
\begin{equation}
\label{e6.new-17}
(\cov{a}^{\rg})^{\star}=-(\cov{a}^{\rg}+ r_{a}^{\rg}),\ \ r_{a}^{\rg}\in S^{-1}_{1}(\rr^{n}).
\end{equation}

 By \eqref{e6.new1}  we obtain that:
\begin{equation}
\label{e6.new2}
\begin{aligned}
(\nabla_{ e_{a}}^{\rg})&=(\nabla_{ e_{a}}^{\reta})+ S_{ka},\ \ S_{ka}\in S^{-2}(\rr^{n}; L(V_{k})),\\
(\nabla_{ e_{a}}^{\rg})^{\star}&=(\nabla_{ e_{a}}^{\reta})^{\star}+ T_{ka},\ \ T_{ka}\in S^{-2}(\rr^{n}; L(V_{k})),
\end{aligned}
\end{equation}
where on the l.h.s. resp.~r.h.s. the adjoint is for  $(\cdot| \cdot)_{V^{\rg}_{k}(\rr^{n})}$ resp.~$(\cdot| \cdot)_{V^{\reta}_{k}(\rr^{n})}$.

  \subsection{The operator $D_{k}$}
 We recall that $-\square^{\rg}= (\nabla^{\rg}_{ e_{a}})^{\star}\rg^{ab}\nabla^{\rg}_{ e_{b}}$ hence
\beq\label{e6.new9}
D_{k}^{\rg}= (\nabla^{\rg}_{ e_{a}})^{\star}\rg^{ab}\nabla^{\rg}_{ e_{b}}+ R_{k}^{\rg}, 
\eeq
for $R_{1}^{\rg}=0$ and $R_{2}^{\rg}= 2\Riem_{\rg}\in S^{-3}(\rr^{n}; L(V_{2}))$.
From \eqref{e6.new1}, \eqref{e6.new2} we obtain that
\beq\label{e6.new3}
 \begin{aligned}
D_{k}^{\rg}&= ((\nabla_{ e_{a}}^{\reta})^{\star}+ T_{ka})\rg^{ab}(\nabla_{ e_{b}}^{\reta}+ S_{kb})+ R_{k}^{\rg}\\
&=D_{k}^{\reta}+ M_{k},
\end{aligned}
 \eeq
  where $M_{k}\in r^{-1}\Diff^{2}(\rr^{n}; V_{k})$, see \ref{secnot.1.1b}.

Let us denote $D_{k}^{\rg}$ , $\nabla^{\rg}$ simply by $D_{k}$ and $\nabla$.  From \eqref{e6.new9}, \eqref{e6.new-11} and the fact that $\rg_{tt}= - c^{-2},\quad \rg_{tr}= \rg_{t\omega}=0$ we obtain  that
\beq\label{e6.new13}
D_{k}= -\cov{t}^{\star}c^{2}\cov{t}+ a_{k}(t), 
\eeq
where
\begin{equation}
\label{e6.new7}
\begin{aligned}
D_{k}&= D_{k}^{\star}\hbox{ for }(\cdot| \cdot)_{V^{\rg}_{k}(\rr^{n})},\\
a_{k}(t)&\in \cinf(]- \delta, \delta[; \Diff^{2}_{1}(\Sigma; V_{k}),\\
a_{k}(t)&=a_{k}(t)^{\star}\hbox{ for }(\cdot| \cdot)_{V_{k}(\rr^{n})},\\
\sigma_{\rm pr}(a_{k}(t))&= c^{2}\xi\dual \rh_{t}^{-1}\xi\one_{V_{k}}.
\end{aligned}
\end{equation}

\subsection{The operator $D_{k\b}$}\label{sec6.2}
We now consider $D_{k\b}$ defined in \eqref{e6.-1}.  We recall that on  $\be TM$ we use the frame
\[
( e_{a\b})\defeq \rho^{-1} e_{a}=(\p_{t},-\rho\p_{\rho},\p_{\omega})
\hbox{ over }M_{\infty},
\]
and the dual frame of $\be T^{*}M$ is $( e^{a}_{\b})=(dt,-d\rho/\rho,d\omega)$. On the fibers of $V_{k}$ we use as before the frame induced from $( e_{a})$.
\def\covb{\nabla_{t}}
We set
\begin{equation}
\label{e6.new15f}
\nabla_{a}\defeq cU_{\b}\nabla_{e_{a}}U_{\b}^{-1}(c\rho)^{-1}= \p_{e_{a\b}}+ R_{k\b a},\ \ R_{k\b a}(t)\in \cA_{1}^{\infty}(M; L(V_{k})).
\end{equation}
We have:
\beq\label{e6.new16}
(\nabla_{t})^{\star}=-(\nabla_{t}+ r_{\b}),\ \ r_{\b}\in  \cA_{1}^{\infty}(M),
\eeq
where $A^{\star}$ is the adjoint for $(\cdot| \cdot)_{V_{k}(M)}$.

Let us summarize the properties of $D_{k\b}$.

\begin{proposition}\label{prop3.normal-split}
We have
\[
D_{k\b}=-(\nabla_{t})^{\star}\nabla_{t}+a_{ k\b}(t),
\]
\[
\nabla_{t}= \p_{t}+ R_{k\b},\quad (\nabla_{t})^{\star}= -(\nabla_{t}+ r_{\b}),
\]
where:
\begin{enumerate}[label={\rm(\roman*)},leftmargin=7ex]
\item
$D_{k\b}= D_{k\b}^{\star}$   for $(\cdot| \cdot)_{V_{k}(M)}$.
\item
$[D_{2\b}, I_{\b}]=[\nabla_{t}, I_{\b}]= [a_{2\b}(t), I_{\b}]=0$.
\item
$a_{ k\b}(\cdot)\in \cinf(I;\Diff^{2}_{\b,1}(\Sigma;V_{k}))$,
\item
$a_{k\b}(t)=a_{k\b}^{\star}(t)$ for $(\cdot| \cdot)_{V_{k}(M)}$.
\item $\sigma_{\rm pr}(a_{ k\b})(t)=\xi\dual\rh_{\b}^{-1}(t)\xi\one_{V_{k}}$.
\item $R_{k\b}(\cdot)\in \cA_{1}^{\infty}(M; L(V_{k}))$, $r_{\b}\in \cA^{\infty}_{1}(M)$.

%\item
%The spatial part can be written as
%\[
%a_{k\b}(t)=(\nabla^{\rm n}_{ e_{\b}^{\alpha}})^{\star_{t}}
%\rh_{\b}^{\alpha\gamma}(t)\nabla^{\rm n}_{ e_{\b}^{\gamma}}+q_{k\b}(t),
%\]
%with $q_{k\b}(t)\in \rho\cA^{\infty}(\Sigma;L(V_{k}))$.

\end{enumerate}
\end{proposition}
\proof
We have:
\[
-(c\rho)^{-1}U_{\b}(\nabla_{ e_{t}})^{\star}c^{2}\nabla_{ e_{t}}U_{\b}^{-1}(c\rho)^{-1}=-(cU_{\b}\nabla_{e_{t}}U_{\b}^{-1}(c\rho)^{-1})^{\star}cU_{\b}\nabla_{e_{t}}U_{\b}^{-1}(c\rho)^{-1}= -(\nabla_{t})^{\star}\nabla_{t},
\]
which implies the expression for $D_{k\b}$ for \[
a_{k\b}(t)=(c \rho)^{-1}U_{\b}a_{k}(t)U_{\b}^{-1}(c\rho)^{-1}.
\]
Properties  (iii) and (v) follow. Properties (i) and (iv) follow from the unitarity of $U_{\b}$ between the appropriate Hermitian spaces. 
Since $\nabla\rg =0$ we see that $[\nabla_{e_{t}}, I]=0$, hence also $[\nabla_{t}, I]=0$.
The second statement of (ii)  follows then from the fact that $I= I_{\b}$, and implies the third statement of (ii).
 (vi) follows from \eqref{e6.new-11} and \eqref{e6.new-17}.   \qed
 
 \subsection{The normal operator $D_{k\b\norm}$}\label{sec6.3}
We will also need to compute the normal operator $D_{k\b\norm}$ of $D_{k\b}$. We will denote by  $D_{k\b}^{\reta}$, $R_{k\b}^{\reta}$, $r_{\b}^{\reta}(t)$ the objects in Subsect.~\ref{sec6.2} for $\rg= \reta$.

\begin{proposition}\label{prop10.1}
 We have 
 \[
 D_{k\b\norm}=D^{\reta}_{k\b}=  -(\nabla_{t\norm})^{\star}\nabla_{t\norm}+a_{k\b\norm}(t), 
 \]
 where $A^{\star}$ is the adjoint for $(\cdot| \cdot)_{V(M_{\norm})}$ and
 \[
 \nabla_{t\norm}= \p_{t}+ R_{k\b\norm}(t),\quad (\nabla_{t\norm})^{\star}= -(\nabla_{t\norm}+ r_{\b\norm}(t)),
 \]
\[
  R_{k\b\norm}(t)= R_{k\b}^{\reta}(t),\quad a_{k\b\norm}(t)= a^{\reta}_{k\b}(t),\quad r_{\b\norm}(t)= r_{\b}^{\reta}(t).
 \]

 \end{proposition}
 \proof We omit the subscript $k$. 
 From \eqref{e6.new3} we obtain that 
 \beq\label{e6.new17}
 D_{\b}= D_{\b}^{\reta}+ \rho\Diff^{2}_{\b}(M; V_{k}),
 \eeq
  so it suffices to compute the normal operator $D_{\b\norm}^{\reta}$. Since the Minkowski metric $\reta$ is homogeneous of degree $2$ in $r$, we see that $D_{\b\norm}^{\reta}= D_{\b}^{\reta}$. We  can then apply  Prop.~\ref{prop3.normal-split} to  the background metric $\reta$ to obtain the proposition.  
 \qed

 \section{Reduction of $K$}\label{sec7}\init
In this section we perform the analogous reduction of the 
gauge operator $K=I\circ d$ defined in Subsect.~\ref{sec1.2}.
We set
\[
 d_{\b}\defeq (c\rho) U_{\b}d U_{\b}^{-1}(c\rho)^{-1},
\] 
and
\beq\label{sec7.e2a}
 K_{\b}\defeq (c\rho) U_{\b}K U_{\b}^{-1}(c\rho)^{-1}= I_{\b}\circ d_{\b},
\eeq
where  $U_{\b}$, $I_{\b}$ are    defined in Subsect.~\ref{sec6.-1}.
  From $D_{2}\circ K=K\circ D_{1}$ we get
\begin{equation}
\label{e7.0}
D_{2\b}\circ K_{\b}=T_{\b}\circ D_{1\b},
\quad
T_{\b}\defeq (c\rho)^{-2}K_{\b}(c\rho)^{2}.
\end{equation}
\subsection{The operator $d_{\b}$}
We use the same convention for frames as in Subsect.~\ref{sec6.0}. To emphasize the dependence of $d$ on the background metric $\rg$ we will sometimes write $d^{\rg}$. We have
\[
(dw)_{ab}= \12((\nabla_{e_{a}}w)_{b}+ (\nabla_{e_{b}}w)_{a}),
\]
so
\beq\label{e7.new6}
(d_{\b}w)_{ab}= \12\rho((\nabla_{a}w)_{b}+ (\nabla_{b}w)_{a}),
\eeq
where we recall that  $\nabla_{a}$ is defined in \eqref{e6.new15f}.
From \eqref{e6.new2} we obtain that
\begin{equation}
\label{e7.new1}
d^{\rg}_{\b}= d^{\reta}_{\b}+ \rho^{2}\cA^{\infty}(M;L(V_1,V_2)).
\end{equation}
\subsection{The Cauchy operator $d_{\b\Sigma}$}\label{sec7.1b}
From  $D_{2}\circ d= d\circ D_{1}$, we deduce that $d_{\b}: \Ker_{\sc}D_{1\b}\to \Ker_{\sc}D_{2\b}$ and  can define the Cauchy operator $d_{\b\Sigma}$ associated to $d_{\b}$ by
\begin{equation}
\label{e7.new2}
d_{\b\Sigma}\varrho_{1\b}w\defeq \varrho_{2\b}d_{\b}w,\ \ \hbox{ for }w\in \Ker_{\sc}D_{1\b}.
\end{equation}

Since $K_{\b}= I_{\b}\circ d_{\b}$ we have
\begin{equation}
\label{e7.new5}
K_{\b\Sigma}= I_{\b\Sigma}\circ d_{\b\Sigma},
\end{equation}
and $I_{\b\Sigma}$ is computed in \eqref{e2.10bf}. 

The first trace $(d_{\b}w)_{ab}\traa{\Sigma}$ can be directly expressed in terms of $\varrho_{1\b}u_{1}$, as also the components $(\i^{-1}\nabla_{t}d_{\b}w)_{ab}\traa{\Sigma}$ of the second trace for $a, b\neq t$. For the components $(\i^{-1}\nabla_{t}d_{\b}w)_{ab}\traa{\Sigma}$ with $a=t$ or $b=t$ we use that $D_{1\b}w=0$ and hence $\nabla_{t}^{2}w= - r_{\b}\nabla_{t}w- a_{1\b}w$.

Denoting by $(d_{\b\Sigma})_{\alpha\beta}$ the matrix entries of $d_{\b\Sigma}$ we obtain that
\begin{equation}
\label{e7.new3}
(d_{\b\Sigma})_{\alpha\beta}\in \rho\Diff^{1+ \alpha- \beta}_{\b\: 1}(\Sigma).
\end{equation}
\subsection{Normal operators}
From \eqref{e7.new1} and \eqref{e7.new1} we obtain that
\begin{equation}
\label{e7.new4}
d^{\rg}_{\b\Sigma}= d^{\reta}_{\b\Sigma}+ \rho^{2}\cA^{\infty}(\Sigma;L(V_1,V_2)).
\end{equation}
For the  gauge operator \eqref{e7.new5}, $K_{\b\Sigma\norm}$ is defined by  
$K_{\b\Sigma\norm}\rho^{-1}=(K_{\b\Sigma}\rho^{-1})_{\norm}$.

\section{\calde projectors for the Wick-rotated normal operators}\label{sec10}\init
In this section we consider the Wick rotation of the normal operator $D_{k\b\norm}$, which is actually the free Laplacian on $\rr^{n}$ written in polar coordinates. We define and study the associated Dirichlet to Neumann map and \calde projectors. 

\def\tn{\tilde{n}}

\subsection{Wick rotation}\label{sec10.1}
In this subsection we discuss  following \cite[Sect.~3]{GW3} the Wick rotation of 
the normal operator $D_{k\b\norm}$, which by Prop.~\ref{prop10.1} is equal to $D^{\reta}_{k\b}$.  As a rule Wick-rotated versions of  Lorentzian objects will be denoted with tildes.
\subsubsection{Wick rotation of analytic tensors}\label{sec10.1.1}
We work on $M_{\norm}=I_{t}\times \rr^{+}_{\rho}\times \bS^{n-2}_{\omega}$. We consider $(k',l')$-tensors over  $M_{\norm}$ which are real analytic in $t$, and express such a  tensor as $u= u_{kl}(t, \cdot)\p_{t}^{\otimes k}dt^{\otimes l}$, where $u_{kl}(t,\cdot)$ is a $(k'-k,l'-l)$-tensor over $\rr^{+}_{\rho}\times \bS^{n-2}_{\omega}$. Then the {\em Wick rotation} $\widetilde{u}$ is by definition
\[
\widetilde{u}= \i^{l-k} u_{kl}(\i s, \cdot)\p_{s}^{\otimes k}ds^{\otimes l},
\]
i.e.~the  Wick rotation consists in replacing  $t$ by $\i s$, $dt$ by $\i ds$, $\p_{t}$ by $\i^{-1}\p_{s}$,  keeping the $(\rho, \omega)$ and the associated tensor components fixed.
\subsubsection{Wick-rotated metric}\label{sec10.1.2}
 Setting $t= \i s$  in the expression of $\reta$ given in \eqref{e2.mink-hyp}, we obtain the metric
 \beq\label{e10.n-22}
 \treta=r^2ds^2+dr^2+r^2\cos^2(s)\rh_{\bS^{n-2}},
 \eeq
 where we recall that $\rh_{\bS^{d}}$  is the canonical metric on $\bS^{d}$.
 
 We identify $ ]-\pi/2, \pi/2[_{s}\times \bS^{n-2}_{\omega}$  with $\bS^{n-1}\subset \rr^{n}$ with the two poles $y^{0}= \pm 1$ removed by the map
 \beq
 \label{e10.n1}
 ]-\pi/2, \pi/2[_{s}\times \bS^{n-2}_{\omega}\ni (s, \omega)\mapsto (\sin s, \cos s \omega)\in \bS^{n-1}.
 \eeq
 
 Then $\treta$ is simply the Euclidean metric on $\rr^{n}$ written in polar coordinates. Accordingly we set
 \[
 \begin{aligned}
\tM_{\norm}&= \rr^{+}_{\rho}\times \bS^{n-1},\quad
\treta_{\b}= (\frac{d\rho}{\rho})^{2}+ \rh_{\bS^{n-1}}=ds^{2}+ \tilde{\boldsymbol{\delta}}_{\b},\\
\tilde{\boldsymbol{\delta}}_{\b}&= (\frac{d\rho}{\rho})^{2}+ \cos^{2}(s)\rh_{\bS^{n-2}} ,\\
\Sigma_{\norm}&= \tM_{\norm}\cap \{s=0\}= \rr^{+}_{\rho}\times \bS^{n-2},\quad \hbox{ equipped with }\tdelta_{\b}(0)= (\frac{d\rho}{\rho})^{2}+ \rh_{\bS^{n-2}}.
\end{aligned}
\]
\subsubsection{Wick rotation of operators on $M_{\norm}$}\label{sec10.1.3}
If $A$ is a $\b$-differential operator acting on tensors on $M_{\norm}$ which is  real analytic in $t$, we define its Wick rotation $\widetilde{A}$ by:
\[
\widetilde{A}\widetilde{u}\defeq \widetilde{Au}, \hbox{ for }u \hbox{ analytic in }t.
\]
The  operator $\widetilde{A}$ can then  be applied to smooth tensors on $\tM_{\norm}$.
\subsubsection{Wick rotation of operators on $\Sigma_{\norm}$}\label{sec10.1.4}
 Let $A_{ij\Sigma}$, $i, j\in \{1, 2\}$ be a linear operator from $\cinf(\Sigma_{\norm}; V_{j}\otimes \cc^{2})$ to $\cinf(\Sigma_{\norm}; V_{i}\otimes \cc^{2})$. We set
\begin{equation}
\label{e10.n-4}
\widetilde{A}_{ij\Sigma}\defeq \cF_{i\Sigma}\circ A_{ij\Sigma}\circ \cF_{j\Sigma}^{-1},
\end{equation}
where
\beq\label{e10.n-7}
\cF_{k\Sigma}= \mat{\cF_{k}}{0}{0}{\cF_{k}}
\eeq
and $\cF_{k}$ is defined by:
 \beq\label{e10.n-6}
   \cF_{1}= \mat{\i}{0}{0}{\one}, \quad \cF_{2}= \left(\begin{array}{ccc}
-1&0&0\\0&\i&0\\0&0&\one
\end{array}\right),
\eeq
where we use the decomposition of $(0, k)$-tensors in Subsect.~\ref{sec2.3.2}. This is consistent with the prescription already introduced in Sects.~\ref{sec10.1.1}-\ref{sec10.1.3}, but we will apply it not only to differential operators.

 \subsubsection{Hilbertian scalar products}\label{sec10.1.5}
 We equip the fibers of $V_{k}$ with the {\em Hilbertian} product $(\cdot| \cdot)_{\tV_{k}^{\treta}}$ induced by $\treta$. 
 We set
 \beq\label{e10.n-3}
 \begin{aligned}
(u| u)_{\tV_{k}(\tM_{\norm})}&= \int_{\tM_{\norm}}(u| u)_{\tV_{k}^{\treta}}\dVol_{\treta_{\b}},\\
(u| u)_{\tV_{k}(\Sigma_{\norm})}&= \int_{\Sigma_{\norm}}(u| u)_{\tV_{k}^{\treta(0)}}\dVol_{\tdelta_{\b}(0)}.
\end{aligned}
 \eeq
The notation $(\cdot| \cdot)_{\tV_{k}(\Sigma_{\norm})}$ is coherent with the one introduced in \eqref{e2.11c} since $\boldsymbol{\delta}(0)= \widetilde{\boldsymbol{\delta}}(0)$.

% We set
% \beq\label{e10.n-3}
% \begin{array}{l}
%  (u| u)_{\tV_{k}(\tM_{\norm})}= \int_{\tM_{\norm}}(u| u)_{\tV_{k}^{\treta}}\dVol_{\treta_{\b}},\\[2mm]
%   (u| u)_{\tV_{k0}(\Sigma_{\norm})}= \int_{\Sigma_{\norm}}(u| u)_{\tV_{k}^{\treta(0)}}\dVol_{\tdelta_{\b}(0)}.
%\end{array}
% \eeq
%The notation $(\cdot| \cdot)_{\tV_{k0}(\Sigma_{\norm})}$ is coherent with the one introduced in \eqref{e2.11c} since $\boldsymbol{\delta}(0)= \widetilde{\boldsymbol{\delta}}(0)$.

\subsection{Wick-rotated normal operators}\label{sec10.2}
We saw in Prop.~\ref{prop10.1} that the normal operator $D_{k\b\norm}$ of 
 $D_{k\b}$ is equal  to $D_{k\b}^{\reta}$.  
 
 The Wick rotation  $\tD_{k}^{\reta}$ of $D_{k}^{\reta}$ is simply the  operator $D_{k}^{\treta}$, i.e.~minus the trivial flat Laplacian  in $\rr^{n}$ acting on $(0, k)$-tensors. It follows that the Wick rotation $\tD_{k\b\norm}$ is 
 $D_{k\b}^{\treta}= \rho^{-1-\frac{n}{2}}
 D_{k}^{\treta}
 \rho^{-1+\frac{n}{2}}$, acting on $\tM_{\norm}$.
  
 \subsection{Properties of $\tD_{k\b\norm}$}\label{sec10.3}

 If we use the Euclidean frame $(dx^{i})$ to identify  $V_{k}$ with $\tM_{\norm}\times \cc^{d_{k}}$, then 
 $\tD_{k}^{\reta}$ is the scalar Laplacian $-\Delta^{\treta}$ and 
 \beq\label{e10.n3}
 \tD_{k\b\norm}= (-(\rho\p_{\rho})^{2}+ (\tfrac{n-2}{2})^{2}- \Delta_{\bS^{n-1}})\otimes\one_{\cc^{d_{k}}},
 \eeq
  where $\Delta_{\bS^{n-1}}$ is the scalar Laplacian on $\bS^{n-1}$. Clearly $\tD_{k\b\norm}$ is selfadjoint for $(\cdot| \cdot)_{\tV_{k}(\tM_{\norm})}$ with $\tD_{k\b\norm}\geq (\frac{n-2}{2})^{2}$. By a standard argument this implies that
  \beq\label{e10.n4}
  \tD_{k\b\norm}: H^{s}_{\b}(\tM_{\norm}; V_{k})\tosim H^{s-2}_{\b}(\tM_{\norm}; V_{k}),\ \ s\in \rr.
  \eeq
  \subsubsection{Gauge invariance}\label{sec10.3.1}
  Let $\tK_{\b\norm}$ be the Wick rotation of $K_{\b\norm}$. From \eqref{e7.0} we obtain $D_{2\b\norm}\circ K_{\b\norm}= T_{\b\norm}\circ D_{1\b\norm}$, where  $T_{\b\norm}= \rho^{-2}K_{\b\norm}\rho^{2}$ hence:
  \begin{equation}
  \label{e10.n-77}
  \tD_{2\b\norm}\circ \tK_{\b\norm}= \widetilde{T}_{\b\norm}\circ \tD_{1\b\norm} \hbox{ for } \widetilde{T}_{\b\norm}= \rho^{-2}\tK_{\b\norm}\rho^{2}.
  \end{equation}
  
\subsubsection{Expression for $\tD_{k\b\norm}$}\label{sec10.3.2}

$\tD_{k\b\norm}$ and associated objects have the simplest expressions when written with respect to the Euclidean frame. However, our ultimate goal is to use it for a construction of Hadamard projectors, and for this purpose we will need expressions with respect to Wick-rotated analogues of the frames used  in Sect.~\ref{sec6}. Therefore, we now derive an expression for $\tD_{k\b\norm}$ similar to the one for $D_{k\b}$ in Subsect.~\ref{sec6.2}. We often suppress the index $k$.

 On $\tM_{\norm}$ we use the frame $\tilde{e}_{a}= (r^{-1}\p_{s}, \p_{r}, r^{-1}\p_{\omega})$.
  We denote by $\widetilde{\nabla}$ the metric connection for $\treta$, set $\tilde{r}_{a}= {\rm div}^{\treta}\tilde{e}_{a}$ and as in Subsect.~\ref{sec6.2} 
 \beq\label{e10.n2b}
 \widetilde{\nabla}_{s}=\rho^{-\frac{n}{2}}\widetilde{\nabla}_{ e_{s}}\rho^{-1+\frac{n}{2}}.
 \eeq
 Then as in Prop.~\ref{prop3.normal-split} we have:
 \beq\label{e10.n2}
 \tD_{\b\norm}= ( \widetilde{\nabla}_{s})^{*} \widetilde{\nabla}_{s}+ \tilde{a}_{\b}(s),
 \eeq
 for $\tilde{a}_{\b}(s)\in \cinf(I; \Diff^{2}_{\b}(\Sigma_{\norm}; V))$.

 If we identify $V$ with $\tM_{\norm}\times \cc^{d}$ using the frame obtained from  the dual frame $\tilde{e}^{a}= (rds, dr, rd\omega)$ we have as in Subsect.~\ref{sec6.2}:
 \begin{equation}
 \label{e10.n-1}
 \tD_{\b\norm}= - (\p_{s}+ \tilde{R}_{\b\norm}(s)+\tilde{r}_{\b\norm}(s))(\p_{s}+\widetilde{R}_{\b\norm}(s))+ \tilde{a}_{\b\norm}(s).
 \end{equation}
 while, see Prop.~\ref{prop10.1}:
 \begin{equation}
 \label{e10.n-5}
 D_{\b\norm}= (\p_{t}+ R_{\b\norm}(t)+r_{\b\norm}(t))(\p_{t}+ R_{\b\norm}(t))+ a_{\b\norm}(t)
 \end{equation}
 
  Let us  relate $\tilde{R}_{\b\norm}(s), \tilde{r}_{\b\norm}(s), \tilde{a}_{\b\norm}(s)$ with their Lorentzian analogs $R_{\b\norm}(t), r_{\b\norm}(t), a_{b\norm}(t)$ in Prop.~\ref{prop10.1}. Recall that $\cF_{k}$ for $k= 1, 2$ is defined in \eqref{e10.n-6}.
  
\begin{lemma}\label{lemma10.0}
We have \[
  \tilde{a}_{k\b\norm}(s)= \cF_{k} a_{k\b\norm}(\i s)\cF_{k}^{-1},\quad \tilde{R}_{k\b\norm}(s)= \i\cF_{k} R_{k\b\norm}(\i s)\cF_{k}^{-1},\quad \tilde{r}_{\b\norm}(s)= \i r_{\b\norm}(\i s). \]
\end{lemma}
\proof   We suppress the index $k$. Let us write $D_{\b\norm}$ as $D_{\b\norm}(t, \p_{t})$ and $\tD_{\b\norm}$ as $\tD_{\b\norm}(s, \p_{s})$. Because of the change of frame from $(e^{a})$ to $(\tilde{e}^{a})$, we have
  \beq\label{e10.n-2}
  \tD_{\b\norm}(s, \p_{s})= \cF D_{\b\norm}(\i s, - \i\p_{s})\cF^{-1}.
  \eeq
This implies the lemma by a direct computation. \qed    
  \subsubsection{Imaginary time reflection}\label{sec10.3.3}
  We denote by $\kappa_{k}: \cinf(\tM_{\norm}; V_{k})\to \cinf(\tM_{\norm}; V_{k})$ the  lift to $V_{k}$ of the involution $\chi: (\rho, s, \omega)\mapsto (\rho, -s, \omega)$ on $\tM_{\norm}$. We have
\begin{equation}
\label{e10.n4b}
\kappa_{k}\circ \tD_{k\b\norm}= \tD_{k\b\norm}\circ \kappa_{k}.
\end{equation}
Since $\chi= \one$ over $\Sigma$, $\kappa_{k}$ preserves the fibers of $V_{k}$ over $\Sigma$ and we denote by $\kappa_{k\Sigma}$ its restriction to $V_{k}$ over $\Sigma$.

If we decompose  a $(0, k)$-tensor as in Subsect.~\ref{sec2.3.2} we obtain that 
\beq\label{e10.n4c}
\kappa_{1\Sigma}= \mat{-1}{0}{0}{\one},\quad \kappa_{2\Sigma}= \left( \begin{array}{ccc}
1&0&0\\
0&-\one&0\\
0&0&\one
\end{array}\right).
\eeq

  \subsection{Green's formula}\label{sec10.3b}
  We set $ \Omega^{\pm}\defeq \rr^{+}_{\rho}\times \bS^{n-1\pm}$,  where $\bS^{n-1\pm}$ are the two hemispheres given by $\pm s>0$, see \eqref{e10.n1}. For $u\in \coinf(\bar{\Omega}^{\pm}; V)$  we set
\beq\label{e10.normal-trace}
        \trho^{\pm}u
        \defeq
        \binom{u(0^{\pm})}{-\widetilde{\nabla}_{s}u(0^{\pm})}.
\eeq

The operators $\trho^{\pm}$ uniquely extend to $\trho^{\pm}: H^{2}_{\b}(\Omega^{\pm}; V)\to H^{3/2}_{\b}(\Sigma_{\norm}; V)\oplus H^{\12}_{\b}(\Sigma_{\norm}; V)$.
\def\tsig{\tilde{\sigma}}
For $u\in\coinf(\tM_{\norm})$ we set
\beq\label{e10.stup-trace}
\trho u \defeq
        \binom{u(0)}{-\widetilde{\nabla}_{s}u(0)}.
\eeq

\begin{proposition}\label{prop10.n1}
Let $ \tsig=\mat{0}{-\one}{\one}{0}$. Then for 
$u,v\in H^{1}_{\b}(\Omega^{\pm}; V)$ with $\tD_{\b\norm}u,\tD_{\b\norm}v\in L^{2}_{\b}(\Omega^{\pm}; V)$ we have:\beq\label{e01.green1}
        (u|\tD_{\b\norm}v)_{\tV(\Omega^{\pm})}
        -(\tD_{\b\norm}u|v)_{\tV(\Omega^{\pm})}
        =
        \pm(\trho^{\pm}u|\tsig\trho^{\pm}v)_{\tV(\Sigma_{\norm})\otimes \cc^{2}}.
\eeq

%Let $ \tsig=\mat{0}{-\one}{\one}{0}$. Then for 
%$u\in H^{1}_{\b}(\Omega^{\pm}; V)$ we have:\beq\label{e01.green1}
%        (u|\tD_{\b\norm}u)_{\tV(\Omega^{\pm})}
%        -(\tD_{\b\norm}u|u)_{\tV(\Omega^{\pm})}
%        =
%        \pm(\trho^{\pm}u|\tsig\trho^{\pm}u)_{\tV_{0}(\Sigma_{\norm})\otimes \cc^{2}}.
%\eeq
%and
%\beq\label{e01.green2}
%        (\tD_{\b\norm}u|u)_{\tV(\Omega^{\pm})}
%        +(u|\tD_{\b\norm}u)_{\tV(\Omega^{\pm})}
%        =2q_{\tV(\Omega^{\pm})}(u,u)
%        \mp(\trho^{\pm}u|\tilde q\trho^{\pm}u)_{\tV(\Sigma_{\norm})\otimes \cc^{2}},
%\eeq
%where
%\[
%        q_{\tV(\Omega^{\pm})}(u,u)
%        =
%        \int_{\tV(\Omega^{\pm})}(\widetilde{\nabla}^{\rm n}_{\p_{s}}u| \widetilde{\nabla}^{\rm n}_{\p_{s}}u)_{\tV}+ (u| \tilde{a}_{\b\norm}(\cdot)u)_{\tV}
%        \,\dVol_{\widetilde g_{\b}}.
%\]
%\textcolor{red}{maybe second Green is not needed}
\end{proposition}
\proof We only prove the $+$ case. By ellipticity we know that $u,v\in H^{2}_{\b}(\Omega^{+}; V)$. Let $eu, ev\in H^{1}_{\b}(\tM_{\norm}; V)$ an extension of $u, v$ to $\tM_{\norm}$. Then using \eqref{e10.n2} we get:
\[
\begin{aligned}
&(u| \tD_{\b\norm}v)_{\tV(\Omega^{+})}- ( \tD_{\b\norm}u| v)_{\tV(\Omega^{+})}\\
&= (eu| [\one_{\Omega^{+}}, \tD_{\b\norm}]ev)_{\tV(\tM_{\norm})}\\
&= (eu| [\one_{\Omega^{+}}, (\widetilde{\nabla}_{s})^{*}]\widetilde{\nabla}_{s}ev)_{\tV(\tM_{\norm})}\\
&\phantom{=\,}+ (eu|(\widetilde{\nabla}_{s})^{*}[\one_{\Omega^{+}}, \widetilde{\nabla}_{s}]ev)_{\tV(\tM_{\norm})}\\
&= (\trho^{\pm}u|\tsig\trho^{\pm}v)_{\tV(\Sigma_{\norm})\otimes \cc^{2}},
\end{aligned}
\]
since $\widetilde{\nabla}_{s}= \p_{s}+ \tilde{R}_{\b\norm}(s)$ where $\tilde{R}_{\b\norm}(s)$ is a multiplication operator. \qed
 \subsection{Dirichlet to Neumann map}\label{sec10.4}
 Let $\gamma u= u\traa{\Sigma_{\norm}}$ and $\gamma^{*}$ is its adjoint for the scalar products $(\cdot|\cdot)_{\tV(\tM_{\norm})}$ and $(\cdot| \cdot)_{\tV(\Sigma_{\norm})}$. Then
 \beq\label{e10.n5}
 \begin{aligned}
\gamma: H^{1}_{\b}(\tM_{\norm}; V)&\to H^{\12}_{\b}(\Sigma_{\norm}; V),\\
\gamma^{*}: H^{-\12}_{\b}(\Sigma_{\norm}; V)&\to H^{-1}_{\b}(\tM_{\norm}; V).
\end{aligned}
 \eeq
 Moreover there exists an extension map \beq\label{e10.n6}
 E:H^{\12}_{\b}(\Sigma_{\norm}; V)\to H^{1}_{\b}(\tM_{\norm}; V)\hbox{ with }\gamma\circ E= \one.
 \eeq
 
 We define $H^{s}_{\b}(\Omega^{\pm}; V)$ to be the space obtained by restrictions of $u\in H^{s}_{\b}(\tM_{\norm}; V)$ to $\Omega^{\pm}$ (so it measures b-regularity at the $\rho=0$ boundary but usual regularity across $\Sigma$).

\begin{definition}\label{def10.1}
 Let $v\in H^{\12}_{\b}(\Sigma_{\norm}; V)$  and $u^{\pm}\in H^{1}_{\b}(\Omega^{\pm}; V)$ the solution of
\beq\label{e10.n2bb}
\left\{
\begin{aligned}
\tD_{\b\norm}u^{\pm}&= 0\hbox{ in }\Omega^{\pm},\\
u^{\pm}(0^{\pm})&=v.
\end{aligned}
  \right.
\eeq
The Dirichlet to Neumann maps are 
\beq\label{e10.fixed-dtn}
        \widetilde n_{\norm}^{\pm}v
        \defeq
        -\widetilde{\nabla}_{s}u^\pm(0^\pm),
\eeq
\end{definition}

\begin{lemma}\label{lemma10.1}
We have
\[
\tn_{\norm}^{\pm}= \pm\12 (\gamma \tD_{\b\norm}^{-1}\gamma^{*})^{-1},
\]
as a map between $H^{\12}_{\b}(\Sigma_{\norm}; V)$ and $H^{-\12}_{\b}(\Sigma_{\norm}; V)$. \end{lemma}
\proof  We denote $H^{s}_{\b}(\tM_{\norm}; V)$, $H^{s}_{\b}(\Sigma_{\norm}; V)$ simply by $H^{s}_{\b}(\tM_{\norm})$, $H^{s}_{\b}(\Sigma_{\norm})$.
 Since $\tD_{\b\norm}$ is invariant under the reflection $(s, \omega)\mapsto (-s, \omega)$, we see that $\tn_{\norm}^{-}= - \tn_{\norm}^{+}$, so it suffices to prove the $+$ case. 

Let $u\in H^{1}_{\b}(\Omega^{+})$ be as in \eqref{e10.n2bb} and let $w\in H^{1}_{\b}(\tM_{\norm})$ the extension of $u$ which is even w.r.t. $s\mapsto -s$. Then $\tD_{\b\norm}w=0$ in $\Omega^{\pm}$ and  by \eqref{e01.green1} we have for $\phi\in \coinf(\tM_{\norm})$:
\[
\begin{aligned}
(\tD_{\b\norm}\phi| w)_{\tV(\tM_{\norm})}
&= (\tD_{\b\norm}\phi| w)_{\tV(\Omega^{+})}+(\tD_{\b\norm}\phi| w)_{\tV(\Omega^{-})}\\
&= -(\tilde{\varrho}\phi| \tsig\tilde{\varrho}^{+}w)_{\tV(\Sigma_{\norm})\otimes \cc^{2}}+(\tilde{\varrho}\phi| \tsig\tilde{\varrho}^{-}w)_{\tV(\Sigma_{\norm})\otimes \cc^{2}}\\
&= (\gamma\phi|n_{\norm}^{+}v)_{\tV(\Sigma_{\norm})}-(\gamma\phi| n_{\norm}^{-}v)_{\tV(\Sigma_{\norm})}\\
&=2(\gamma\phi|n_{\norm}^{+}v)_{\tV(\Sigma_{\norm})}=2(\phi|\gamma^{*}n_{\norm}^{+}v)_{\tV(\tM_{\norm})},
\end{aligned}
\]
using that $\tn_{\norm}^{-}= - \tn_{\norm}^{+}$. Therefore we have $\tD_{\b\norm}w=2\gamma^{*}\tn_{\norm}^{+}v$ hence $v=2 \gamma \tD_{\b\norm}^{-1}\gamma^{*}n_{\norm}^{+}v$. 
Let us set 
\beq\label{e10.n111}
T\defeq  \gamma \tD_{\b\norm}^{-1}\gamma^{*}: H^{-\12}_{\b}(\Sigma_{\norm})\to H^{\12}_{\b}(\Sigma_{\norm}).
\eeq
For $v\in H^{-\12}_{\b}(\Sigma_{\norm})$, $w=  \tD_{\b\norm}^{-1}\gamma^{*}v$ and $u\in H^{1}_{\b}(\Sigma_{\norm})$ we have
$(u| \tD_{\b\norm}w)_{\tV(\tM_{\norm})}=(\gamma u| v)_{\tV(\Sigma_{\norm})}$, hence \beq\label{e10.n8}
(w| \tD_{\b\norm}w)_{\tV(\tM_{\norm})}=(Tv| v)_{\tV(\Sigma_{\norm})}.
\eeq
On the other hand for $\psi\in H^{\12}_{\b}(\Sigma_{\norm})$ we have
\beq\label{e10.n11}
\begin{aligned}
|(\psi| v)_{\tV(\Sigma_{\norm})}|
&= |(\gamma\circ E\psi| v)_{\tV(\Sigma_{\norm})}|= |(E\psi| \gamma^{*}v)_{\tV(\tM_{\norm})}|\\
&= |(E\psi|\tD_{\b\norm}w)_{\tV(\tM_{\norm})}|\\
&\leq C\|E\psi\|_{H^{1}_{\b}(\tM_{\norm})}\| w\|_{H^{1}_{\b}(\tM_{\norm})}\\
&\leq C\|\psi\|_{H^{\12}_{\b}(\Sigma_{\norm})}\| w\|_{H^{1}_{\b}(\tM_{\norm})},
\end{aligned}
\eeq
hence
\begin{equation}
\label{e10.n9}
\| v\|_{H^{-\12}_{\b}(\Sigma_{\norm})}\leq C\| w\|_{H^{1}_{\b}(\tM_{\norm})}.
\end{equation}
Using \eqref{e10.n8} we obtain that
\[
\begin{aligned}
\|v\|^{2}_{H^{-\12}_{\b}(\Sigma_{\norm})}
&\leq C \|w\|^{2}_{H^{1}_{\b}(\tM_{\norm})}\\
&\leq C(w| \tD_{\b\norm}w)_{\tV(\tM_{\norm})}=C(v| Tv)_{\tV(\Sigma_{\norm})}.
\end{aligned}
\]
By Lax-Milgram this implies that 
 \begin{equation}
 \label{e10.n7}
T: H^{-\12}_{\b}(\Sigma_{\norm})\tosim H^{\12}_{\b}(\Sigma_{\norm}),
 \end{equation}
 which shows $\tn_{\norm}^{+}= \12 (\gamma \tD_{\b\norm}^{-1}\gamma^{*})^{-1}$ as a map between $H^{\12}_{\b}(\Sigma_{\norm})$ and $H^{-\12}_{\b}(\Sigma_{\norm})$.  \qed

We now summarize the properties of $\tn_{\norm}^{\pm}$. We use the notation in  Appendix~\ref{sec9}, in particular Subsects.~\ref{sec9.3} and \ref{sec9.4}. 

\begin{proposition}\label{prop10.2}
\ben
\item  $\tn_{\norm}^{-}= - \tn_{\norm}^{+}$.

\item  $\tn_{\norm}^{+}= (\tn_{\norm}^{+})^{*}$ for $(\cdot| \cdot)_{\tV(\Sigma_{\norm})}$ and $\tn_{\norm}^{+}= (\tn_{\norm}^{+})^{\star}$ for $(\cdot| \cdot)_{V(\Sigma_{\norm})}$.
 \item as an unbounded operator  with domain  $H^{1}_{\b}(\Sigma_{\norm}; V)$,  $\tn_{\norm}^{+}$ is  self-adjoint for $(\cdot| \cdot)_{\tV(\Sigma_{\norm})}$ with:
\beq\label{e10.n112}
\pm\tn_{\norm}^{\pm}\geq C(-(\rho\p_{\rho})^{2}- \Delta_{\bS^{n-2}}+1)^{\12}\one_{V},\ \ C>0.
\eeq

 \item 
\[
\tn_{\norm}^{\pm}: H^{m,r}_{\b}(\Sigma_{\norm};V)\tosim H^{m-1,r}_{\b}(\Sigma_{\norm};V)\hbox{ for }m\in \rr,\ r\in J,
\]
$J= ]-\frac{n-2}{2}, \frac{n-2}{2}[$.
\item There exist   $\tb_{\norm}(s)\in \cinf(\rr; \Psi^{1}_{{\rm s}\: \norm}(\Sigma_{\norm}; \tV_{\norm}))$ which is the unique solution   modulo an error in $\cinf(\rr; \Psi^{-\infty}_{{\rm s}\: \norm}(\Sigma_{\norm}; \tV_{\norm}))$ of
 \beq\label{e10.2b}
-\p_{s}\tilde{b}_{\norm}(s)- \tilde{r}_{\b\norm}(s)\tilde{b}_{\norm}(s)- [\tilde{R}_{\b\norm}(s), \tilde{b}_{\norm}(s)]+ \tilde{b}^{2}_{\norm}(s)- \tilde{a}_{\b\norm}(s)
\in \cinf(\rr, \Psi^{-\infty}_{{\rm s}\: \norm}(\Sigma_{\norm}; \tV_{\norm})),
\eeq
with  $\sigma_{\rm pr}(\tb_{\norm})(s)= (\xi\dual \tilde{\boldsymbol{\delta}}_{\b}^{-1}(s)\xi)^{\12}\one_{V_{\norm}}$ such that:
\beq\label{e10.2bc}
\tn_{\norm}^{+}= \tb_{\norm}(0)+ r_{-\infty\norm}, 
\eeq
where 
\[
r_{-\infty\norm}\in B^{-\infty}_{J}(\Sigma_{\norm}; V_{\norm}),\quad \ad_{h}r_{-\infty\norm}\in B^{-\infty}_{J, 0^{+}}(\Sigma_{\norm}),
\]
 for any $h\in \coinf(\rr^{+}_{\rho})$ with $h=1$ near $\rho=0$ and $J=\,\open{-\frac{n-2}{2},\frac{n-2}{2}}$. In other words $\tn_{\norm}^{+}\in N^{1}_{J}(\Sigma_{\norm}; V_{\norm})$.
\een
\end{proposition}
\proof 
   We have already checked  (1) in the proof of Lemma~\ref{lemma10.1}. Applying \eqref{e01.green1} to the solution $u$ of \eqref{e10.n2bb} we obtain that $\tn_{\norm}^{+}= (\tn_{\norm}^{+})^{*}$.  By \eqref{e2.12e} we have $(\tn_{\norm}^{+})^{\star}= \tau^{-1}\tn_{\norm}^{+}\tau$. If we use the Euclidean frame $(dx^{i})$ as in Subsect.~\ref{sec10.3} to identify $V$ with $\tM_{\norm}\times \cc^{d_{k}}$, then $\tD_{b\norm}$ is the scalar $\b$-Laplacian and $\tn_{\norm}^{+}$ is a scalar operator, commuting with $\tau$. Note that the expression of $\tau$ in the frame $(dx^{i})$ or $(e_{a})$ is identical since over $s= 0$ we have $dx^{0}= rds,\quad dx^{i}= \omega^{i}dr+ rd\omega^{i}$. This proves the second statement of (2).

   Let us now prove (3).
   By (2) we know that  $\tn_{\norm}^{+}$ is symmetric on $\coinf(\Sigma_{\norm}; V)$. By the ellipticity of $\tn_{\norm}^{\pm}$ proved in (5) this implies that $\tn_{\norm}^{\pm}$ is selfadjoint with domain $H^{1}_{\b}(\Sigma_{\norm};V)$. This proves the first statement of (3). Next from \eqref{e10.n111} in the proof of Lemma~\ref{lemma10.1} we obtain that
 \[
 (u| Tu)_{\tV(\Sigma_{\norm})}\leq C\|u\|^{2}_{H^{-\12}_{\b}(\Sigma_{\norm})},\ \ u\in \coinf(\Sigma_{\norm}), 
 \]
 hence
 \[
 T\leq C (-(\rho\p_{\rho})^{2}- \Delta_{\bS^{n-2}}+1)^{-\12}\one_{V}, \hbox{ for }(\cdot| \cdot)_{\tV(\Sigma_{\norm})}.
 \]
Then \eqref{e10.n112} follows from the Kato-Heinz theorem. 

Let us now prove (4). 
We use the notation in the proof of Lemma~\ref{lemma10.1}.
  For $r\in J$ we set $\tn^{+}_{\b\norm, r}=\rho^{r}\tn^{+}_{\norm}\rho^{-r}$, and $\tD_{\b\norm, r}= \rho^{r}\tD_{\b\norm}\rho^{-r}$. Since $\rho^{r}(\rho\p_{\rho})\rho^{-r}= \rho\p_{\rho}- r$ we obtain that
  \beq\label{e10.n101}
  \tD_{\b\norm, r}: H^{m, 0}_{\b}(\tM_{\norm})\tosim H^{m-2, 0}_{\b}(\tM_{\norm}),\ \ m\in \rr\hbox{ for }r\in J
  \eeq
  and
  \begin{equation}
  \label{e10.n10}
  \Re(u| \tD_{\b\norm, r}u)_{\tV(\tM_{\norm})}\geq C \|u\|^{2}_{H^{1}_{\b}(\tM_{\norm})},\ \ u\in H^{1}_{\b}(\tM_{\norm}).
  \end{equation}
  Setting $w_{r}= \, \tD_{\b\norm, r}^{-1}\gamma^{*}v$, we deduce as in \eqref{e10.n11} that
  \begin{equation}
  \label{e10.n12}
  \|v\|_{H^{-\12}_{\b}(\Sigma_{\norm})}\leq C\| w_{r}\|_{H^{1}_{\b}(\tM_{\norm})}.
\end{equation}
  Considering now 
 $T_{r}= \gamma \tD_{\b\norm, r}^{-1}\gamma^{*}$, we obtain as in \eqref{e10.n8} that 
\[
 (w_{r}| \tD_{\b\norm, r}w_{r})_{\tV(\tM_{\norm})}=(T_{r}v| v)_{\tV(\Sigma_{\norm})},
 \]
hence using \eqref{e10.n10}, \eqref{e10.n12}
\[
\Re(T_{r}v| v)_{\tV(\Sigma_{\norm})}\geq C\|v\|^{2}_{H^{-\12}_{\b}(\Sigma_{\norm})}.
\]
By Lax Milgram this implies that 
\[
T_{r}: H^{-\12}_{\b}(\Sigma_{\norm})\tosim H^{\12}_{\b}(\Sigma_{\norm}),
\]
hence $\tn^{+}_{\b\norm, r}: H^{\12}_{\b}(\Sigma_{\norm})\tosim H^{-\12}_{\b}(\Sigma_{\norm})$. Commuting $\rho\p_{\rho}$ 	and $\p_{\omega}$ derivatives, we obtain that
$\tn^{+}_{\b\norm, r}: H^{s}_{\b}(\Sigma_{\norm})\tosim H^{s-1}_{\b}(\Sigma_{\norm})$
 for $s\in \rr$ which proves (4).
 
Let us now prove (5).
 We first construct   a parametrix for the Dirichlet problem \eqref{e10.n2bb} following  the arguments in \cite[Subsect.~5.3]{GMW}.  To that end  we factorize $\tD_{\b\norm}$ as \[
 \tD_{\b\norm}= (-(\p_{s}+\tilde{R}_{\b\norm}(s)+\tilde{r}_{\b\norm}(s)) + \tb_{\norm}(s))((\p_{s}+\tilde{R}_{\b\norm}(s))+ \tb_{\norm}(s)), \hbox{ modulo smoothing}.
 \] This leads to \eqref{e10.2b}, which is a  Wick-rotated version of \eqref{e8.10b} below.

As in \cite[Subsect.~5.3]{GMW} one shows that \eqref{e10.2b} has a unique solution with the properties stated in the proposition.  Defining $\tb_{\norm}(s)$ as in the proposition  we construct by symbolic calculus $Z_{\norm}(\cdot)\in \cinf(]-\pi/2, \pi/2[; \Psi^{0}_{\b\s}(\Sigma_{\norm}))$ with \[
\sigma_{\rm pr}(Z_{\norm})(s)= \e^{-\int_{0}^{s}\sigma_{\rm pr}(\tb_{\norm})(\sigma)d\sigma},
\]
such that \[
\begin{cases}
 (\p_{s}+\tilde{R}_{\b\norm}(s)+  \tb_{\norm}(s))Z_{\norm}(s)= r_{-\infty}(s)\hbox{ in }]0, \pi/2[\times \Sigma_{\norm}, \\
 Z_{\norm}(0)= \one.
\end{cases}
\]
for $r_{-\infty}(\cdot)\in  \cinf(]-\pi/2, \pi/2[; \Psi^{-\infty}_{\b\:{\rm s}}(\Sigma_{\norm}))$.
 We obtain that
\[
\begin{cases}
\tD_{\b\norm}Z_{\norm}= \tilde{r}_{-\infty}(\cdot)\hbox{ in }\Omega^{+},\\
Z_{\norm}(0)=\one,\end{cases}
\]
where $\tilde{r}_{-\infty}(\cdot)\in \cinf(]-\pi/2, \pi/2[; \Psi^{-\infty}_{\b\:{\rm s}}(\Sigma_{\norm}))$.

Let us fix a cutoff function $\chi\in \coinf(]-\pi/2, \pi/2[_{s})$ with   $\chi=1$ near $s=0$.  Since 
\[
[\tD_{\b\norm}, \chi]= - 2\p_{s}\chi(\p_{s}+ \tilde{R}_{\b\norm}(s))- \p_{s}^{2}\chi-\tilde r_{\b\norm}(s)\p_s\chi,
\]
we have 
\[
\begin{cases}
\tD_{\b\norm}\chi Z_{\norm}=\chi\tilde{r}_{-\infty}-2(\p_s\chi)r_{-\infty}+ T(s),\\
\chi Z_{\norm}(0)=\one.
\end{cases}
\]
for $T(s)= (2\p_{s}\chi(s)\tb_{\norm}(s)-\p_{s}^{2}\chi(s)-\tilde r_{\b\norm}(s)\p_s\chi(s))Z_{\norm}(s)\in \cinf(]-\pi/2, \pi/2[; \Psi^{1}_{\b\:{\rm s}}(\Sigma_{\norm}))$.
It follows that the solution of the Dirichlet problem \eqref{e10.n2bb} in $\Omega^{+}$ is
\[
u= (\chi Z_{\norm}v- \tD_{\b\norm}^{-1}(z_{-\infty}(\cdot)+ R(\cdot))v),
\]
for
\[
\begin{aligned}
z_{-\infty}(s)&=\chi(s)\tilde r_{-\infty}(s)-2(\p_s\chi)(s)r_{-\infty}(s)\\
&\quad-\chi(-s)\tilde r_{-\infty}(-s)+2(\p_s\chi)(-s)r_{-\infty}(-s),
\end{aligned}\quad R(s)=T(s)-T(-s).
\] 
This gives
\beq\label{e10.2c}
\begin{aligned}
\tn_{\norm}^{+}v
&= - ((\p_{s}+\tilde{R}_{\b\norm}(s))Z_{\norm}v)(0)
 +(\p_{s}\tD_{\b\norm}^{-1}z_{-\infty}(\cdot)v)(0)\\
&\phantom{=\,}+ (\p_{s}\tD_{\b\norm}^{-1}R(\cdot)v)(0).
\end{aligned}
\eeq
The first term in the r.h.s. of \eqref{e10.2c} equals $(\tb_{\norm}(0)- r_{-\infty}(0))v$.  To estimate the third term, we will rely on the almost locality of $\tD_{\b\norm}^{-1}$. In fact  let us  fix two  cutoff functions $\chi_{1}, \chi_{2}\in \coinf(]-\pi/2, \pi/2[)$ with disjoint supports such that $\chi_{2}=1$ near $\supp \p_{s}\chi$ and $\chi_{1}=1$ near $s=0$. Then 
\[
(\p_{s}\tD_{\b\norm}^{-1}R(\cdot)v)(0)= (\p_{s}\chi_{1}(\cdot)\tD_{\b\norm}^{-1}\chi_{2}(\cdot)R(\cdot)v)(0).
\]
Therefore to complete the proof of (5) it remains to examine the operators
\[
r_{1,-\infty}\defeq \trho^{+}\tD_{\b\norm}^{-1}z_{-\infty}(\cdot),\quad r_{2,-\infty}\defeq \trho^{+}\chi_{1}(\cdot)\tD_{\b\norm}^{-1}\chi_{2}(\cdot)R(\cdot).
\]
To estimate $r_{i, -\infty}$, we  apply Lemma~\ref{lemma9.1} in the appendix to $\tD_{\b\norm}$ with $M_{\norm}= \rr^{+}_{\rho}\times\bS^{n-1}$.   We know that  $\tD_{\b\norm}$ is elliptic and $\rho^{r}\tD_{\b\norm}\rho^{-r}: H^{m, 0}_{\b}(\tM_{\norm})\tosim H^{m-2, 0}_{\b}(\tM_{\norm})$ for $r\in J$, see \eqref{e10.n101}. Therefore we can apply Lemma~\ref{lemma9.1} (3) and obtain that:
\[
\tD_{\b\norm}^{-1}= T+ s_{-\infty},
\]
for $T\in \Psi^{-2}_{\b\s}(\tM_{\norm})$, $s_{-\infty}\in B^{-\infty}_{J}(\tM_{\norm})$ and $\ad_{h}s_{-\infty}\in B^{-\infty}_{J, 0^{+}}(\tM_{\norm})$, see Subsect.~\ref{sec9.4} for notation. By $\Psi$DO calculus we obtain that $\chi_{1}T\chi_{2}\in \Psi^{-\infty}_{\b\s}(\tM_{\norm})$. From this it is straightfoward to verify that $r_{i, -\infty}\in B^{-\infty}_{J}(\Sigma_{\norm})$ with $\ad_{h}r_{i,-\infty}\in B^{-\infty}_{J, 0^{+}}(\Sigma_{\norm})$. This completes the proof of the proposition.  \qed

\subsection{\calde projectors}\label{sec10.5}

The Calder\'on projectors associated with $\tD_{k\b\norm}$ are defined by
\[
        \tc_{k\norm}^{\pm}f
        \defeq
        \mp\trho^{\pm}(\tD_{k\b\norm})^{-1}
        (\trho)^{*}\tsig f,
        \ \ f\in\coinf(\Sigma_{\norm};V_{k}\otimes\cc^2),
\]
where $\trho^{*}$ is the adjoint of $\trho$ defined in \eqref{e10.stup-trace} for the Hilbertian scalar products $(\cdot| \cdot)_{\tV(\tM_{\norm})}$ and $(\cdot| \cdot)_{\tV(\Sigma_{\norm})\otimes \cc^{2}}$.

Equivalently, they are the projections onto the Cauchy data of solutions of $\tD_{\b\norm}u= 0$ in $\Omega^{\pm}_{\norm}$.
 By Prop.~\ref{prop10.2}, the operator $\tn^+_{\norm}-\tn^-_{\norm}=\pm 2\tn^{\pm}_{\norm}$ is elliptic and invertible $H^{m,r}_{\b}(\Sigma_{\norm};V)\to H^{m-1,r}_{\b}(\Sigma_{\norm};V)$ for $m\in\mathbb{R}$ and $r\in J=\open{-\frac{n-2}{2},\frac{n-2}{2}}$. The \calde projectors are
\beq\label{e01.16d}
        \tc_{k\norm}^{\pm}
        =
        \mat{\mp(\tn_{k\norm}^{+}-\tn_{k\norm}^{-})^{-1}
                 \tn_{k\norm}^{\mp}}
             {\pm(\tn_{k\norm}^{+}-\tn_{k\norm}^{-})^{-1}}
             {\mp\tn_{k\norm}^{\pm}
                 (\tn_{k\norm}^{+}-\tn_{k\norm}^{-})^{-1}
                 \tn_{k\norm}^{\mp}}
             {\pm\tn_{k\norm}^{\pm}
                 (\tn_{k\norm}^{+}-\tn_{k\norm}^{-})^{-1}}
                 =
                 \frac{1}{2}
                 \mat{1}
             {(\tn_{k\norm}^{\pm})^{-1}}
             {\tn_{k\norm}^{\pm}}
             {1}.
\eeq
 Let
\beq\label{e01.16e}
        \cH^{m,r}_{\b}(\Sigma_{\norm};V_{k}\otimes\cc^2)
        \defeq
        H^{m,r}_{\b}(\Sigma_{\norm};V_{k})
        \oplus
        H^{m-1,r}_{\b}(\Sigma_{\norm};V_{k}).
\eeq
Then
\beq\label{e10.4}
        \tc_{k\norm}^{+}+\tc_{k\norm}^{-}=\one,
        \quad
        (\tc_{k\norm}^{\pm})^2=\tc_{k\norm}^{\pm},
\eeq
and the \calde projectors are bounded on $\cH^{m,r}_{\b}(\Sigma_{\norm};V_{k}\otimes\cc^2)$ for $m\in\mathbb{R}$, $r\in J$.

\subsubsection{Trace reversal}\label{sec10.5.1}
Let us denote by $I_{\norm}= I_{\b\norm}$ the trace reversal for $\reta$ defined as in \eqref{e6.-01}. 
We set 
\[
\widetilde{I}_{\b\norm}u= u- \frac{1}{4}(\treta| u)_{\tV_{2}}\treta,
\]
which is the trace reversal with respect to $\treta$.  We define $\widetilde{I}_{\b\Sigma\norm}$ by
\[
\widetilde{I}_{\b\Sigma\norm}\trho u_{2}\defeq \trho \widetilde{I}_{\b\norm}u_{2}.
\]
\begin{lemma}\label{lemma10.2}
\ben
\item  we have \[
\widetilde{I}_{\b\Sigma\norm}=\mat{ \widetilde{I}_{\b\norm}\traa{\Sigma}}{0}{0}{ \widetilde{I}_{\b\norm}\traa{\Sigma}}.
\]
\item  recall that $I_{\b\Sigma\norm}$ was defined in Subsect.~\ref{sec2.3.3}. We have: 
 \[
 \widetilde{I}_{\b\Sigma\norm}= \cF_{2\Sigma}\circ I_{\b\Sigma\norm}\circ \cF_{2\Sigma}^{-1},
 \]
 where $\cF_{2\Sigma}$ is defined in \ref{sec10.1.4}.
\item we have:
\[
[\tc_{2\norm}^{\pm}, \widetilde{I}_{\b\Sigma\norm}]=0.
\]
 \een
\end{lemma}
\proof Since $\widetilde{\nabla}\treta=0$ we deduce from \eqref{e10.n2b} that $\widetilde{\nabla}_{s}\treta\traa{\Sigma}=0$. This implies  (1). 
 (2) follows  from the fact that $\cF_{2}\reta\traa{\Sigma}= \treta\traa{\Sigma}$.

It remains to prove (3).  Clearly 
 $[\tD_{2\b\norm}, \widetilde{I}_{\b\norm}]= 0$. Since $\widetilde{\nabla}\treta=0$ we obtain that $[\widetilde{\nabla}_{s}, \widetilde{I}_{\b\norm}]=0$ and hence $[\tn^{\pm}_{2\b\norm}, \widetilde{I}_{\b\norm}]= 0$, which using (1) proves (3). 
 \qed
\subsection{Charge  self-adjointness of the \calde projectors}\label{sec10.6}
%Let
%\beq\label{e10.3b}
%        \kappa_k u(s,\cdot)=\tau_k u(-s,\cdot),
%\eeq
%where $\tau_k$ is defined in \eqref{e8.00}.  The symmetry of the normal model gives
%\beq\label{e10.3c}
%        \kappa_k\tD_{k\b\norm}\kappa_k=\tD_{k\b\norm}.
%\eeq
%

\begin{proposition}\label{prop10.3}
The Calder\'on projectors are self-adjoint for the normal charge $\rq_{k\norm}$ i.e.~
\[
        (\tc_{k\norm}^{\pm})^{*}\circ \rq_{k\b\norm}=\rq_{k\b\norm}\circ \tc_{k\b\norm}^{\pm}.
\]
%For $k=2$ they are also self-adjoint for the physical normal charge $\rq_{I, 2\norm}= \rq_{2\norm}\circ I_{\b\norm}$.
\end{proposition}
\proof
 Since $\tD_{k\b\norm}= \tD_{k\b\norm}^{*}= \kappa_{k}\circ \tD_{k\b\norm}\circ \kappa_{k}$, 
where $\kappa_{k}$ is defined in \ref{sec10.3.3} we can apply  \cite[Prop.~4.5]{GW3}.
We obtain that $\tc_{k\norm}^{\pm}$ are selfadjoint for the charge
\[
\begin{aligned}
\bar{f}\dual {\bf b}_{k}f
&= (f| \tsig\circ \mat{\kappa_{k\Sigma}}{0}{0}{-\kappa_{k\Sigma}}f)_{\tV_{k}(\Sigma_{\norm}\otimes \cc^{2})}\\
&= (f_{0}| \kappa_{k\Sigma}f_{1})_{\tV_{k}(\Sigma_{\norm})}+ (f_{1}| \kappa_{k\Sigma}f_{0})_{\tV_{k}(\Sigma_{\norm})}.
\end{aligned}
\]
Using the expression of $\kappa_{k\Sigma}$ recalled in \eqref{e10.n4c} and the tensor decomposition in \eqref{e2.10bb} we see that ${\bf b}_{k}= \rq_{k\b\norm}$, which proves the lemma. \qed
%The second follows then from the fact that $[\tc_{2\norm}^{\pm}, ]$ \qed
%The proof is the standard Calder\'on argument.  Let
%$u=-r^+(\tD_{\b\norm})^{-1}(\trho)^*\tsig f$ and
%$v=-r^+\kappa(\tD_{\b\norm})^{-1}(\trho)^*\tsig g$.  Both solve the homogeneous equation in $\Omega^+_{\norm}$.  Applying \eqref{e01.green1} gives
%\[
%        (\trho^+v|\tsig\trho^+u)_{\tV(\Sigma_{\norm})}=0.
%\]
%Using $\trho^+\kappa=\mat{\tau}{0}{0}{-\tau}\trho^-$ and
%$\mat{\tau}{0}{0}{-\tau}^{*}\tsig=-q$ gives
%$(\tc_{\norm}^{-})^{*}q\tc_{\norm}^{+}=0$.  The opposite mixed term is identical, and the assertion follows from $\tc_{\norm}^{+}+\tc_{\norm}^{-}=\one$.  The physical charge statement follows from \eqref{e10.-2}.\qed

\subsection{Gauge invariance of the \calde projectors}\label{sec10.7}
If $u_{1}\in \cinf(\bar{\Omega^{\pm}}; V_{1})$ with $\tD_{1\b\norm}u_{1}\traa{\Sigma}= 0$ we define $\tK_{\b\Sigma\norm}\trho_{1}^{\pm}u_{1}$ by:
\[
\trho_{2}^{\pm}\tK_{\b\norm}u_{1}\eqdef
        \tK_{\b\Sigma\norm}\trho_{1}^{\pm}u_{1}.
\]
This defines $\tK_{\b\Sigma\norm}$ as a differential operator on $\Sigma_{\norm}$ (with either choice of $\pm$ giving the same operator), so, corresponding to the orders of the components of $K_{\b\norm}$, it acts continuously $\cH^{m,r}_{\b}(\Sigma_{\norm};V_{1}\otimes\cc^2)\to \cH^{m-1,r+1}_{\b}(\Sigma_{\norm};V_{2}\otimes\cc^2)$.

\begin{proposition}\label{prop10.4}
The Calder\'on projectors satisfy
\beq\label{e8.calde-gauge}
        \tc_{2\norm}^{\pm}\circ \tK_{\b\Sigma\norm}
        =
        \tK_{\b\Sigma\norm}\circ \tc_{1\norm}^{\pm}.
\eeq
\end{proposition}

\proof
Fix $r\in J\cap(J-1)$ and first take $f\in \cH^{\infty,r}_{\b}(\Sigma_{\norm}; V_{1}\otimes \cc^{2})$ in the range of $\tc_{1\norm}^{\pm}$,  and let $u_{1}$ be the corresponding solution of $\tD_{1\b\norm}u_{1}=0$ in $\Omega^{\pm}$ with $\trho_{1}^{\pm}u_{1}= f$.  By \eqref{e10.n-77} we have  $\tD_{2\b\norm}\tK_{\b\norm}u_{1}=0$ in $\Omega^{\pm}$. Mapping properties of  Poisson operators imply $u_1\in H^{\infty,r}_{\b}(\Omega^{\pm};V_1)$ and $\tK_{\b\norm}u_1\in H^{\infty,r+1}_{\b}(\Omega^{\pm};V_2)$ (both weights belong to $J$). Since         $\trho_{2}^{\pm} \tK_{\b\norm}u_{1}= \tK_{\b\Sigma\norm}\trho_{1}^{\pm}u_{1}$ we obtain that $\tK_{\b\Sigma\norm}f\in \Ran \tc_{2\norm}^{\pm}$, so $\tK_{\b\Sigma\norm}$ maps  $\Ran \tc_{1\norm}^{\pm}$ into  $\Ran \tc_{2\norm}^{\pm}$.  Applying the same argument to $\one-\tc_{1\norm}^{\pm}= \tc_{1\norm}^{\mp}$ we obtain that  $\tK_{\b\Sigma\norm}$ maps  $\Ker \tc_{1\norm}^{\pm}$ into  $\Ker \tc_{2\norm}^{\pm}$ which proves \eqref{e8.calde-gauge}. The identity extends to $\cH^{m,r}_{\b}$ by continuity.\qed
\subsection{Lorentzian \calde projectors}\label{sec10.8}
In Subsect.~\ref{sec8.3} we will need to consider the Lorentzian versions of $\tc_{k\norm}^{\pm}$ defined as:
\beq\label{e01.16f}
c_{k\norm}^{\pm}\defeq \cF_{k\Sigma}^{-1}\tc_{k\norm}^{\pm}\cF_{k\Sigma},
\eeq
where $\cF_{k\Sigma}$ are defined in \ref{sec10.1.4}. We recall that  $\rq_{I, 2\b\norm}= \rq_{2\b\norm}\circ I_{\b\Sigma\norm}$, see \eqref{e2.charge-c}.

\begin{proposition}\label{prop10.4b}
One has
\ben
\item $(c_{k\norm}^{\pm})^{*}\circ \rq_{k\b\norm}=\rq_{k\b\norm}\circ c_{k\norm}^{\pm}$.
\item $(c_{2\norm}^{\pm})^{*}\circ \rq_{I,2\b\norm}=\rq_{I,2\b\norm}\circ c_{2\norm}^{\pm}$.
\item $c_{2\norm}^{\pm}\circ K_{\b\Sigma\norm}=K_{\b\Sigma\norm}\circ c_{1\norm}^{\pm}$.
\een
\end{proposition}

\proof  Let us first prove (1), removing the subscript $k$. We have $\cF_{\Sigma}^{*}\rq_{\b\norm} \cF_{\Sigma}= \rq_{\b\norm}$ hence using Prop.~\ref{prop10.3} we have:
\[
\begin{aligned}
\rq_{\b\norm}c^{\pm}_{\norm}
&= \rq_{\b\norm}\cF_{\Sigma}^{-1}\tc_{\norm}^{\pm}\cF_{\Sigma}= \cF_{\Sigma}^{*}\rq_{\b\norm} \tc^{\pm}_{\norm}\cF_{\Sigma}\\
&= \cF_{\Sigma}^{*}(\tc^{\pm}_{\norm})^{*}\rq_{\b\norm} \cF_{\Sigma}\\
&= \cF_{\Sigma}^{*}(\tc^{\pm}_{\norm})^{*}(\cF_{\Sigma}^{-1})^{*}\rq_{\b\norm} = (c^{\pm}_{\norm})^{*}\rq_{\b\norm}.
\end{aligned}
\]
From Lemma~\ref{lemma10.2} we deduce then that $[c_{2\norm}^{\pm}, I_{\b\Sigma\norm}]=0$ which proves (2).
From \eqref{e10.n-4} we obtain that $K_{\b\Sigma\norm}= \cF_{2\Sigma}^{-1}\circ \widetilde{K}_{\b\Sigma\norm}\circ \cF_{1\Sigma}$, so (3) follows from Prop.~\ref{prop10.4}. \qed

\section{Hadamard projectors for the reduced operators}\label{sec8}\init
In this section we construct Hadamard projectors for the reduced operators
\[ 
        D_{k\b}=-(\nabla_{t})^{\star}\nabla_{t}+a_{ k\b}(t)
\]
introduced in Subsect.~\ref{sec6.2}.  The Cauchy data map $\varrho_{k\b}$ adapted to $D_{k\b}$ is defined in \eqref{e2.10bc}.

The construction has two steps.  First we construct symbolic Hadamard projectors by factorizing $D_{k\b}$, using the $\Psi^{m}_{{\rm s}\: 1}$-calculus  recalled in Subsect.~\ref{sec9.3}.  Then we perform the transplant at the normal operator, in the sense of Subsect.~\ref{sec9.4}: the normal operators of the projectors are replaced by the Calder\'on projectors of the Wick-rotated normal operator.  This second step is needed for the gauge-invariance modulo compact errors used in Sect.~\ref{sec11}. In this section we often suppress the tensor degree $k$.

For  $\b$-pseudodifferential classes, residual ideals, and normal operators we use the notation in  Appendix~\ref{sec9}, in particular Subsects.~\ref{sec9.3} and \ref{sec9.4}. We will often refer to the range of weights $J=\open{-\frac{n-2}{2},\frac{n-2}{2}}$, whose importance stems from Prop.~\ref{prop10.2} (4) and (5).
\subsection{Factorization of $D_{\b}$}\label{sec8.2}
\begin{proposition}\label{prop8.3}
 \ben
 
 \item There exists  $b\in \cinf(I; \Psi^{1}_{{\rm s}\: 1}(\Sigma; V))$   with $\sigma_{\rm pr}(b)(t)= (\xi\dual \rh_{\b}^{-1}(t)\xi)^{\12}\one_{V}$ such that
 \beq\label{e8new.1}
 D_{\b}= (-(\nabla_{t})^{\star}+  \i b(t))(\nabla_{t}- \i b(t))+ r_{-\infty}(t),\ \ r_{-\infty}\in \cinf(I; \Psi^{-\infty}_{\rm s\: 1}(\Sigma;V)).
 \eeq
 Equivalently $b$ solves 
 \beq\label{e8new.2}
-\i\p_{t}b(t) -[R_{\b}(t), \i b(t)]- \i r_{\b}(t)b(t)+ b^{2}(t)- a_{\b}(t)=0\hbox{ mod } \cinf(I; \Psi^{-\infty}_{\rm s\: 1}(\Sigma;V)).
 \eeq

  \item  $b$ is unique modulo $\cinf(I, \Psi^{-\infty}_{\rm s\: 1}(\Sigma;V))$.
  \item the operator $-b^{\star}$ is also a solution of \eqref{e8new.2}.
 \item If $b_{\norm}(t)$ is the normal operator of $b(t)$ one has
 \[
 b_{\norm}(0)= \cF^{-1}\tn_{\norm}^{+}\cF+ r_{-\infty \norm}, 
 \]
 \[
r_{-\infty\norm}\in B^{-\infty}_{J}(\Sigma_{\norm}; V_{\norm}),\quad \ad_{h}r_{-\infty\norm}\in B^{-\infty}_{J, 0^{+}}(\Sigma_{\norm}),
\]
 for any $h\in \coinf(\rr^{+}_{\rho})$ with $h=1$ near $\rho=0$.
 \een
\end{proposition}

\proof
One has $D_{\b}= \nabla_{t}^{2}+ r_{\b}(t)\nabla_{t}+ a_{\b}(t)$, so using that $\nabla_{t}= \p_{t}+ R_{\b}(t)$ an easy computation shows that \eqref{e8new.1} and \eqref{e8new.2} are equivalent. 

Let $\epsilon\in\cinf(I;S^1_{\cl\: 1}({}^{\b}T^*\Sigma;L(V)))$ be a classical elliptic symbol whose principal part is
\[
        \epsilon_{\rm pr}(t,y,\xi)=(\xi\dual\rh_{\b}^{-1}(t)\xi)^{1/2}\one_V.
\]
Using the symbolic fractional powers from Subsect.~\ref{sec9.2b.3}, we choose $\epsilon$ so that
\[
        \Op(\epsilon)^2-a_{\b}(t)\in\cinf(I;\Psi^{-\infty}_{\rm s\: 1}(\Sigma;V)).
\]
We also choose a symbolic inverse $(2\epsilon)^{(-1)}\in\cinf(I;S^{-1}_{\cl \: 1}({}^{\b}T^*\Sigma;L(V)))$, i.e.
\[
        (2\epsilon)^{(-1)}\moyal 2\epsilon=1
        \quad \hbox{mod }S^{-\infty}.
\]
We look for
\[
        b=\Op(\epsilon)+\Op(c),
        \ \ c\in\cinf(I;S^0_{\cl \: 1}({}^{\b}T^*\Sigma;L(V))).
\]
Substituting this expression into \eqref{e8new.2} gives, modulo
$\cinf(I;\Psi^{-\infty}_{\rm s\: 1}(\Sigma;V))$:
\[
\begin{aligned}
\Op(c)={}&(2\Op(\epsilon))^{(-1)}
        \big(\i\p_t\Op(\epsilon)+\i r_{\b} \Op(\epsilon)+ [R_{\b},\i\Op(\epsilon)]\big)\\
&+(2\Op(\epsilon))^{(-1)}
        \big(\i\p_t\Op(c)+\i r_{\b}\Op(c)+ [R_{\b},\i\Op(c)]
        +[\Op(\epsilon),\Op(c)]-\Op(c)^2\big).
\end{aligned}
\]
Equivalently, on the level of full symbols we obtain
\beq\label{e8.fixedpoint-symbol}
        c=c_0+F(c)\quad\hbox{mod }S_{{\rm cl}\: 1}^{-\infty},
\eeq
where
\[
        c_0=(2\epsilon)^{(-1)}\moyal
        \big(\i\p_t\epsilon+\sigma(\i r_{\b}\Op(\epsilon)+ [R_{\b},\i\Op(\epsilon)])\big),
\]
and
\[
F(c)=(2\epsilon)^{(-1)}\moyal
\big(\i\p_t c+\sigma(\i r_{\b}\Op(c)+ [R_{\b},\i\Op(c)]+ [\Op(\epsilon), \Op(c)])
-c\moyal c\big).
\]
Here $\sigma$ denotes any choice of  a full symbol representative of the indicated operator. A different representative changes the right hand side of \eqref{e8.fixedpoint-symbol} only by an element of $S_{{\rm cl}\: 1}^{-\infty}$.

The map $F$ improves symbolic order by one: if
$c-c'\in\cinf(I;S^{-m}_{\cl \: 1})$, then
\[
        F(c)-F(c')\in\cinf(I;S^{-m-1}_{\cl \: 1}).
\]
Indeed, the factor $(2\epsilon)^{(-1)}$ has order $-1$, while the expression in parentheses has order at most $-m$. Starting with $c^{(0)}=c_0$ and defining
$c^{(j+1)}=c_0+F(c^{(j)})$, one obtains
\[
        c^{(j+1)}-c^{(j)}\in\cinf(I;S^{-j-1}_{\cl \: 1}).
\]
Borel summation gives $c\in\cinf(I;S^0_{\cl \: 1})$ solving \eqref{e8.fixedpoint-symbol}. Consequently
\beq\label{e8.8b}
      -\i\p_{t}b(t) -[R_{\b}(t), \i b(t)]- \i r_{\b}(t)b(t)+ b^{2}(t)- a_{\b}(t)
        =r_{-\infty}
        \in\cinf(I;\Psi^{-\infty}_{\rm s\: 1}(\Sigma;V)).
\eeq
This proves (1). The same argument applied to the difference of two solutions proves (2). Since $D_{\b}= D_{\b}^{\star}$ we obtain that $ -b^{\star}$ solves \eqref{e8new.1} hence also \eqref{e8new.2} which proves (3).

It remains to prove (4). 
The remainder in \eqref{e8.8b} has a normal operator by Lemma~\ref{lemma9.2}. Taking normal operators gives:
\beq\label{e8.10b}
        -\i\p_t b_{\norm}(t)-[R_{\b\norm}(t),\i b_{\norm}(t)]-\i r_{\b\norm}(t)b_{\norm}(t)
        +b_{\norm}^{2}(t)-a_{\b\norm}(t)
        \in\cinf(I;\Psi^{-\infty}_{{\rm s}\: \norm}(\Sigma_{\norm};V_{\norm})).
\eeq
Moreover $b_{\norm}(t)$ has principal symbol
\[
(\xi\dual\boldsymbol{\delta}_{\b}^{-1}(t)\xi)^{1/2}\one_{V_{\norm}}.
\]
We now compare the solution $b_{\norm}(t)$ of \eqref{e8.10b} with the solution $\tb_{\norm}(s)$ of the analogous equation \eqref{e10.2b} appearing in Prop.~\ref{prop10.2}.

Since the metric $\reta$ is real analytic in $t$ near $t=0$, all the terms in the symbolic expansion of $b_{\norm}(t)$ are real analytic in $t$ near $t=0$. We apply Wick rotation coefficientwise to this formal expansion. Using the notation introduced in \ref{sec10.1.4} and Lemma~\ref{lemma10.0}, we obtain that the formal operator $\cF b_{\norm}(\i s)\cF^{-1}$ solves \eqref{e10.2b}. By uniqueness in the formal symbolic calculus, the symbols of this operator and $\tb_{\norm}(s)$ agree to all orders. Evaluating at $s=0$, we obtain
\[
b_{\norm}(0)=\cF^{-1}\tb_{\norm}(0)\cF
\hbox{ modulo }\Psi^{-\infty}_{{\rm s}\:\norm}(\Sigma_{\norm}).
\]
By \eqref{e10.2bc} in Prop.~\ref{prop10.2} we obtain that
\[
b_{\norm}(0)= \cF^{-1}\tilde{n}^{+}_{\norm}\cF+ r_{-\infty\norm},  
\]
where the remainder has the stated properties. This completes the proof. \qed

\subsection{Hadamard projectors for $D_{\b}$}\label{sec8.3}
\subsubsection{Summary of the construction}
From $b$ constructed in Prop.~\ref{prop8.3} one constructs Hadamard projectors.  We  set initially
\[
       b^{+}= b(0),\quad b^{-}=-b^{\star}(0)
\]
modulo residual terms.  After changing $b$ by a residual term if necessary, we may assume that $b^{+}-b^{-}$ is invertible in the calculus.

The Hadamard projectors are
\beq\label{e8.9bb}
        c^{\pm}=
        \mat{\mp(b^{+}-b^{-})^{-1}b^{\mp}}
             {\pm(b^{+}-b^{-})^{-1}}
             {\mp b^{+}(b^{+}-b^{-})^{-1}b^{-}}
             {\pm b^{\pm}(b^{+}-b^{-})^{-1}}.
\eeq
They project  Cauchy data onto the graphs
\[
        \Ran\binom{\one}{b^{+}},
        \quad
        \Ran\binom{\one}{b^{-}},
\]
 If $u$ solves
$(\nabla_t-\i b^{\pm})u=0$ with $u(0)=f$, then
$\varrho_{\b}u=\binom{f}{b^{\pm}f}$.

In our situation it is essential that the Hadamard projectors $c_k^{\pm}$, $k=1,2$, satisfy the gauge invariance identity asymptotically near $\{\rho=0\}$.  We therefore transplant the normal operators so that they are exactly the Lorentzian Calder\'on projectors $c_{k\norm}^{\pm}$ defined in Subsect.~\ref{sec10.8} , which satisfy the gauge identity in Prop.~\ref{prop10.4b} (3).
\subsubsection{Sobolev spaces of Cauchy data}
We set for $m,r\in\rr$
\beq\label{e8.9b}
        \cH^{m,r}_{\b}(\Sigma;V\otimes\cc^2)
        \defeq
        H^{m,r}_{\b}(\Sigma;V)
        \oplus
        H^{m-1,r}_{\b}(\Sigma;V).
\eeq
We denote by $\cU_{\b}: \coinf(\Sigma; V\otimes \cc^{2})\to \Ker_{\sc}D_{\b}$ the solution of the Cauchy problem for $D_{\b}$:
\beq\label{e0b.3}
\begin{cases}
D_{\b}\cU_{\b}= 0, \\
\varrho_{\b}\cU_{\b}= \one,
\end{cases}
\eeq
where the Cauchy data map $\varrho_{\b}$ is defined in \eqref{e2.10bc}.
\begin{proposition}\label{prop8.4}
There exist a pair of operators $c^{\pm}$ with the following properties:
\ben
\item $c^{+}+c^{-}=\one$, and $c^{\pm}$ are bounded projections on $\cH^{p,r}_{\b}(\Sigma;V\otimes\cc^2)$ for $p\in\rr$, $r\in J$.
\item $\rq_{\b}\circ c^{\pm}=(c^{\pm})^*\circ\rq_{\b}$, i.e.~$c^{\pm}$ are self-adjoint for the Lorentzian charge.
\item $\pm c^{\pm *}\circ\tilde{\rq}_{\b}\circ c^{\pm}\geq 0$, where the Euclidean charge $\tilde{\rq}_{\b}$ is defined in \ref{sec8.1.2}.
\item If $c^{\pm}_{\alpha\beta}$  for $0\leq\alpha,\beta\leq1$ are the matrix entries of $c^{\pm}$ then
\[
        c^{\pm}_{\alpha\beta}\in\Psi^{\alpha-\beta}_{J}(\Sigma;V).
\]
\item $\WF(\cU_{\b}c^{\pm})'\subset(\cN^{\pm}\cup\cF)\times T^*\Sigma$, over $\mathcal{O}\times \Sigma$, where $\mathcal{O}$ is a neighborhood of $\Sigma$ in $M$ and $\cF$ is conic with  $\cF\cap\cN= \emptyset$.
\item If $(c^{\pm})_{\norm}$ is the normal operator of $c^{\pm}$, then
\[
        (c^{\pm})_{\norm}=c^{\pm}_{\norm},
\]
where $c^{\pm}_{\norm}= \cF^{-1}_{\Sigma}\tc^{\pm}_{\norm}\cF_{\Sigma}$ are the Lorentzian Calder\'on projectors defined in Subsect.~\ref{sec10.8}.
\een
\end{proposition}

\proof
We denote the operator $b(0)$  in Prop.~\ref{prop8.3} simply by $b$. We first modify $b$ by a residual term so that its normal operator $b_{\norm}$ is exactly $\cF^{-1}\tn^{+}_{\norm}\cF$.  Since $b\in\Psi^{1}_{{\rm s\: 1}}(\Sigma;V)\subset \Psi^{1}_{\rr}(\Sigma; V)$, Lemma~\ref{lemma9.2b} gives
\[
        b=\chi_{\infty}b_{\norm}\chi_{\infty}
        +x^{1/2}r_1x^{1/2}+r_{-\infty},
\]
where $r_1\in\Psi^1_{\rm s}(\Sigma;V)$, $r_{-\infty}\in B^{-\infty}_{J,0^+}(\Sigma;V)$, and $\chi_{\infty}=1$ near $\p\Sigma$.  We set
\beq\label{e8.10a}
n_{\norm}^{+}= \cF^{-1}\tn^{+}_{\norm}\cF,
\eeq
and 
\[
        \hat b\defeq
        \chi_{\infty}n_{\norm}^{+}\chi_{\infty}
        +x^{1/2}r_1x^{1/2}+r_{-\infty}.
\]
By Prop.~\ref{prop8.3} (4) we have $b-\hat b\in B^{-\infty}_{J}(\Sigma;V)$ $\hat b_{\norm}=n^{+}_{\norm}$.

From Prop.~\ref{prop10.2} (2) we have $n^{+}_{\norm}= (n^{+}_{\norm})^{*}= (n^{+}_{\norm})^{\star}$. Therefore the operators 
$\hat b+\hat b^{\star}$ and $\hat b+\hat b^{*}$ have normal operator $2n^{+}_{\norm}$ and principal symbol $2(\xi\rh_{\b}^{-1}\xi)^{1/2}\one_V$.  Since by Prop.~\ref{prop10.2} (3) $n^{+}_{\norm}: H^{m, r}_{\b}(\Sigma; V)\tosim H^{m-1, r}_{\b}(\Sigma; V)$, for $m\in \rr,\ r\in J$, Lemma~\ref{lemma9.4} shows that $\hat{b}+ \hat{b}^{*}, \hat{b}+ \hat{b}^{\star}: H^{m, r}_{\b}(\Sigma; V)\tosim H^{m-1, r}_{\b}(\Sigma; V)$ are Fredholm.

From  Prop.~\ref{prop10.2} (3) we know that $2n^{+}_{\norm}\geq C>0$ for $(\cdot|\cdot)_{\tV(\Sigma_{\norm})}$. As in the proof of 
 Lemma~\ref{lemma9.6} we can add a finite-rank operator $s_{-\infty}=s_{-\infty}^{*}\geq0$, with smooth kernel supported away from $\{\rho=0\}$, so that
\[
        \Re(\hat b+s_{-\infty})\geq C>0
\]
for the Hilbertian scalar product $(\cdot| \cdot)_{\tV(\Sigma)}$.  Applying \eqref{e2.12e} we obtain that $\Re(\hat{b}+ \hat{b}^{\star})= \Re\hat{b}+ \tau^{-1}\Re\hat{b}\tau$. Since $\tau^{*}= \tau^{-1}$ we also have that
\[
\Re(\hat{b}+ s_{-\infty}+ (\hat{b}+ s_{-\infty})^{\star})\geq 2C>0.
\]
We replace $\hat{b}$ by $\hat{b}+ s_{-\infty}$ which we again denote by $\hat{b}$. We obtain as in the proof of Lemma~\ref{lemma9.6} that 
 $\hat{b}, \hat{b}+ \hat{b}^{\star}, \hat{b}+ \hat{b}^{*}: H^{p, r}_{\b}(\Sigma; V)\tosim H^{p-1, r}_{\b}(\Sigma; V)$ for $p\in \rr,\ r\in J$.
 Moreover $\hat{b}+ \hat{b}^{\star}, \hat{b}$ with domain $H^{1, 0}_{\b}(\Sigma; V)$ are $m$-accretive for  $(\cdot| \cdot)_{\tV(\Sigma)}$ and since  the kernel of $s_{-\infty}$ vanishes near $\{\rho=0\}$ we still have  $\hat{b}_{\norm}= n^{+}_{\norm}$.

Set
\beq\label{e8.10}
        b^{+}=\hat b,
        \quad
        b^{-}=-\hat b^{\star}.
\eeq
Then $b^{+}-b^{-}:H^{p,r}_{\b}(\Sigma;V)\tosim H^{p-1,r}_{\b}(\Sigma;V)$ for $p\in\rr$, $r\in J$.
Let
\[
        S\defeq\i^{-1}\mat{\one}{-\one}{b^{+}}{-b^{-}}(b^{+}-b^{-})^{-1},
        \quad
        S^{-1}=\i\mat{-b^{-}}{\one}{-b^{+}}{\one}.
\]
Since $b^{+}-b^{-}=\hat{b}+ \hat{b}^{\star}$ is $m$-accretive, we can define $d=(b^{+}-b^{-})^{1/2}$, with $d=d^{\star}$.  Recall that in \eqref{e2.10bee} we defined the charge operator $q$ by:
\[
\bar{f}\dual \rq_{\b} f= (f| qf)_{\tV(\Sigma)\otimes \cc^{2}},\quad q= \mat{0}{\tau}{\tau}{0},
\]
where $\tau$ is defined in \eqref{e2.12b}.

  As in \cite[5.2.9]{GMW}, if
$T=S\circ\mat{d}{0}{0}{d}$, then using that $d^{\star}(b^{+}- b^{-})^{-1}d=1$ we obtain that:
\beq\label{e8.10bb}
        T^{*}qT=\mat{\tau}{0}{0}{-\tau}.
\eeq
  We set then for $\pi^{+}= \mat{1}{0}{0}{0}$ and $\pi^{-}= \mat{0}{0}{0}{1}$:
 \beq\label{e01.21b}
c^{\pm}\defeq T\pi^{\pm}T^{-1}=  \mat{\mp(b^{+}- b^{-})^{-1}b^{\mp}}{\pm(b^{+}- b^{-})^{-1}}{\mp b^{\pm}(b^{+}- b^{-})^{-1}b^{\mp}}{\pm b^{\pm}(b^{+}- b^{-})^{-1}}.
\eeq
We obtain  immediately that $c^{\pm}$ are bounded projections on $\cH^{p, r}_{\b}(\Sigma; V\otimes \cc^{2})$ for $p\in \rr,\ r\in J$. Moreover $c^{+}+ c^{-}= \one$ and  using \eqref{e8.10bb} we have also $\rq_{\b}\circ c^{\pm}= (c^{\pm})^{*}\circ \rq_{\b}$, i.e.~(1) and (2) of the proposition hold. 

 To check (3) we follow verbatim the proof in \cite[Lem.~3.1]{G2}. We see that (3) is equivalent to  the fact $\hat{b}+ \hat{b}^{*}\geq 0$ for  $(\cdot| \cdot)_{\tV(\Sigma)}$, which we have already proved.

Statement (4) follows from pseudodifferential calculus and the mapping properties of $(b^{+}-b^{-})^{-1}$.  To prove (5) we use that if $\varrho_{\b}u\in \Ran c^{\pm}$ then $(\nabla_{t}- \i b)u$ is smooth away from $\rho=0$, see \cite[Prop.~5.4]{GMW} for details.

Finally recall that $b^{\pm}_{\norm}=n^{\pm}_{\norm}= \cF^{-1}\tn^{\pm}_{\norm}\cF$.  Substituting these normal operators into \eqref{e01.21b}  and using \eqref{e01.16d} we obtain that $(c^{\pm})_{\norm}= \cF_{\Sigma}^{-1}\tc^{\pm}_{\norm}\cF_{\Sigma}$, which proves (5).\qed

\subsubsection{Hadamard projectors and physical charge}\label{sec8.3.1}
For $k=2$ the Lorentzian physical charge is $\rq_{I,2\b}=\rq_{2\b}\circ I_{\b\Sigma}$, see Subsect.~\ref{sec2.3.3}.   We recall that $[D_{2\b},I_{\b}]=0$ and from Lemma~\ref{lemma10.2} we have $[\cF_{2}^{-1}\tn_{2\norm}^{+}\cF_{2}, I_{\b}]=0$.  From these two facts we obtain that $I_{\b}\circ b_{2}\circ I_{\b}$ also satisfies the conditions in Prop.~\ref{prop8.3}. Therefore we may replace  $b_{2}$ by $\12 (b_{2}+ I_{\b}b_{2}I_{\b})$ and assume that 
        $[b_2,I_{\b}]=0$.
Consequently the Hadamard projectors satisfy $[c_2^{\pm},I_{\b\Sigma}]=0$
and hence are self-adjoint for the physical charge:
\beq\label{e8.12b}
        (c_2^{\pm})^{*}\circ\rq_{I,2\b}=\rq_{I,2\b}\circ c_2^{\pm}.
\eeq
\subsection{Gauge invariance modulo compact errors}\label{sec8.4}
Because of the residual choices in Props.~\ref{prop8.3} and \ref{prop8.4}, the projectors need not satisfy the exact identity
\beq\label{e8.12}
        c_2^{\pm}K_{\b\Sigma}=K_{\b\Sigma}c_1^{\pm}.
\eeq
The next proposition gives the weaker property used later.  

\begin{proposition}\label{prop8.5bis}
The map
\[
        c_2^{\pm}K_{\b\Sigma}-K_{\b\Sigma}c_1^{\pm}
\]
is smoothing on $\Sigma$ and compact from
$\cH^{p,r}_{\b}(\Sigma;V_1\otimes\cc^2)$ to
$\cH^{p-1,r+1}_{\b}(\Sigma;V_2\otimes\cc^2)$ for $p\in\rr, r\in J\cap(J-1)$.
\end{proposition}

\proof
We omit the vector-bundle symbols and prove only the $+$ case. From \eqref{e7.new3} we have:
\[
        (K_{\b\Sigma})_{\alpha\beta}
        \in
        \rho\Diff^{1+\alpha-\beta}_{\b,1}(\Sigma),
        \quad
        K_{\b\Sigma}:
        \cH^{p,r-1}_{\b}(\Sigma)\to\cH^{p-1,r}_{\b}(\Sigma),\ \ p,\ r\in \rr.
\]
From Prop.~\ref{prop8.4} (4)  we know that $(c_{k}^{\pm})_{\alpha\beta}\in \Psi^{\alpha- \beta}_{\b, J}(\Sigma)$, so by Lemma~\ref{lemma9.2b} (5) we have:
\[
\rho (c_{1}^{+})_{\alpha\beta}\rho^{-1}\in \Psi^{\alpha- \beta}_{\b, J+1}(\Sigma),\quad \rho c_{1}^{+}\rho^{-1}: \cH^{p, r'}_{\b}(\Sigma)\to \cH^{p, r'}_{\b}(\Sigma),
\]
for $p\in \rr,\ r'\in J+1$.

Therefore
\[
\begin{aligned}
& (c_{2}^{+}\circ K_{\b\Sigma}- K_{\b\Sigma}\circ c_{1}^{+})\rho^{-1}\in \Psi^{\infty}_{J\cap(J+1)}(\Sigma), \\
&( c_{2}^{+}\circ K_{\b\Sigma}- K_{\b\Sigma}\circ c_{1}^{+})\rho^{-1}: \cH^{p, r'}_{\b}(\Sigma)\to \cH^{p-1, r'}_{\b}(\Sigma),\ \ p\in \rr,\ r'\in J\cap(J+1).
\end{aligned}
\]

Let $f_1\in\cE'(\Sigma;V_1\otimes\cc^2)$ and set
\[
        u_1=\cU_{1\b}(\cdot,0)c_1^{+}f_1,
        \quad
        u_2=K_{\b}u_1,
        \quad
        f_2=K_{\b\Sigma}c_1^{+}f_1.
\]
By the Hadamard property, $\WF(u_1)\subset\cN^{+}$, hence $\WF(u_2)\subset\cN^{+}$.  On the other hand,
$\WF(\cU_{2\b}(\cdot,0)c_2^{-}f_2)\subset\cN^{-}$.  Since $\cN^{+}\cap\cN^{-}=\emptyset$, the operator $c_2^{-}K_{\b\Sigma}c_1^{+}$ is smoothing.  This implies that
$(c_2^{+}K_{\b\Sigma}-K_{\b\Sigma}c_1^{+})\rho^{-1}$ is smoothing and hence has vanishing principal symbol.

By Prop.~\ref{prop8.4} (6) its normal operator is $(c_{2\norm}^{+}K_{\b\Sigma\norm}- K_{\b\Sigma\norm}c_{1\norm}^{+})\rho^{-1}$ which vanishes by Prop.~\ref{prop10.4b}. Therefore
\[
(c_2^{+}K_{\b\Sigma}-K_{\b\Sigma}c_1^{+})\rho^{-1}: \cH^{p, r'}_{\b}(\Sigma)\to \cH^{p-1, r'}_{\b}(\Sigma)\hbox{ is compact for }p\in \rr,\ r'\in J\cap(J+1).
\]
Setting $r= r'-1$ we obtain the proposition.  
  \qed

\section{Gauge fixing}\label{sec11}\init
This section contains the main part of the construction of Hadamard states for linearized gravity in $(\rr^{n}, \rg)$. It relies on the {\em full gauge fixing} method, which is explained in more detail in \ref{sec0.3.1}. 

After imposing an additional gauge fixing condition, the problem amounts to inverting a system of pseudodifferential equations on $\Sigma$, which is elliptic in the $\b$-calculus framework and hence Fredholm.  The ellipticity in the $\b$-calculus framework  relies on the gauge invariance modulo compact errors of the Hadamard projectors. 

The analysis of full gauge fixing is first done in Subsect.~\ref{sec11.3} in what we call the {\em regular case}, namely, if the Fredholm problem is invertible. The more complicated {\em singular case} is studied in Subsect.~\ref{sec11.4}.

%\subsection{Introduction}\label{sec11.0}
%We fix throughout this section a weight in the gauge-fixing interval.  The interval $J$ is the one introduced in \eqref{e8.J}, and we repeat it here for convenience:
%\beq\label{e11.Jgf}
%        r\in J_{\rm gf}\defeq J\cap(J+1),
%        \quad J=\open{-\frac n2+1,\frac n2-1}.
%\eeq
%Below we take $1-r_0<r<1$ with $r_0>0$ sufficiently small.  Thus the two-tensor data are measured at weight $r$, while the gauge parameter is measured at weight $r-1$, which is close to zero.    Compactly supported Cauchy data belong to all these weighted spaces, so the final compactly supported phase-space statements are unaffected by this auxiliary choice of $r$.
%
%The charges $\rq_{k},\rq_{I,2}$ are well defined on the two-tensor spaces $\cH^{p,r}_{\b}(\Sigma; V_{2}\otimes \cc^{2})$ for $p\geq\12$, since $r>0$.  The point of choosing $r\in J_{\rm gf}$ is that $c_2^{\pm}$ is used at weight $r$, $c_1^{\pm}$ is used at weight $r-1$, and the compact gauge-error estimate Prop.~\ref{prop8.5bis} is available.
%
\subsection{Microlocal TT-synchronous gauge condition}\label{sec11.1}
We denote by $\hat{\ell}: \cinf(\Sigma; V_{2})\to \cinf(\Sigma; V_{1})$ the map:
 \[
 \hat{\ell}u_{2}\defeq \col{c^{-2}(n-2)^{-1}(\rg(0)| u_{2})_{V_{2}^{\rg(0)}}}{ 2 u_{2t\Sigma}},
 \]
  where $c$ is the conformal factor and on the RHS~we identify as in   \ref{sec2.3.2}, $(v_{t}, v_{\Sigma})\in \cinf(\Sigma; \cc\oplus T^{*}\Sigma)$ with $v_{t}rdt+ v_{\Sigma}\in \cinf(\Sigma; V_{1})$. The choice of the  coefficients $c^{-2}(n-2)^{-1}$ and $2$ in the entries of $\hat{\ell}$ will produce clean formulas below, see for example \eqref{e11.2}, but is apart from that not important.

  If $c_{2}^{\pm}$ are the Hadamard projectors constructed  in Prop.~\ref{prop8.4}  and $\pi_{0}: V_{2}\otimes \cc^{2}\to V_{2}$ denotes the projection on the first component, we set:
    \begin{equation}
 \label{e10.3}
 R_{\Sigma}\defeq \col{\hat{\ell} \pi_{0}c_{2}^{+}}{\hat{\ell}\pi_{0}c_{2}^{-}},
  \end{equation}
  so that
  \beq\label{e10.3z}
  R_{\Sigma}: \cH^{p,r}_{\b}(\Sigma; V_{2}\otimes \cc^{2})\to H^{p,r}_{\b}(\Sigma; V_{1}\otimes\cc^{2})\hbox{ is bounded for }p\in \rr,\ r\in J.
  \eeq
\begin{definition}\label{def11.1}
 The {\em microlocal TT-synchronous gauge} condition is defined in terms of Cauchy data on $\Sigma$ by:
 \beq\label{e10.-3}
 \begin{aligned}
(1)&\quad K_{\b\Sigma}^{\dag}f_{2}=0,\\
(2)&\quad R_{\Sigma}f_{2}=0.
\end{aligned}
 \eeq
 \end{definition}
\subsection{Properties of  $R_{\Sigma}\circ K_{\b\Sigma}$}\label{sec11.2}
Imposing  \eqref{e10.-3} by a gauge transformation  is equivalent to find $f_{1}\in \cinf(\Sigma; V_{1}\otimes \cc^{2})$ solving
\beq\label{e10.-4b}
R_{\Sigma}K_{\b\Sigma}f_{1}= R_{\Sigma}f_{2},
\eeq
for some given $f_{2}\in \coinf(\Sigma; V_{2}\otimes \cc^{2})\cap \Ker K_{\b\Sigma}^{\dag}$. If \eqref{e10.-4b} is satisfied then  $f_{2}- K_{\b\Sigma}f_{1}$ satisfies the microlocal TT-synchronous gauge condition. In this subsection we study the operator $R_{\Sigma}K_{\b\Sigma}$ appearing in \eqref{e10.-4b}. 

The leading factor $\rho$ in $K_{\b\Sigma}$ is accounted for by putting the gauge parameter in weight $r-1$ and the two-tensor data in weight $r$.

\subsubsection{Some equivalent norms}\label{sec11.2.1}

Recall that 
 the space $\cH_{\b}^{p, r}(\Sigma; V_{k}\otimes\cc^{2})$ is 
  equipped with the norm
\[
\|f\|_{p, r}^{2}= \|f_{0}\|^{2}_{H_{\b}^{p, r}(\Sigma; V_{k})}+ \|f_{1}\|^{2}_{H^{p-1, r}_{\b}(\Sigma; V_{k})},\ \ f= \col{f_{0}}{f_{1}}.
\]
We recall that $\pi_{i}f= f_{i}$ for $f= \col{f_{0}}{f_{1}}$.

It is convenient to use the equivalent norm following from Lemma~\ref{lemma11.1} below.
\begin{lemma}\label{lemma11.1}
 The map:
\beq\label{e11.-5}
L_{k}:\begin{aligned}
\cH^{p,r}_{\b}(\Sigma; V_{k}\otimes \cc^{2})\to H^{p,r}_{\b}(\Sigma; V_{k}\otimes \cc^{2})\\
 f\mapsto L_{k}f= \col{\pi_{0}c_{k}^{+}f}{\pi_{0}c_{k}^{-}f}\eqdef\col{v^{+}}{v^{-}}=v
\end{aligned}
\eeq
is an isomorphism for $r\in J$.
\end{lemma}
We will  denote by $\pi^{\pm}: \cinf(\Sigma; V_{k}\otimes \cc^{2})\to \cinf(\Sigma; V_{k})$ the maps $v= \col{v^{+}}{v^{-}}\mapsto v^{\pm}$. 

It is often more convenient to use $\col{v^{+}}{v^{-}}$ instead of $f$. 

\medskip

\noindent {\bf Proof of Lemma~\ref{lemma11.1}.} From Prop.~\ref{prop8.4} (1) we see that  an equivalent norm on $\cH_{\b}^{p,r}(\Sigma; V_{k}\otimes \cc^{2})$  for $p\in \rr,\ r\in J$ is given by
\[
\| f\|^{2}_{p,r, {\rm mod}}= \| c^{+}f\|^{2}_{p,r}+ \| c^{-}f\|^{2}_{p,r}.
\]
If $\hat{b}$ is the operator introduced in the proof of Prop.~\ref{prop8.4}, we know that $\hat{b}: H^{p,r}_{\b}(\Sigma; V)\tosim H^{p-1, r}_{\b}(\Sigma; V)$ for $p\in \rr,\ r\in J$, so  an equivalent norm is 
\[
(\|v^{+}\|^{2}_{H_{\b}^{p,r}(\Sigma; V_{k})}+ \| v^{-}\|^{2}_{H^{p,r}_{\b}(\Sigma; V_{k})})^{\12},\quad v^{\pm}= \pi_{0}c^{\pm}f,
\]
and the inverse is
\[
L_k^{-1}\col{v^+}{v^-}
=\col{v^++v^-}{b_k^+v^++b_k^-v^-}.
\]
This proves the lemma. \qed
\begin{proposition}\label{prop11.1bis}
The map
\[
        R_{\Sigma}K_{\b\Sigma}:
        \cH^{p,r-1}_{\b}(\Sigma;V_{1}\otimes\cc^{2})
        \to
        H^{p-1,r}_{\b}(\Sigma;V_{1}\otimes\cc^{2})
\]
is Fredholm for $p\in\rr$ and $r\in J\cap(J+1)\cap\,]1-r_{0},1+r_{0}[$ for $0<r_{0}$ small enough.
\end{proposition}
\proof 
We have \[
R_{\Sigma}K_{\b\Sigma}
= \col{\hat{\ell}\pi_{0}K_{\b\Sigma}c_{1}^{+}}{\hat{\ell}\pi_{0}K_{\b\Sigma}c_{1}^{-}}+ \col{\hat{\ell}\pi_{0} ( c_{2}^{+}K_{\b\Sigma}- K_{\b\Sigma}c_{1}^{+})}{\hat{\ell}\pi_{0}(c_{2}^{-}K_{\b\Sigma}- K_{\b\Sigma}c_{1}^{-})}.
\]
Let us  look at the first term on the rhs. We recall that $K_{\b\Sigma}= I_{\b\Sigma}\circ d_{\b\Sigma}$, see \eqref{e7.new5}. We have $I_{\b}\rg(0)= (1-\frac{n}{2})\rg(0)$ and $(I_{\b}u_{2})_{t\Sigma}= u_{2t\Sigma}$ since $\rg(0)_{t\Sigma}=0$.  Therefore
 \[
 \hat{\ell}\pi_{0}K_{\b\Sigma}= \ell\pi_{0}d_{\b\Sigma},\quad \hbox{ for } 
 \ell u_{2}= \col{-\12c^{-2}(\rg(0)| u_{2})_{V_{2}^{\rg(0)}}}{ 2 u_{2t\Sigma}}.
\]
We factor out a power of $\rho$ setting: 
\beq\label{e11.1b}
 \hat{\ell}\pi_{0}K_{\b\Sigma}c_{1}^{\pm}
 \eqdef  \rho\,M_{\Sigma}^{\pm}\pi_{0}c_{1}^{\pm},
 \quad
 M_{\Sigma}^{\pm}
 =
 \rho^{-1}\ell\pi_{0}d_{\b\Sigma}\begin{pmatrix}\one\\ b_{1}^{\pm}(0)\end{pmatrix}.
\eeq
Here we recall that $\pi_{1}c_{1}^{\pm}=b_{1}^{\pm}(0)\pi_{0}c_{1}^{\pm}$, where $b_{1}^{\pm}$ is defined in \eqref{e8.10}.  From \eqref{e7.new3} and Prop.~\ref{prop8.4} (4) we see that $M_{\Sigma}^{\pm}\in \Psi^{1}_{J}(\Sigma; L(V_{1}))$.

\medskip
\noindent\textit{Elllipticity  of $M_{\Sigma}^{\pm}$.}\quad We first compute the principal symbol of $M_{\Sigma}^{\pm}$. We know that $\sigma_{\rm pr}(b_{1}^{\pm}(0))= \pm (\xi\dual\rh_{\b}(0)^{-1}\xi)^{\12}\one_{V_{1}}$. We use  \eqref{e7.new6} and the fact that 
$\12(\rg(0)| u_{2})_{V_{2}^{\rg(0)}}= c^{2}(- u_{tt}+ \rh^{\alpha\beta}u_{\alpha\beta})$, where $\alpha, \beta$ range over the frame indices $dr, rd\omega$. We obtain that
\beq\label{e11.2}
\sigma_{\rm pr}(M_{\Sigma}^{\pm})= \pm \i (\xi\dual \rh_{\b}^{-1}(0)\xi)^{\12}\one_{V_{1}}+\mat{0}{-\i (\xi|}{\i |\xi)}{0}
\eeq
where 
\[
 \mat{0}{-\i (\xi|}{\i |\xi)}{0}\col{w_{t}}{w_{\Sigma}}= \col{-\i \xi\dual\rh^{-1}(0)w_{\Sigma}}{\i \xi w_{t}},\ \ \xi,\ w_{\Sigma}\in \be T^{*}\Sigma,
\]
The second term in the rhs of \eqref{e11.2} is selfadjoint for the   Hilbertian scalar product $(\cdot| \cdot)_{\tV_{1}^{\tilde{\rg}(0)}}$ defined in \ref{sec2.3.1d1}, which shows that $M_{\Sigma}^{\pm}$ is elliptic in $\Psi^{1}_{J}(\Sigma; L(V_{1}))$.

\noindent\textit{The normal operator $M_{\Sigma\norm}^{\pm}$.}\quad
By \eqref{e7.new1} and Prop.~\ref{prop8.4} (5)  we have
\[
M_{\Sigma\norm}^{\pm}= \rho^{-1}\ell_{\norm}d_{\b\Sigma}^{\reta}\begin{pmatrix}\one\\n_{1\norm}^{\pm}\end{pmatrix}
\]
where $\ell_{\norm}$ is the analog of $\ell$ with $\rg(0)$ replaced by $\reta(0)$,  the conformal factor $c$ replaced by $1$  and $n_{1\norm}^{\pm}= \cF_{1}^{-1}\tn_{1\norm}^{\pm}\cF_{1}$, see \eqref{e8.10a}. We claim that 
\begin{equation}
\label{e11.new1}
M_{\Sigma\norm}^{\pm}: H^{p, r}_{\b}(\Sigma_{\norm}; V_{1})\to H^{p-1, r}_{\b}(\Sigma_{\norm}; V_{1})\hbox{ is invertible for }p\in \rr,\ r\in J\cap\,]-r_{0}, r_{0}[
\end{equation}
for $0<r_{0}\ll 1$. To prove \eqref{e11.new1} we can perform Wick rotation and normalize the output components, setting
\[
\tM_{\Sigma\norm}^{\pm}
=-\i J_s\cF_{1}M_{\Sigma\norm}^{\pm}\cF_{1}^{-1},
\quad J_s=\operatorname{diag}(1,-1,-1).
\]
Indeed, $\cF_1\ell_{\norm}\cF_2^{-1}=\i J_s\tilde\ell_{\norm}$.
By analyticity we have:
\beq\label{e11.new3}
\tM_{\Sigma\norm}^{\pm}= \rho^{-1}\tilde{\ell}_{\norm}\pi_{0}\tilde{d}_{\b\Sigma}\begin{pmatrix}\one\\\tn_{1\norm}^{\pm}\end{pmatrix},
\eeq
where $\tilde{\ell}_{\norm}$ is the analog of $\ell_{\norm}$ with $\reta(0)$ replaced by $\treta(0)$, $\tilde{d}_{\b\Sigma}= d_{\b\Sigma}^{\treta}$ and $\treta$ is defined in \eqref{e10.n-22}.

Let us now compute $\tM_{\Sigma\norm}^{\pm}$ explicitly. Recall that $\tilde{d}_{\b\Sigma}$ is defined using the Riemannian Cauchy data
\[
\widetilde\varrho w=\binom{w(0)}{-\widetilde{\nabla}_{s}w(0)},
\]
see \eqref{e10.stup-trace}, where $\widetilde{\nabla}_{s}$ is defined in \eqref{e10.n2b}.
Therefore to  compute $\tM_{\Sigma\norm}^{\pm}v$ for $v\in \coinf(\Sigma_{\norm}; V_{1})$ we have to compute $\rho^{-1}\tilde{\ell}_{\norm}\tilde{d}_{\b}w(0)$, where $w\in \coinf(\tM_{\norm}; V_{1})$ is such that $\tilde{\varrho}w= \col{v}{\tn_{1\norm}^{\pm}v}$.

We identify  $V_{k}$ with $\tM_{\norm}\times \cc^{d_{k}}$ using the frame $(dr, rds, rd\omega)$, writing 
\[
w=w_rdr+w_srds+ w_\omega rd\omega,\ \ w\in \cinf(\tM_{\norm}; V_{1}),
\]
and similarly  for $u\in \coinf(\tM_{\norm}; V_{2})$.
We obtain 
\begin{align*}
(\widetilde d_{\b}w)_{rr}
  &=-\rho(\rho\partial_\rho+n/2-1)w_r,\\
(\widetilde d_{\b}w)_{ri}
  &=\frac{\rho}{2}\big(\vec{\nabla}_i w_r-(\rho\partial_\rho+n/2)w_i\big),\\
(\widetilde d_{\b}w)_{ij}
  &=\rho\big((\vec{d}w)_{ij}+\rh_{ij}w_r\big),
\end{align*}
where $i,j$ range over $(s,\omega)$, $\vec{\nabla}$, $\vec{d}$ are the connection and symmetric derivative for $\rh_{\bS^{n-1}}$. We obtain that
\begin{equation}\label{e7.boundary-trace-computation}
\rho^{-1}\widetilde\ell_{\norm}\widetilde d_{\b}w
=
\begin{pmatrix}
-\big(- (\rho\partial_\rho-n/2)w_r+\partial_s w_s-\hat{\delta} w_\omega\big)\\
-(\rho\partial_\rho+n/2)w_s+\partial_s w_r\\
\partial_s w_\omega+\hat{d} w_s
\end{pmatrix},
\end{equation}
where $\hat{d}, \hat{\delta}$ are the symmetric derivative and coderivatives for $\rh_{\bS^{n-2}}$.  If $-\widetilde{\nabla}_{s}w(0)= \tn_{1\norm}^{\pm}v$, we obtain that
\[
\partial_s w_s(0)=-(\tn_{\norm}^{\pm}v)_s-v_r,
\quad
\partial_s w_r(0)=-(\tn_{\norm}^{\pm}v)_r+v_s,
\quad
\partial_s w_\omega(0)=-(\tn_{\norm}^{\pm}v)_\omega.
\]
We substitute  these identities in \eqref{e7.boundary-trace-computation}  and multiply the last two lines by $-1$. We obtain that:
\[
\begin{pmatrix}1&0&0\\
0&-1&0\\
0&0&-1\end{pmatrix}\circ \tM_{\Sigma\norm}^{\pm}=  \tn_{1\norm}^{\pm}+ \begin{pmatrix}
        0&\rho\partial_\rho+1-n/2&\hat{\delta}\\
        \rho\partial_\rho+n/2-1&0&0\\
        -\hat{d}&0&0
        \end{pmatrix}\eqdef  \tn_{1\norm}^{\pm}+ \tilde{B}_{\Sigma\norm}.
\]
We first claim that $\tM_{\Sigma\norm}^{\pm}: H^{1, 0}_{\b}(\Sigma_{\norm}; V_{1})\to H^{0, 0}_{\b}(\Sigma_{\norm }; \tV_{1})$ is invertible.  

Computing the principal symbol of $\tM_{\Sigma\norm}^{\pm}\in N^{1}_{J}(M_{\norm})$  we first obtain that $\tM_{\Sigma\norm}^{\pm}$ is elliptic in $N^{1}_{J}(M_{\norm})$, see Def.~\ref{def9.2}, hence $\tM_{\Sigma\norm}^{\pm}$ is closed with domain $H^{1, 0}_{\b}(\Sigma_{\norm}; V_{1})$.  Next  by  Prop.~\ref{prop10.2} $\tn_{1\norm}^{\pm}$ is selfadjoint for $(\cdot| \cdot)_{\tV_{1}(\Sigma_{\norm})}$ with domain $H^{1, 0}_{\b}(\Sigma_{\norm}; V_{1}) $ with $\pm \tn_{1\norm}^{\pm}\geq C>0$. On the other hand we have $\tilde{B}_{\Sigma\norm}^{*}= - \tilde{B}_{\Sigma\norm}$ for $(\cdot| \cdot)_{\tV_{1}(\Sigma_{\norm})}$. 

Thus
\[
\pm\Re(J_s\tM_{\Sigma\norm}^{\pm})
=\pm\tn_{1\norm}^{\pm}\geq C>0.
\]
Applying the same estimate to the adjoint and using the ellipticity, we deduce that $J_s\tM_{\Sigma\norm}^{\pm}$, and hence $\tM_{\Sigma\norm}^{\pm}$, is invertible from $H^{1,0}_{\b}(\Sigma_{\norm};V_1)$ to $H^{0,0}_{\b}(\Sigma_{\norm};V_1)$.  Since $\rho^{r}\tM_{\Sigma\norm}^{\pm}\rho^{-r}= \tM_{\Sigma\norm}^{\pm}+ O(r)$ in $B(H^{1, 0}_{\b}(\Sigma_{\norm}; V_{1}))$ we obtain that $\tM_{\Sigma\norm}^{\pm}: H^{1, r}_{\b}(\Sigma_{\norm}; V_{1})\tosim H^{0, r}_{\b}(\Sigma_{\norm }; V_{1})$ for $r\in J\cap\,]-r_{0}, r_{0}[$.

This implies also that  $\tM_{\Sigma\norm}^{\pm}: H^{p, r}_{\b}(\Sigma_{\norm}; V_{1})\tosim H^{p-1, r}_{\b}(\Sigma_{\norm }; V_{1})$ for $p\in\rr,\ r\in J\cap\,]-r_{0}, r_{0}[$ by the usual argument. Using $M_{\Sigma\norm}^{\pm}=\i\cF_1^{-1}J_s\tM_{\Sigma\norm}^{\pm}\cF_1$ we obtain \eqref{e11.new1}.

\noindent\textit{Fredholm property of $R_{\Sigma}K_{\b\Sigma}$.}

Using Lemma~\ref{lemma9.4} we deduce from the ellipticity of $M_{\Sigma}^{\pm}$ and invertibility of $M^{\pm}_{\Sigma\norm}$ that 
\[
M_{\Sigma}^{\pm}: H^{p, r}_{\b}(\Sigma; V_{1})\to H^{p-1, r}_{\b}(\Sigma; V_{1})\hbox{ is Fredholm for }p\in \rr,\ r\in J\cap\,]-r_{0}, r_{0}[, 
\]
Using Lemma~\ref{lemma11.1} and the extra factor $\rho$ that was included in the definition \eqref{e11.1b} of $M_{\Sigma}^{\pm}$ we obtain that 
\[
 \col{\hat{\ell}\pi_{0}K_{\b\Sigma}c_{1}^{+}}{\hat{\ell}\pi_{0}K_{\b\Sigma}c_{1}^{-}}:  \cH^{p,r}_{\b}(\Sigma;V_{1}\otimes\cc^{2})
        \to
        H^{p-1,r+1}_{\b}(\Sigma;V_{1}\otimes\cc^{2})
\]
is Fredholm for $p\in \rr,\ r\in J\cap\,]-r_{0}, r_{0}[$. We recall that 
\beq\label{e11.1}
R_{\Sigma}K_{\b\Sigma}
= \col{\hat{\ell}\pi_{0}K_{\b\Sigma}c_{1}^{+}}{\hat{\ell}\pi_{0}K_{\b\Sigma}c_{1}^{-}}+ \col{\hat{\ell}\pi_{0} ( c_{2}^{+}K_{\b\Sigma}- K_{\b\Sigma}c_{1}^{+})}{\hat{\ell}\pi_{0}(c_{2}^{-}K_{\b\Sigma}- K_{\b\Sigma}c_{1}^{-})}.
\eeq
By Prop.~\ref{prop8.5bis} the second term in the RHS of \eqref{e11.1} is compact from $\cH^{p,r}_{\b}(\Sigma; V_{1}\otimes \cc^{2})$ to $H_{\b}^{p-1,r+1}(\Sigma;V_{1}\otimes \cc^{2})$ for $r\in J\cap(J-1)$ which completes the proof of the proposition. \qed

 \subsubsection{Positivity property}\label{sec11.2.2}
 We now prove the key positivity property of $\pm \rq_{I, 2}\circ c_{2}^{\pm}$ on $\Ker R_{\Sigma}$. The restriction to $p\geq \12,\ r\geq 0$ in Lemma~\ref{lemma11.2} below  comes from the fact that  the charges $\rq_{2\b}$, $\rq_{I, 2\b}$, $\tilde{\rq}_{2\b}$, see Subsect.~\ref{sec2.3.3} are defined on $\cH^{p,r}_{\b}(\Sigma; V_{2}\otimes \cc^{2})$ only if $p\geq \12,\ r\geq 0$. 
\begin{lemma}\label{lemma11.2}
 We have
 \[
 \pm \rq_{I,2\b}\circ c_{2}^{\pm}\geq 0 \hbox{ on }\Ker R_{\Sigma}\cap \cH^{p,r}_{\b}(\Sigma; V_{2}\otimes \cc^{2})\hbox{ for }p\geq \12,\ r\in J\cap \rr^{+}.
 \]
\end{lemma}
\proof
Using \eqref{e8.12b}  and the fact that $c_{2}^{\pm}$ are projections  we obtain that
\[
\pm \bar{f}_{2}\dual \rq_{I, 2\b}c_{2}^{\pm}f_{2}= \pm \bar{c_{2}^{\pm}f_{2}}\dual \rq_{I, 2\b}c_{2}^{\pm}f_{2}
\]
and if $R_{\Sigma}f_{2}=0$ we obtain that 
\[
 \pm \bar{c_{2}^{\pm}f_{2}}\dual \rq_{I, 2\b}c_{2}^{\pm}f_{2}=  \pm \bar{c_{2}^{\pm}f_{2}}\dual \tilde{\rq}_{2\b}c_{2}^{\pm}f_{2}.
\]
This is positive by Prop.~\ref{prop8.4} (3). \qed
\subsection{Gauge fixing in the regular case}\label{sec11.3}
\subsubsection{Notation}
To ease notation we will sometimes write $\cH^{p,r}_{\b}(V_{k})$ for $\cH^{p,r}_{\b}(\Sigma;V_{k}\otimes\cc^{2})$.

\subsubsection{Gauge fixing}
We fix $r\in J\cap (J+1)\cap\,]1-r_{0}, 1+r_{0}[$. For the construction below, we may choose $1-\delta<r<1$ with $\delta>0$ sufficiently small. The usual argument using the ellipticity of   $R_{\Sigma}K_{\b \Sigma}$ shows that if 
\beq\label{e11.n3}
R_{\Sigma}K_{\b\Sigma}: \cH^{p,r-1}_{\b}(\Sigma; V_{1}\otimes \cc^{2})\tosim  H^{p-1,r}_{\b}(\Sigma; V_{1}\otimes \cc^{2})
\eeq
holds for some $p\in \rr$ then \eqref{e11.n3} holds for all $p\in \rr$.  In this subsection we assume that \eqref{e11.n3} holds for some (and hence for all) $p\in \rr$. 
\begin{remark}\label{remark9.1}
 If the background metric $\rg$ is such that $\rg- \reta$ is small enough in $S^{-1}(\rr^{n})$, then it is easy to verify that \eqref{e11.n3} is satisfied. In fact  the operators $M_{\Sigma}^{\pm}$ are close to  the invertible $M_{\Sigma\norm}^{\pm}$ in $B(H^{p, r}_{\b}(V_{1}), H^{p-1, r+1}_{\b}(V_{1}))$ and the error terms $c_{2}^{\pm}K_{\b\Sigma}- K_{\b\Sigma}c_{1}^{+}$ are close to $0$ in the same topology. This implies the invertibility of $R_{\Sigma}K_{\b \Sigma}$.
\end{remark}

We set
\begin{equation}
\label{e11.4}
T= \one- K_{\b\Sigma}(R_{\Sigma}K_{\b\Sigma})^{-1}R_{\Sigma}.
\end{equation}
%Note that 
%\[
%T= \one- K_{\b\Sigma}(R_{\Sigma}K_{\b\Sigma})^{-1}R_{\Sigma},
%\]
%where now $R_{\Sigma}K_{\b\Sigma}: \cH^{p, -1}_{\b}(\Sigma; V_{1}\otimes \cc^{2})\tosim  H^{p-1,r}_{\b}(\Sigma; V_{1}\otimes \cc^{2})$

\begin{proposition}\label{prop11.2}
 \ben
 \item $T: \cH^{p,r}_{\b}(\Sigma; V_{2}\otimes \cc^{2})\to \Ker R_{\Sigma}\cap\cH^{p,r}_{\b}(\Sigma; V_{2}\otimes \cc^{2})$  is a bounded projection  on $\cH^{p,r}_{\b}(\Sigma; V_{2}\otimes \cc^{2})$, $p\in \rr$.
 \item  $T\circ K_{\b\Sigma}=0$.
\item $T$ preserves $\Ker K_{\b\Sigma}^{\dag}\cap \cH^{p,r}_{\b}(\Sigma; V_{2}\otimes \cc^{2})$.
\item $c_{2}^{\pm}T c_{2}^{\mp}$ is smoothing on $\Sigma$.
\een
\end{proposition}
\proof
By \eqref{e10.3z} $R_{\Sigma}: \cH^{p,r}_{\b}(V_{2}\otimes\cc^{2})\to H^{p,r}_{\b}(V_{1}\otimes\cc^{2})$ is bounded. Using   Prop.~\ref{prop11.1bis}  and the fact that $K_{\b\Sigma}: \cH^{p, r-1}_{\b}(V_{1}\otimes\cc^{2})\to \cH^{p-1, r}_{\b}(V_{2}\otimes\cc^{2})$  we obtain that $T$ is bounded on $ \cH^{p,r}_{\b}(V_{2}\otimes\cc^{2})$.  We have $R_{\Sigma}\circ T=0$  and $T= \one$ on $\Ker R_{\Sigma}\cap\cH^{p-1,r}_{\b}(V_{2})$ which implies (1).  (2)  is immediate. 
Next the identities in  \eqref{e2.12x}, applied to our $\b$-framework imply that  $K_{\b\Sigma}^{\dag}\circ K_{\b\Sigma}=0$, which implies (3).
Let us prove (4). It  suffices to consider the $+$ case and to show that the operator
 $c_{2}^{+}K_{\b\Sigma}(R_{\Sigma}K_{\b\Sigma})^{-1}R_{\Sigma} c_{2}^{-}$ is smoothing, or using Prop.~\ref{prop8.5bis} that $c_{1}^{+}(R_{\Sigma}K_{\b\Sigma})^{-1}R_{\Sigma} c_{2}^{-}$ is smoothing.
 
Recall that the projections $\pi^{\pm}$ are defined after Lemma~\ref{lemma11.1}. We note that $\pi^{+}R_{\Sigma}c_{2}^{-}=0$ and if  $R_{\Sigma}K_{\b\Sigma}f= v$ with $\pi^{+}v=0$ then by \eqref{e11.1} we have $\rho^{-1}l\pi_{0}K_{\b\Sigma}c_{1}^{+}f= M_{\Sigma}^{+}\pi_{0}c_{1}^{+}f= r_{-\infty}v$ where $r_{-\infty}$ is smoothing, since $c_{2}^{+}K_{\b\Sigma}- K_{\b\Sigma}c_{1}^{+}$ is smoothing, see Prop.~\ref{prop8.5bis}. 

By the ellipticity of $M_{\Sigma}^{+}$ we obtain that $c_{1}^{+}(R_{\Sigma}K_{\b\Sigma})^{-1}v= r_{-\infty}v$ on $\Kerop \pi^{+}$, which completes the proof of (4). \qed

\begin{remark}\label{rem11.1}
 Property (4) in Prop.~\ref{prop11.2} is the key property of the microlocal  TT-synchronous gauge condition. It will be used in the next subsection for the construction of Hadamard states. It is not satisfied by the usual TT-synchronous gauge condition.
\end{remark}

\subsection{Gauge fixing in the singular case}\label{sec11.4}
As before we fix $r\in J\cap  (J+1)\cap\,]1-r_{0}, 1+r_{0}[$ but assume now that $R_{\Sigma}K_{\b\Sigma} :\cH^{p,r-1}_{\b}(\Sigma; V_{1}\otimes \cc^{2})\to H^{p-1,r}_{\b}(\Sigma; V_{1}\otimes \cc^{2})$ is {\em not invertible}. 

\subsubsection{Notation}\label{sec11.4.1}

 The orthogonal of a subspace $E\subset \cH^{p,r}_{\b}(\Sigma; V_{2}\otimes \cc^{2})$ for $\rq_{I,2\b}$ will be denoted by $E^{\rq_{I,2\b}}$. 

%For simplicity of notation, we often write $\bar{f}\dual q f$ instead of $(f| qf)_{V\otimes \cc^{2}}$, i.e.~we identify self-adjoint operators with Hermitian %forms using the Hilbertian scalar product $(\cdot | \cdot)_{V\otimes \cc^{2}}$.
%

If $A:\cH^{\infty,r}_{\b}(\Sigma; V_{k}\otimes\cc^{2})\to H^{\infty,r}_{\b}(\Sigma; V_{1}\otimes \cc^{2})$ we define $A^{\dag}$ by
\beq\label{e11.4bb}
(u| A f)_{\tV_{1}(\Sigma)\otimes \cc^{2}}\eqdef \bar{A^{\dag}u}\dual \rq_{I,k\b}f,
\eeq
where we recall that $\rq_{I,2\b}= \rq_{2\b}\circ I_{\b\Sigma}$, with $I_{\b\Sigma}$ defined in \eqref{e2.10bf}, and we set $\rq_{I,1\b}= \rq_{1\b}$.  
Denoting by $A^{*}$ the usual adjoint obtained from the Hilbertian scalar products  $(\cdot | \cdot)_{\tV_{k}(\Sigma)\otimes \cc^{2}}$ we have $A^{*}= q_{I,k}A^{\dag}$, where the charge matrices $q_k$ were defined in \eqref{e2.10bee} and $q_{I, 2}= q_{2}\circ I_{\b}$, $q_{I, 1}= q_{1}$.

\subsubsection{The space $\cH^{p,r}_{\b \: \reg}(\Sigma; V_{2}\otimes \cc^{2})$}\label{sec11.4.2}
We set:
\[
\begin{aligned}
&\cH^{p,r}_{\b \: \reg}(\Sigma; V_{2}\otimes \cc^{2})\\ &= \{f_{2}\in \cH^{p,r}_{\b}(\Sigma; V_{2}\otimes \cc^{2}): R_{\Sigma}f_{2}\in  R_{\Sigma}K_{\b\Sigma}\cH^{p+1,r-1}_{\b}(\Sigma; V_{1}\otimes \cc^{2})\}.
\end{aligned}
\]
We will often write $\cH^{p,r}_{\b\:\reg}(V_{2})$   for  $\cH^{p,r}_{\b\:\reg}(\Sigma; V_{2}\otimes \cc^{2})$.

\subsubsection{Gauge fixing}\label{sec11.4.2b}

 The gauge fixing equation \eqref{e10.-4b} can be solved  if and only if $f_{2}\in \cH^{p,r}_{\b \: \reg}(V_{2})$. Clearly
\beq\label{e11.3z}
\Ker R_{\Sigma}\cap \cH^{p,r}_{\b}(V_{2})\subset \cH^{p,r}_{\b \: \reg}(V_{2}),\quad K_{\b\Sigma}\cH^{p+1,r-1}_{\b}(V_{1})\subset \cH^{p,r}_{\b \: \reg}(V_{2}).
\eeq

By ellipticity and  Fredholm property of $R_{\Sigma}K_{\b\Sigma}$,  
 there exist $u_{i}\in \Ker (R_{\Sigma}K_{\b\Sigma})^{*}\cap H^{\infty,-r}_{\b}(V_{1}\otimes \cc^{2})$, for $1\leq i\leq m$ such that 
 \[
 u\in R_{\Sigma}K_{\b\Sigma}\cH^{p+1,r-1}_{\b}(V_{1})\hbox{ iff }
(u_{i}| u)_{\tV_{1}(\Sigma)\otimes \cc^{2}}=0,\ \ 1\leq i\leq m.
\]
We claim that $v_{i}\defeq R_{\Sigma}^{\dag}u_{i}\in  \cH^{\infty,-r}_{\b}(V_{2})$. In fact   by the definition of $A^{\dag}$ in \eqref{e11.4bb}:
\[
\begin{aligned}
|\bar{v}_{i}\dual \rq_{2\b}f|
&= |(u_{i}| R_{\Sigma}I_{\b\Sigma}^{-1}f)_{\tV_{1}(\Sigma)\otimes \cc^{2}}|
\leq C \|f\|_{\cH_{\b}^{-p, r}(V_{2})},
\end{aligned}
\]
since  $u_{i}\in  H^{p,-r}_{\b}(V_{1}\otimes \cc^{2})$, $I_{\b\Sigma}^{-1}\in B(\cH_{\b}^{-p, r}(V_{2}))$ and $R_{\Sigma}: \cH_{\b}^{-p, r}(V_{2})\to H^{-p, r}_{\b}(V_{1}\otimes \cc^{2})$ by \eqref{e10.3z}. This implies that 
 $v_{i}\in \cH^{p+1,-r}_{\b}(V_{2})$, which proves our claim.  
 
 Since $u_i\in\Ker (R_{\Sigma}K_{\b\Sigma})^*$, the preceding adjoint convention gives
\[
        0=(R_{\Sigma}K_{\b\Sigma})^{*}u_{i}
        =q_{1}K_{\b\Sigma}^{\dag}R_{\Sigma}^{\dag}u_{i}.
\]
The matrix $q_1$ is non-degenerate, hence
\beq\label{e11.3}
K_{\b\Sigma}^{\dag}v_{i}=0.
\eeq
After discarding redundant vectors and relabelling, we assume that $v_1,\ldots,v_m$ are linearly independent and set $\mathcal{V}= {\rm Vect}\{v_{i}: 1\leq i\leq m\}$. 
Then 
\beq\label{e11.4b}
\cH^{p,r}_{\b \: \reg}(V_{2})= \{f_{2}\in \cH^{p,r}_{\b}(V_{2}): \bar{v}_{i}\dual \rq_{I,2\b}f_{2}=0,\ \ 1\leq i\leq m\}= \cV^{\rq_{I,2\b}},
\eeq
and $\cH^{p,r}_{\b \: \reg}(V_{2})$ is of codimension $m$ in $\cH^{p,r}_{\b}(V_{2})$. 
Let us pick a supplementary space $\cV_{1}$ in $\cV$ of $\cV\cap  K_{\b\Sigma}\cH^{\infty,-r-1}_{\b}(V_{1})$.  We can assume that $v_{1}, \dots, v_{l}\in \cV_{1}$ and $v_{l+1}, \dots, v_{m}\in \cV\cap K_{\b\Sigma}\cH^{\infty,-r-1}_{\b}(V_{1})$. Since
\[
f_{2}\in \Ker K_{\b\Sigma}^{\dag}\cap \cH^{p,r}_{\b}(V_{2})\hbox{ iff }\bar{f}_{2}\dual \rq_{I,2\b}K_{\b\Sigma}f_{1}=0,\ \ \forall f_{1}\in \cH^{\infty,-r-1}_{\b}(V_{1}),
\]
we obtain that
\[
\Ker K_{\b\Sigma}^{\dag}\cap \cH^{p,r}_{\b \: \reg}(V_{2})=\Ker K_{\b\Sigma}^{\dag}\cap \cH^{p,r}_{\b}(V_{2})\cap \cV_{1}^{\rq_{I,2\b}},
\]
and $\Ker K_{\b\Sigma}^{\dag}\cap \cH^{p,r}_{\b \: \reg}(V_{2})$ is of codimension $l$ in $\Ker K_{\b\Sigma}^{\dag}\cap \cH^{p,r}_{\b}(V_{2})$.  Therefore we can find $w_{1}, \dots, w_{l}\in\Ker K_{\b\Sigma}^{\dag}\cap \cH^{\infty,r}_{\b}(V_{2})$ such that  the matrix $A= [\bar{w}_{i}\dual \rq_{I,2\b}v_{j}]_{1\leq i, j\leq l}$ is invertible. 

%Since $\cH^{p,r}_{\b \: \reg}(V_{2})$ is of codimension $n$ in $\cH^{p,r}_{\b}(V_{2})$, 
We can   complete the $w_{i}$ for $1\leq i\leq l$ by $w_{i}\in \cH^{\infty,r}_{\b}(V_{2})$ for $l+1\leq i\leq m$ so that $Q= [\bar{w}_{i}\dual \rq_{I,2\b}v_{j}]_{1\leq i,j\leq m}$ is invertible.  Since  $w_{i}\in \Ker K_{\b\Sigma}^{\dag}$ and $v_{j}\in  K_{\b\Sigma}\cH^{\infty,-r-1}_{\b}(V_{1})$ for $i\leq l,\quad j\geq l+1$ we have $Q= \mat{A}{0}{B}{C}$, where $A\in M_{l}(\cc)$, $C\in M_{m-l}(\cc)$ are invertible.
Replacing $w_{i}$ by $\sum_{k}\bar{T}_{ik}w_{k}$ for $T= \mat{A^{-1}}{0}{-C^{-1}B A^{-1}}{C^{-1}}$ we obtain
\beq\label{e11.5}
\begin{aligned}
i)&\quad w_{i}\in \cH^{\infty,r}_{\b}(V_{2}),\quad v_{j}\in \cH^{\infty, -r}_{\b}(V_{2}),\\
ii)&\quad \bar{w}_{i}\dual \rq_{I,2\b}v_{j}= \delta_{ij},\\
iii)&\quad w_{i}\in \Ker K_{\b\Sigma}^{\dag}\hbox{ for } i\leq l,\\
iv)&\quad v_{j}\in \Ker K_{\b\Sigma}^{\dag}\hbox{ for }1\leq j\leq m,\\
v)&\quad v_{j}\in  K_{\b\Sigma}\cH^{\infty,-r-1}_{\b}(V_{1})\hbox{ for }l+1\leq j\leq m.
\end{aligned}
\eeq
We set now
\begin{equation}
\label{e10.8}
\pi f_{2}= \sum_{k=1}^{m}\bar{v}_{k}\dual \rq_{I,2\b}f_{2} w_{k}.
\end{equation}

\begin{lemma}\label{lemma11.3}
\ben
\item $1- \pi: \cH^{p,r}_{\b}(\Sigma; V_{2}\otimes \cc^{2})\to \cH^{p,r}_{\b \: \reg}(\Sigma; V_{2}\otimes \cc^{2})$ is a projection.
\item $\pi\circ K_{\b\Sigma}=0$ on $\cH^{p,r-1}_{\b}(\Sigma; V_{1}\otimes \cc^{2})$.
\item $1- \pi$ preserves  $\Ker K_{\b\Sigma}^{\dag}\cap \cH^{p,r}_{\b}(\Sigma; V_{2}\otimes \cc^{2})$.
\item $\pi^{\dag}$ maps $K_{\b\Sigma}\cH^{p,r-1}_{\b}(\Sigma; V_{1}\otimes\cc^{2})$ into $K_{\b\Sigma}\cH^{\infty,-r-1}_{\b}(\Sigma; V_{1}\otimes\cc^{2})$.
\een
\end{lemma}
\proof The fact that $\pi$ is a projection is easy. We have $\Kerop \pi\cap \cH^{p,r}_{\b}(V_{2})= \cH^{p,r}_{\b \: \reg}(V_{2})$ by  \eqref{e11.4b} which proves (1).   (2) follows from  \eqref{e11.5} {\it iv)}, (3) follows from  \eqref{e11.5} {\it iii)} and {\it v)}.   Finally we have
\[
\pi^{\dag}f_{2}= \sum_{k=1}^{m}\bar{w}_{k}\dual \rq_{I,2\b}f_{2}v_{k},
\]
which again by \eqref{e11.5} {\it iii)} and {\it v)} proves (4).\qed

%ici 31/8/2026.

We now construct the analog of the projection $T$ in the singular case. We fix two projections:
\[
\begin{aligned}
\pi_{1}&: \cH^{p,r-1}_{\b}(\Sigma; V_{1}\otimes \cc^{2})\to \Ker R_{\Sigma}K_{\b\Sigma}\cap  \cH^{p,r-1}_{\b}(\Sigma; V_{1}\otimes \cc^{2}),\\
\pi_{2}&: \cH^{p,r}_{\b}(\Sigma; V_{2}\otimes \cc^{2})\to   K_{\b\Sigma}\cH^{p+1,r-1}_{\b}(\Sigma; V_{1}\otimes \cc^{2})\cap \Ker R_{\Sigma}.
\end{aligned}
\]
The projections $\pi_{i}$ are finite rank with smooth distributional kernels. Let us define the map
\[
(R_{\Sigma}K_{\b\Sigma})^{-1}R_{\Sigma}: \cH^{p,r}_{\b \: \reg}(\Sigma; V_{2}\otimes \cc^{2})\to \cH^{p+1,r-1}_{\b}(\Sigma; V_{1}\otimes \cc^{2})\cap \Ker \pi_{1},
\]
such that  $f_{1}= (R_{\Sigma}K_{\b\Sigma})^{-1}R_{\Sigma}f_{2}$  is the unique solution of
\beq\label{e-06}
\begin{cases}
R_{\Sigma}K_{\b\Sigma}f_{1}= R_{\Sigma}f_{2}, \\
 \pi_{1}f_{1}=0.
\end{cases}
\eeq
We set 
\beq\label{e10.20}
\begin{aligned}
T_{\rm reg}&=\one - K_{\b\Sigma}(R_{\Sigma}K_{\b\Sigma})^{-1}R_{\Sigma} \\ & \phantom{=\,} : \cH^{p,r}_{\b \: \reg}(\Sigma; V_{2}\otimes\cc^{2})\to  \cH^{p,r}_{\b}(\Sigma; V_{2}\otimes\cc^{2}).
\end{aligned}
\eeq

\begin{lemma}\label{lemma11.4}
  \ben
 \item The map \[
 T_{\rm reg}:\cH^{p,r}_{\b \: \reg}(\Sigma; V_{2}\otimes \cc^{2})\to\Ker R_{\Sigma}\cap \cH^{p,r}_{\b}(\Sigma;V_{2}\otimes \cc^{2})
  \]
  is a projection.
  \item $  T_{\rm reg}$ preserves $\cH^{p,r}_{\b \: \reg}(\Sigma; V_{2}\otimes \cc^{2})\cap \Ker K_{\b\Sigma}^{\dag}$.
  \item $  T_{\rm reg}K_{\b\Sigma}= K_{\b\Sigma}\pi_{1}$.
  \een

\end{lemma}
 \proof It is easy to check that $T_{\rm reg}$ is a projection. We  have $R_{\Sigma}T_{\rm reg}=0$ hence $ T_{\rm reg}\cH^{p,r}_{\b\:{\rm reg}}(V_{2})\subset \Ker R_{\Sigma}\cap \cH^{p,r}_{\b\:{\rm reg}}(V_{2})$ and $ T_{\rm reg}= \one $ on $\Ker R_{\Sigma}\cap \cH^{p,r}_{\b\:{\rm reg}}(V_{2})$ which implies (1).   (2) follows from $K_{\b\Sigma}^{\dag}K_{\b\Sigma}=0$. If $f_{2}= K_{\b\Sigma}g_{1}$, the unique solution of \eqref{e-06} is $f_{1}= (1-\pi_{1})g_{1}$  which implies (3). \qed
 
 We set now
\beq\label{e10.21}
T= (1-\pi_{2})\circ T_{\rm reg}\circ (1-\pi):\cH^{p,r}_{\b}(\Sigma; V_{2}\otimes \cc^{2})\to \cH^{p,r}_{\b}(\Sigma; V_{2}\otimes \cc^{2}).
\eeq
\begin{proposition}\label{prop11.3}
 \ben
 
 \item $T: \cH^{p,r}_{\b}(\Sigma; V_{2}\otimes \cc^{2})\to  \Ker R_{\Sigma}\cap \Ker \pi_{2}\cap\cH^{p,r}_{\b}(\Sigma; V_{2}\otimes \cc^{2})$ is a bounded projection on $\cH^{p,r}_{\b}(\Sigma; V_{2}\otimes \cc^{2})$.
  \item $T\circ K_{\b\Sigma}=0$.
  \item $T$ preserves $\Ker K_{\b\Sigma}^{\dag}\cap \cH^{p,r}_{\b}(\Sigma; V_{2}\otimes \cc^{2})$.
  \item $c_{2}^{\pm}T c_{2}^{\mp}$ is smoothing on $\Sigma$.
 \een
\end{proposition}
\proof  
 From Lemma~\ref{lemma11.4} (1) and $ \pi_{2}\cH^{p,r}_{\b\:{\rm reg}}(V_{2})\subset \Ker R_{\Sigma}\cap \cH^{p,r}_{\b\:{\rm reg}}(V_{2})$ we get that  $R_{\Sigma}T=0$, so $ T\cH^{p,r}_{\b}(V_{2})\subset \Ker R_{\Sigma}\cap \Ker \pi_{2}\cap \cH^{p,r}_{\b}(V_{2})$. Since $\Ker R_{\Sigma}\subset\cH^{p,r}_{\b\:{\rm reg}}(V_{2})$ by \eqref{e11.3z} we have  $\pi=0$ on $\Ker R_{\Sigma}\cap \cH^{p,r}_{\b}(V_{2})$, hence $ T_{\rm reg}(1-\pi)=\one$ on $\Ker R_{\Sigma}\cap \cH^{p,r}_{\b}(V_{2})$ hence $T= \one$ on $\Ker R_{\Sigma}\cap \Ker \pi_{2}\cap \cH^{p,r}_{\b}(V_{2})$. This proves (1).

By Lemma~\ref{lemma11.3} (2) we have $(1- \pi)K_{\b\Sigma}= K_{\b\Sigma}$ hence $TK_{\b\Sigma}= (1- \pi_{2})K_{\b\Sigma}\pi_{1}=0$ by Lemma~\ref{lemma11.4} (3) and the definition of $\pi_{1}, \pi_{2}$. This proves (2).

(3) follows from Lemma~\ref{lemma11.4} (2) and Lemma~\ref{lemma11.3} (3). To prove (4) we argue  as in the proof of Prop.~\ref{prop11.2} (4) using additionally the fact that $\pi, \pi_{2}$ have smooth distributional kernels.  \qed

\section{Hadamard states}\init\label{sec12}
In this final section we construct a Hadamard state by modifying the Hadamard projectors $c_{2}^{\pm}$ using the projection $T$. In the regular case $T$ was defined in \eqref{e11.4}; in the singular case it is defined in \eqref{e10.21}. We recall that   $\Ker_{\rm c}K_{\b\Sigma}^{\dag}$  and $\Ran_{\c}K_{\b\Sigma}$ in Thm.~\ref{thm11.1} and Thm.~\ref{thm11.2} below denote $\Ker K_{\b\Sigma}^{\dag}\cap \coinf(\Sigma; V_{2}\otimes \cc^{2})$ and $K_{\b\Sigma}\coinf(\Sigma; V_{1}\otimes\cc^{2})$ respectively.  Throughout this section  the weight index $r$ belongs to an interval $]1-\delta, 1[$ for $0<\delta\ll 1$, see  the proof of Lemma~\ref{lemma11.z}.
 
\subsection{Covariances in the $\b$-framework}
For each pair of covariances constructed below, we set
\[
C_{2\b}^{\pm}\defeq\pm\rq_{I,2\b}^{-1}\boldsymbol{\lambda}_{2\b}^{\pm}.
\]
We start by working in the $\b$-framework, i.e.~we consider the operators $D_{k\b}$, $K_{\b}$ etc.
\subsubsection{Regular case}
We consider first the  simpler regular case.
 Let us set
\beq\label{e12.1}
\boldsymbol{\lambda}_{2\b}^{\pm}\defeq \pm T^{*}\circ \rq_{I,2\b}c_{2}^{\pm}\circ T.
\eeq
Here $\rq_{I,2\b}$ is the physical Lorentzian charge from \eqref{e2.physical-charge}, $c_2^\pm$ are the Hadamard projectors constructed in Prop.~\ref{prop8.4} and the  projection $T$ is constructed in Prop.~\ref{prop11.2}. 

\begin{lemma}\label{lemma11.z}
 The Hermitian forms $\boldsymbol{\lambda}_{2\b}^{\pm}$ are well defined on $\coinf(\Sigma; V_{2}\otimes \cc^{2})$ for $n\geq 4$.
\end{lemma}
\proof
 Let $f_{2}\in \cH^{\infty, r}_{\b}(\Sigma; V_{2}\otimes\cc^{2})$ for $r\in ]1- \delta, 1[$ and $0<\delta\ll 1$. Then  $r\in J\cap (J+1)\cap\,]1-r_{0}, 1+ r_{0}[$  so $Tf_{2}\in  \cH^{\infty, r}_{\b}(\Sigma; V_{2}\otimes\cc^{2})$ by Prop.~\ref{prop11.2} and $r\in J$ so $c_{2}^{\pm}Tf_{2}\in \cH^{\infty, r}(\Sigma; V_{2}\otimes\cc^{2})$. Therefore $\bar{Tf_{2}}\dual \rq_{I, 2\b}c_{2}^{\pm}Tf_{2}$ is well defined since $r\geq 0$. \qed

We recall that  $\cU_{2\b}$ is the solution of the Cauchy problem for $D_{2\b}$, see \eqref{e0b.3}. 

\begin{theoreme}\label{thm11.1}
 The pair $\boldsymbol{\lambda}_{2\b}^{\pm}$ satisfy:
 \ben
\item  $\WF(\cU_{2\b}\circ C_{2\b}^{\pm})'\subset (\cN^{\pm}\cup\cF)\times T^{*}\Sigma$ 
over $\mathcal{O}\times \Sigma$, where $\mathcal{O}$ is a neighborhood of $\Sigma$ in $M$ and $\cF$ is conic with  $\cF\cap\cN= \emptyset$.
\item $ \boldsymbol{\lambda}_{2\b}^{+}- \boldsymbol{\lambda}_{2\b}^{-}= \rq_{I,2\b}$ on $\Ker_{\c} K_{\b\Sigma}^{\dag}$.
\item $\boldsymbol{\lambda}_{2\b}^{\pm}=0$ on $\Ker_{\c}K_{\b\Sigma}^{\dag}\times \Ran_{\c} K_{\b\Sigma}$.
\item $\boldsymbol{\lambda}_{2\b}^{\pm}= \boldsymbol{\lambda}_{2\b}^{\pm*}$, $\boldsymbol{\lambda}_{2\b}^{\pm}\geq 0$ on $\Ker_{\c} K_{\b\Sigma}^{\dag}$.
\een
\end{theoreme}
\proof 
Let us first prove (1). We apply \cite[Prop.~11.1.1]{G1}. 
For $A\in L(\cH_{\b}^{p,r}(V_{2}))$ we set  $A^{*}\rq_{I,2\b}\eqdef \rq_{I,2\b}A^{\dag}$. Note that Prop.~\ref{prop8.4} (2) can be rewritten as $c_{2}^{\pm}= (c_{2}^{\pm})^{ \dag}$.

 We  need to show that $\WF(\cU_{2\b}T^{\dag}c_{2}^{\pm}T)'\subset (\cN^{\pm}\cup \cF)\times T^{*}\Sigma$ for some conic set $\cF\subset T^{*}M$ with $\cF \cap \cN= \emptyset$.

By Prop.~\ref{prop11.2} $c_{2}^{\pm}T  c_{2}^{\mp}$ is smoothing, hence $c_{2}^{\pm}T^{\dag}  c_{2}^{\mp}$ is smoothing, since $c_{2}^{\pm}= ( c_{2}^{\pm})^{\dag}$.  Therefore $\cU_{2\b}T^{\dag}c_{2}^{\pm}T= \cU_{2\b}c_{2}^{\pm}T^{\dag}T$ modulo a smooth kernel.  So the result follows from the Hadamard property of $c_{2}^{\pm}$, see Prop.~\ref{prop8.4} (5).

Let us now prove (2).   We have $\boldsymbol{\lambda}_{2\b}^{+}- \boldsymbol{\lambda}_{2\b}^{-}= T^{*}\rq_{I,2\b}T$. Let us fix some $p\geq \12$, see the discussion in \ref{sec11.2.2} and $r\in ]1-\delta, 1[$ as in the proof of Lemma~\ref{lemma11.z}.
$T$ preserves $\Ker K_{\b\Sigma}^{\dag}\cap\cH^{p,r}_{\b}(\Sigma; V_{2}\otimes \cc^{2})$ by Prop.~\ref{prop11.2} (3) and $Tf_{2}= f_{2}$ mod $K_{\b\Sigma}\cH^{p+1,r-1}_{\b}(\Sigma; V_{1}\otimes \cc^{2})$ hence $\boldsymbol{\lambda}_{2\b}^{+}- \boldsymbol{\lambda}_{2\b}^{-}= \rq_{I,2\b}$ on $\Ker K_{\b\Sigma}^{\dag}\cap\cH^{p,r}_{\b}(\Sigma; V_{2}\otimes \cc^{2})$.

(3) follows from Prop.~\ref{prop11.2} (2).

Since  $ T\cH^{p,r}_{\b}(\Sigma; V_{2}\otimes \cc^{2})\subset \Ker R_{\Sigma}\cap \cH^{p,r}_{\b}(\Sigma; V_{2}\otimes \cc^{2})$ and $\pm \rq_{I,2\b}c_{2}^{\pm}\geq 0$ on $\Ker R_{\Sigma}\cap \cH^{p,r}_{\b}(\Sigma; V_{2}\otimes \cc^{2})$ by Lemma~\ref{lemma11.2}, we obtain (4). 
 \qed

\subsubsection{Singular case} We now consider the singular case.  We need an additional modification  of $c_{2}^{\pm}$ since  $T^{*}\rq_{I,2\b}T\neq \rq_{I,2\b}$ on $\Ker K_{\b\Sigma}^{\dag}$ because of the projection $\pi$.  This modification is inspired by  a construction in \cite[Subsect.~4.4]{FS}.
Let
\beq\label{e11.18}
\tilde{\pi}f_{2}= \sum_{k=1}^{l}\bar{v}_{k}\dual \rq_{I,2\b}f_{2}w_{k},
\eeq
and note using \eqref{e11.5} {\it v)} that $\tilde{\pi}= \pi$ on $\Ker K_{\b\Sigma}^{\dag}$.

\def\tpi{\tilde{\pi}}
Consider the hermitian form \beq\label{e11.12}
\begin{aligned}
\boldsymbol{\nu} &\defeq\rq_{I,2\b}- (1-\tpi)^{*}\rq_{I,2\b}(1- \tpi)\\
&=\rq_{I,2\b}\tpi^{\dag}+ \rq_{I,2\b} \tpi- \rq_{I,2\b}\tpi^{\dag}\tpi.
\end{aligned}
\eeq
 The Hermitian form $\boldsymbol{\nu}$ has a smooth distributional kernel,  a finite rank and is bounded on $\cH^{0,r}_{\b}(\Sigma; V_{2})$ by \eqref{e11.5} {\it i)}.    We define the finite rank self-adjoint operator $\nu$ by
 \[
 \bar{f}_{2}\dual \boldsymbol{\nu}f_{2}\eqdef (\rho^{-r}f_{2}| \nu \rho^{-r} f_{2})_{\tV_{2}(\Sigma)\otimes \cc^{2}}
 \]
   and we can find $u_{i}\in \cH^{\infty,0}_{\b}(\Sigma; V_{2}\otimes \cc^{2})$, $1\leq i\leq q$ linearly independent such that 
\[
\nu= \sum_{i=1}^{q}\alpha_{i}|u_{i})(u_{i}|,\ \ \alpha_{i}\in \rr^{*}.
\]

We set
\beq\label{e11.10}
\bar{f}_{2}\dual \boldsymbol{\lambda}_{2\b}^{\pm} f_{2}\defeq\pm \bar{Tf_{2}}\dual \rq_{I,2\b}c_{2}^{\pm}Tf_{2} \pm (\rho^{-r}f_{2}| \one_{\rr^{\pm}}(\nu)\nu \rho^{-r}f_{2})_{\tV_{2}(\Sigma)\otimes \cc^{2}}.
\eeq
As before $ \boldsymbol{\lambda}_{2\b}^{\pm}$ are well defined on $\coinf(\Sigma; V_{2}\otimes \cc^{2})$ for $n\geq 4$.

\begin{theoreme}\label{thm11.2}
 The pair $\boldsymbol{\lambda}_{2\b}^{\pm}$ satisfy:
\ben
\item  $\WF(\cU_{2\b}\circ C_{2\b}^{\pm})'\subset (\cN^{\pm}\cup\cF)\times T^{*}\Sigma$ 
over $\mathcal{O}\times \Sigma$, where $\mathcal{O}$ is a neighborhood of $\Sigma$ in $M$ and $\cF$ is conic with  $\cF\cap\cN= \emptyset$.
\item $\boldsymbol{\lambda}_{2\b}^{+}- \boldsymbol{\lambda}_{2\b}^{-}= \rq_{I,2\b}$ on $\Ker_{\c} K_{\b\Sigma}^{\dag}$.
\item $\boldsymbol{\lambda}_{2\b}^{\pm}=0$ on $\Ker_{\c}K_{\b\Sigma}^{\dag}\times \Ran_{\c} K_{\b\Sigma}$.
\item $\boldsymbol{\lambda}_{2\b}^{\pm}= \boldsymbol{\lambda}_{2\b}^{\pm*}$, $\boldsymbol{\lambda}_{2\b}^{\pm}\geq 0$ on $\Ker_{\c} K_{\b\Sigma}^{\dag}$.
\een
\end{theoreme}
\proof 
The proof of (1) is identical to Thm.~\ref{thm11.1} since  the additional term in \eqref{e11.10}  has a smooth distributional kernel.

Let us prove (2).  We have using \eqref{e11.12}:
\[
\boldsymbol{\lambda}_{2\b}^{+}- \boldsymbol{\lambda}_{2\b}^{-}= T^{*}\rq_{I,2\b}T+ \rq_{I,2\b}- (1- \tpi)^{*}\rq_{I,2\b}(1- \tpi).
\]
Let us compute the first term in the RHS  on $\Ker K_{\b\Sigma}^{\dag}\cap\cH^{p,r}_{\b}(V_{2})$. Recall that $T= (1- \pi_{2})T_{\reg}(1- \pi)$. We   have $T_{\rm reg}(1- \pi)f_{2}= (1- \pi)f_{2}$ modulo  $K_{\b\Sigma}\cH^{p+1,r-1}_{\b}(V_{1})$ by \eqref{e10.20}. Since $ \pi_{2}\cH^{p,r}_{\b}(V_{2})\subset K_{\b\Sigma}\cH^{p+1,r-1}_{\b}(V_{1})$ we obtain that $Tf_{2}= (1- \pi)f_{2}$ modulo  $K_{\b\Sigma}\cH^{p+1,r-1}_{\b}(V_{1})$. By Prop.~\ref{prop11.3} $T$ preserves $\Ker K_{\b\Sigma}^{\dag}\cap \cH^{p,r}_{\b}(V_{2})$ hence for $f_{2}\in \Ker_{\c} K_{\b\Sigma}^{\dag}$ we have:
\[
 \bar{Tf_{2}}\dual \rq_{I,2\b}Tf_{2}= \bar{(1-\pi)f_{2}}\dual \rq_{I,2\b}(1- \pi)f_{2}
 =\bar{(1-\tpi)f_{2}}\dual \rq_{I,2\b}(1- \tpi)f_{2}, 
\]
since $\pi= \tilde{\pi}$ on  $\Ker_{\c} K_{\b\Sigma}^{\dag}$.
Therefore $\boldsymbol{\lambda}_{2\b}^{+}- \boldsymbol{\lambda}_{2\b}^{-}= \rq_{I,2\b}$ on $\Ker_{\c}\! K_{\b\Sigma}^{\dag}$.

We now prove (3). By Prop.~\ref{prop11.3} (2) $T\circ K_{\b\Sigma}=0$. Since $v_{k}, w_{k}\in \Ker K_{\b\Sigma}^{\dag}$ for $1\leq k\leq l$, see \eqref{e11.5}, we obtain that $\tilde{\pi}= \tilde{\pi}^{\dag}=0$ on $\Ran_{\c}K_{\b\Sigma}$.  Therefore $\boldsymbol{\nu}=0$ on $\Ran_{\c} K_{\b\Sigma}$ hence $\nu=0$ on $\rho^{-r}\Ran_{\c}K_{\b\Sigma}$  and  $\boldsymbol{\lambda}_{2\b}^{\pm}= 0$ on $\Ker_{\c}K_{\b\Sigma}^{\dag}\times \Ran_{\c} K_{\b\Sigma}$. This proves (3). 
  
It remains to prove (4). The first term in the RHS of \eqref{e11.10} is positive  by the same argument as in the proof of Thm.~\ref{thm11.1}, since $ T\cH^{p,r}_{\b}(V_{2})\subset \Ker R_{\Sigma}\cap \cH^{p,r}_{\b}(V_{2})$. The second term is also clearly positive. This completes the proof of the theorem. \qed

\subsection{Hadamard state on $\CCR(\cV_{P},{\bf Q}_{P})$}

We now construct a Hadamard state on the physical algebra $\CCR(\cV_{P}, {\bf Q}_{P})$, i.e.~for the original linearized Einstein equation around the background metric $\rg$. 

To this end we apply the discussion in Subsect.~\ref{sec1.8}. The operators $D_{k\b}$, $K_{\b}$ etc are obtained from $D_{k}, K$  by \eqref{e2.12r}, \eqref{e2.12} for 
 $c$ replaced by  $c\rho$ and $U_{k}$  by $U_{\b}$, using the identities in \eqref{e6.-1}. If we use the convention for Cauchy data in \eqref{e00.3a} and \eqref{e2.10bc}, we obtain that the map $\mathcal{T}_{k}$ defined in \eqref{e2.12p} equals
 \begin{equation}
 \label{e100}
 \mathcal{T}_{k}\col{f_{0}}{f_{1}}= \col{c\rho U_{\b}f_{0}}{U_{\b}f_{1}},\ \ f= \col{f_{0}}{f_{1}}\in \coinf(\Sigma; V_{k}\otimes \cc^{2}).
 \end{equation}
 Following \eqref{e2.12o} we set
 \begin{equation}
 \label{e101}
 \boldsymbol{\lambda}_{2}^{\pm }\defeq (\mathcal{T}_{2})^{*}\circ \boldsymbol{\lambda}_{2\b}^{\pm}\circ \mathcal{T}_{2},
 \end{equation}
 where $\boldsymbol{\lambda}_{2\b}^{\pm}$ is defined by \eqref{e12.1} in the regular case of by \eqref{e11.10} in the singular case.
 
 From the discussion  in Subsect.~\ref{sec1.8} we obtain the following theorem.
 
\begin{theoreme}\label{mainthm}
 The Hermitian forms $ \boldsymbol{\lambda}_{2}^{\pm }$ are the Cauchy surface covariances of a quasi-free Hadamard state on $\CCR(\cV_{P}, {\bf Q}_{P})$.
\end{theoreme}

\appendix\section{Pseudodifferential calculus}\init\label{sec9}
\def\cyl{{\rm cyl}}\def\Mcyl{M_{\rm cyl}}
\subsection{Notation}\label{sec9.1}
Let $M$ be a smooth $n$-dimensional manifold with compact boundary $\p M= Y$ and $x$ a boundary defining function. We refer back to \ref{secnot.1.2} for definitions of general notions for manifolds with boundary which we use below. As in that section, we will fix $x$ and not study the invariance of objects under a change of boundary defining function.

The calculus below is a variant of the small $\b$-calculus and of pseudodifferential calculi on manifolds with cylindrical ends; see, for instance, \cite{GilLoya,MitreaNistor}. In the cylindrical picture $x=\e^{-s}$, the normal operator below is the analogue of the indicial (or translation-invariant) operator introduced in these references. Compared with the almost translation-invariant algebra of Mitrea--Nistor, the main difference relevant here is that we work uniformly on a fixed interval of $\b$-weights $H^{p,r}_{\b}=x^{r}H^{p}_{\b}=\e^{-rs}H^{p}$, whereas their calculus is formulated in the framework of  bounded geometry Sobolev spaces without weight. 

For our constructions the relevant interval of weights is $J=\open{-\frac n2+1,\frac n2-1}$, 
the definitions below are however stated for an arbitrary open interval $J\subset\rr$.

\subsubsection{Hermitian and Hilbertian vector  bundles}\label{sec9.2.2}
%where do we actually use this?

 We fix a  rank $N$ complex vector bundle $V\overset{\pi}{\to} M$ which is assumed to be trivial over $M_{\infty}$.  If $V$ is a Hermitian bundle, we will assume that, after trivialisation over $M_{\infty}$, the Hermitian form $h(x, y)$ on the fibers of $V$ belongs to $\mathcal{A}^{\infty}_{1}(M; L_{\rh }(\cc^{N}))$, where $L_{\rm h}(\cc^{N})$ denotes the space of Hermitian forms on $\cc^{N}$. 
 
 The scalar product on the fibers of $V$ obtained from $h$ is denoted by $(\cdot| \cdot)_{V}$, the one obtained from $h_{\norm}$ over $M_{\norm}$ by $(\cdot| \cdot)_{V_{\norm}}$.

Later on we will need to define $L(V)$-valued symbols or $V$-valued Sobolev spaces. This can be done using local trivializations, cf.~Subsect.~\ref{not:tensor-valued-classes}, but it will often be convenient for our purposes to also equip the fibers of $V$ with some reference Hilbertian structure. 

We fix a Hilbertian scalar product $r(s, y)$ on the fibers of $V$  belonging to $\mathcal{A}^{\infty}_{1}(M; L_{\rm h}(\cc^{N}))$. If $h$ is a Hermitian form on the fibers of $V$ as above, then 
$h(x, y)= r(x, y)\circ \tau(x, y)$, where $\tau\in \mathcal{A}^{\infty}_{1}(M; L(\cc^{N}))$.

\subsection{Phase-space symbols}
Recall that over $M_{\infty}$, the coordinates on $\be T^{*}M$ associated   to the frame $(\dfrac{dx}{x}, dy)$ are denoted by $(x, y, \zeta)$,   $\zeta= (\xi, \eta)$, corresponding to writing $\omega\in \be T^{*}M$ as $\xi\frac{dx}{x}+ \eta dy$. 

Similarly, passing to the cylinder picture, over $M_{\infty\:{\rm cyl}}$ the coordinates on $T^{*}M_{\rm cyl}$ associated to the frame $(ds, dy)$ are denoted by $(s, y,  \sigma, \eta)$ corresponding to writing $\omega\in T^{*}M_{\cyl}$ as $\sigma ds + \eta dy$.

 \subsubsection{Symbols on $\be T^{*}M$ with conormal regularity}\label{sec9.2.5}
 
 For $m\in \rr$ we denote by $S^{m}(\be T^{*}M)$ the space of functions $a$ on $\be T^{*}M$ such that 
 \beq\label{e9.1}
 \begin{aligned}
\p_{z}^{\alpha}\p_{\zeta}^{\beta}a(z, \zeta)&\in O(\langle \zeta\rangle^{m- |\beta|}),\ \ \alpha,\ \beta\in \nn^{n} \hbox{ over }M_{0},\\
(x\p_{x})^{\alpha}\p^{\beta}_{y}\p^{\gamma}_{\zeta}a(x, y, \zeta)&\in O(\langle \zeta\rangle^{m-|\gamma|}),\ \ (\alpha, \beta),\ \gamma\in \nn^{n}\hbox{ over }M_{\infty}.
\end{aligned}
  \eeq
  In other words $a$ has conormal regularity with respect to $\{x=0\}$  and symbolic estimates hold with respect to the b-dual variables $\zeta$.

We set $S^{-\infty}(\be T^{*}M)= \bigcap_{m\in \rr}S^{m}(\be T^{*}M)$, $S^{\infty}(\be T^{*}M)= \bigcup_{m\in \rr}S^{m}(\be T^{*}M)$. 
  
  For $r>0$ we denote by $S^{m}_{r}(\be T^{*}M)$ the subspace of $a\in S^{m}(\be T^{*}M)$ such that
  \[
  a(x, y, \zeta)= a_{\norm}(y, \zeta)+ x^{r}v(x, y, \zeta)\hbox{ over }M_{\infty},
  \]
  for $a_{\norm}\in \cinf(\be T^{*}M)$, $v\in S^{m}(\be T^{*}M)$.

    A symbol $a\in S^{m}(\be T^{*}M)$ is {\em classical} if there exist $a_{k}\in S^{m-k}(\be T^{*}M)$, homogeneous of degree $m-k$ in $\zeta$ in $|\zeta|\geq 1$ such that
\[
a- \sum_{k=0}^{N-1}a_{k}\in S^{m-N}(\be T^{*}M),\ \ \forall N\in \nn.
\]
We denote by $S^{m}_{\cl}(\be T^{*}M)$ the space of classical symbols. We note that while in e.g.~Sect.~\ref{secnot.1.1} ``classical'' refers to existence of an expansion at the boundary in the base space, here it refers to existence of an expansion at infinity in the fibers.

The quotient space $S^{m}_{\cl}(\be T^{*}M)/S^{m-1}_{\cl}(\be T^{*}M)$ is isomorphic to the space $S^{m}_{\rm h}(\be T^{*}M)$ of homogeneous symbols of degree $m$ in $\zeta$ and the image of $a$ under the quotient map is called the principal part of $a$ and denoted by $a_{\rm pr}$.
 
 We use the same notation for the subspaces $S^{m}_{\cl\:r}(\be T^{*}M)$ etc. 

%     For ease of notation we will in  the sequel  often omit the subscripts $\cl$. 

  \subsubsection{Symbol classes on $T^{*}M_{\cyl}$}\label{sec9.2.4}
   For $m\in \rr$  we  denote by $S^{m}_{\cl}(T^{*}M_{\cyl})$  the image of $S^{m}_{\cl}(\be T^{*}M)$ under the cylinder map. If $a\in S^{m}_{\cl}(T^{*}M_{\cyl})$   we have over $M_{\infty\:\cyl}$:
    \[
   \p_{s}^{\alpha}\p_{y}^{\beta}\p_{\zeta}^{\gamma}a(s, y, \sigma, \eta)\in O(\langle (\sigma, \eta)\rangle^{m- |\gamma|}), (\alpha, \beta), \gamma\in \nn^{n},
   \]
   so $S^{m}_{\cl}(T^{*}M_{\cyl})$ is the standard  class of uniform classical symbols on $T^{*}M_{\cyl}$.

 \subsection{$\b$-quantization}\label{sec9.3}
   
  We fix a quantization map $\Op_{\infty}$ on $M_{\norm\:\cyl}= \rr_{s}\times Y_y$:
   \[
   \Op_{\infty}: S^{m}(T^{*}M_{\norm\:\cyl})\to B(\cS(M_{\norm\:\cyl}),
   \]
   where $\cS(M_{\norm\:\cyl})$ is the Schwartz class on $M_{\norm\:\cyl}$, by
   \beq\label{app.e2}
   \Op_{\infty}(a)u(s, \cdot)= (2\pi)^{-1}\int \e^{\i(s- s')\sigma}\chi(s-s')\Op_{Y}(a)(s', \sigma)u(s', \cdot)ds'd\sigma,
   \eeq
   where  $\Op_{Y}(a)(s', \sigma)$ is a quantization map of $a(s, \sigma, \cdot)$  on  the compact manifold $Y$ constructed using  a finite atlas of $Y$  and $\chi\in \coinf(\rr)$ equals to $1$ near $0$.
   
   Thanks to the addition of the cutoff function $\chi(s-s')$ in \eqref{app.e2} we see that if $a\in S^{m}(T^{*}M_{\norm\:\cyl})$ then
   \beq\label{e9.3}
   \Op_{\infty}(a)\in B(\e^{-r|s|}H^{p}(M_{\norm\:\cyl}), \e^{-r|s|}H^{p-m}(M_{\norm\:\cyl})),\ \ p,\ r\in \rr.
   \eeq

    Similarly we fix a quantization map on the compact piece $M_{0}$:
\[
\Op_{0}: S^{\infty}(T^{*}M_{0})\to B(\coinf(M_{0}))
\]
We also fix  extension maps $E_{1}: \coinf(M_{\infty})\to \coinf(M_{\norm\:\cyl})$,  $E_{2}: S^{\infty}_{\cl}(T^{*}M_{\infty\:\cyl})\to S^{\infty}_{\cl}(T^{*}M_{\norm\:\cyl})$ and a restriction map $R_{1}:\cinf(M_{\norm\:\cyl})\to \cinf(M_{\infty})$. 
 
 Finally we fix a  partition of unity $1= \chi_{0}^{2}+ \chi_{\infty}^{2}$, where  $\chi_{0}\in \coinf(M_{0})$ and $\chi_{\infty}$ depends only on $s$ in $M_{\infty\:\cyl}$. Using the identification of $M_{\infty\:\cyl}$ with a subset of $M_{\norm\:\cyl}$, we set
\beq\label{e9.3b}
\begin{aligned}
\Op &: S^{\infty}_{\cl}(T^{*}M_{\cyl})\to  B(\coinf(M_{\cyl}))\\
&a\mapsto \chi_{0}\Op_{0}(a)\chi_{0}+ \chi_{\infty}R_{1}\Op_{\infty}(E_{2}a)E_{1}\chi_{\infty}.
\end{aligned}
\eeq
We identify $S_{\cl}^{\infty}(\be T^{*}M)$ with $S^{\infty}_{\cl}(T^{*}M_{\cyl})$ and $\coinf(M)$ with $\coinf(M_{\cyl})$ 
using the cylinder map. Using these identifications we  obtain  the quantization map:
\[
\Op: S^{\infty}_{\cl}(\be T^{*}M)\to B(\coinf(M)).
\]
\begin{definition}\label{def9.0}
For  $m\in \rr\cup\{\pm \infty\}$ and $r>0$  we set
\[
\Psi^{m}_{\rm s}(M)\defeq \Op(S^{m}_{\cl}(\be T^{*}M)),\quad \Psi^{m}_{{\rm s}\: r}(M)= \Op(S^{m}_{\cl\: r}(\be T^{*}M)).
\]
 \end{definition}
  If $A= \Op (a)\in\operw{m}{r}$, its principal symbol is  $\sigma_{\rm pr}(A)= a_{\rm pr}$. 
Usually in the $\b$-calculus framework, a smaller symbol space $S^{m}_{\b}(\be T^{*}M)$ is used where instead of \eqref{e9.1} one requires
\[
 \begin{aligned}
\p_{z}^{\alpha}\p_{\zeta}^{\beta}a(z, \zeta)&\in O(\langle \zeta\rangle^{m- |\beta|}), \alpha, \beta\in \nn^{n} \hbox{ over }M_{0},\\
\p_{x}^{\alpha}\p^{\beta}_{y}\p^{\gamma}_{\zeta}a(x, y, \zeta)&\in O(\langle \zeta\rangle^{m-|\gamma|}), (\alpha, \beta),\gamma\in \nn^{n}, \hbox{ over }M_{\infty},
\end{aligned}
  \]
i.e.~the symbols in $S^{m}_{\b}(\be T^{*}M)$ are smooth up to $\{x=0\}$ (``classical'' not only in the frequency variables but also in terms of behavior at the spacetime boundary). The space of operators on $M$ obtained by the same quantization map  is denoted by $\Psi^{m}_{\b\:{\rm s}}(M)$ and 
  sometimes called  the 'very small' $\b$-calculus.  
  
  We have  $\Diff^{m}_{\b}(M)\subset \Psi^{m}_{\b\:{\rm s}}(M)$ and $\Diff^{m}_{r}(M)\subset \operw{m}{r}$.

Because of the cutoffs and extension/restriction maps in \eqref{e9.3b}, $\oper{\infty}$ is not stable under products and adjoints. However one has
\begin{equation}\label{e9.3c}
\begin{aligned}
\Op(a)\circ \Op(b)&= \Op(a\moyal b)+ r_{-\infty},\\
\Op(a)^{*}&= \Op(Ta)+ r_{-\infty}
\end{aligned}
\end{equation}
where $\moyal: S^{m}_{\cl}(\be T^{*}M)\times S^{m'}_{\cl}(\be T^{*}M)\to S^{m+m'}_{\cl}(\be T^{*}M)$ is the Moyal product, $T: S^{m}_{\cl}(\be T^{*}M)\to S^{m}_{\cl}(\be T^{*}M)$ is antilinear with $Ta- \bar{a}\in S^{m-1}_{\cl}(\be T^{*}M)$, and $r_{-\infty}\in B^{-\infty}_{\rr}(M)$, see Def.~\ref{def9.1b} below.

\subsubsection{ Classes of operators}\label{sec9.2.1b}
To obtain a class  of operators stable under products and adjoints, it is necessary to add to $\oper{\infty}$ an ideal of smoothing operators. The action of this ideal on weighted Sobolev spaces is also important. 

In the sequel we fix an open interval $J\subset\rr$ defining a range of weight indices. We set $J(\epsilon)=\{r\in\rr:[r-\epsilon,r+\epsilon]\subset J\}$ for $0<\epsilon\ll 1$.
\begin{definition}\label{def9.1b}
 For $m\in \rr$  we set \[
B^{m}_{J}(M)= \bigcap_{p\in \rr, r\in J}B(H_{\b}^{p, r}(M); H_{\b}^{p-m, r}(M)), 
 \]
 \[
B^{-\infty}_{J}(M)= \bigcap_{m\in \rr}B^{m}_{J}(M).
 \]
 \end{definition}

We see that $\oper{m}\subset B ^{m}_{\rr}(M)$, and  changing the quantization map $\Op_{Y}$, the cutoff functions $\chi, \chi_{0}, \chi_{\infty}$ or the extension/restriction maps $E, R$ changes $\Op(a)$ by an error in $B^{-\infty}_{\rr}(M)$.
\def\logx{\langle \ln x\rangle}
The following facts are almost immediate:
\beq\label{e9.3d}
x^{-\alpha}B^{m}_{J}(M)x^{\alpha}= B^{m}_{J-\alpha}(M),\quad
x^{\alpha}B_{J}^{m}(M),\ B^{m}_{J}(M)x^{\alpha}\subset x^{\alpha/2}B^{m}_{J(\alpha/2)}(M)x^{\alpha/2}, 
\eeq
for $m\in \rr\cup\{-\infty\}$.

Moreover using symbolic calculus on $M_{\cyl}$ we easily obtain that if $a\in S^{m}_{\cl}(\be T^{*}M)$, then 
\[
\Op(xa)= x^{\12}\Op(\tilde{a})x^{\12}+ x^{\12}r_{-\infty}x^{\12},\ \ \tilde{a}\in S^{m}_{\cl}(\be T^{*}M),\ r_{-\infty}\in B_{\rr}^{-\infty}(M).
\]
We will also need spaces of operators that improve the decay at $\{x=0\}$. 

\begin{definition}\label{def9.1c}
 For $m\in \rr\cup \{-\infty\}$ we  set
 \[
 B^{m}_{J, 0^{+}}(M)= \bigcap_{\epsilon>0}\bigcup_{\delta>0}x^{\delta}B^{m}_{J(\epsilon)}(M)x^{\delta}.
 \]
 \end{definition}
In other words $a\in B^{m}_{J, 0^{+}}$ if for any $\epsilon>0$ there exists $\delta>0$ such that $a= x^{\delta}b x^{\delta}$ for $b\in B^{m}_{J(\epsilon)}(M)$. From \eqref{e9.3d} we see that
\begin{equation}
\label{e9.3dd}
\begin{aligned}
&B^{m_{1}}_{J}(M)B^{m_{2}}_{J, 0^{+}}(M),\ B^{m_{2}}_{J, 0^{+}}(M)B^{m_{1}}_{J}(M)
\subset B^{m_{1}+ m_{2}}_{J, 0^{+}}(M), \\
&x^{\alpha}B^{m}_{J}(M),\ B^{m}_{J}(M)x^{\alpha}
\subset B^{m}_{J, 0^{+}}(M),\ \ \forall \alpha>0.
\end{aligned}
\end{equation}

 \subsubsection{The case of  vector bundles}\label{sec9.3.1}

It is easy to extend  the calculus to vector bundles over $M$. We fix a  vector bundle $V\overset{\pi}{\to} M$ equipped with a Hilbertian structure as in \ref{sec9.2.2}.

The spaces $H_{\b}^{p, r}(M; V)$, $S^{m}_{(r)}(\be T^{*}M; L(V))$,  $\opersig{m}$, $\operwsig{m}{r}$, $B^{m}_{J}(M; V)$ and $B^{-\infty}_{J}(M; V)$ are then defined in the obvious way, using the Hilbertian norm on $V$.

\subsection{Symbolic fractional powers}\label{sec9.2b}
\subsubsection{$\lambda$-dependent $\b$-quantization}   \label{sec9.2b.1}
It is straightforward to consider symbols depending on a large real parameter $\lambda$.

In fact let us  denote by $l\in \rr$ a variable dual to $\lambda$ i.e.~let us write   the Fourier transform on $\rr_{l}$ as $\cF u(\lambda)= \int_{\rr}\e^{-\i l\lambda}u(l)dl$. We replace  $M$ by $\widetilde{M}= M\times\rr_{l}$ i.e.~replace $Y$  by $Y\times \rr_{l}$. Because of its simple product structure, the fact that $Y\times \rr_{l}$ is not compact does not cause any problem.

We denote by $\widetilde{S}_{\cl}^{m}(\be T^{*}M)$ resp.~$\widetilde{S}_{\cl\: r}^{m}(\be T^{*}M)$ the subspace of $S_{\cl}^{m}(\be T^{*}\widetilde{M})$ resp.~$S_{\cl\: r}^{m}(\be T^{*}\widetilde{M})$of symbols independent of $l$. 

The spaces $\Op(\widetilde{S}_{\cl}^{m}(T^{*}M))$ resp.~$\Op(\widetilde{S}_{\cl\: r}^{m}(T^{*}M))$ are denoted by $\widetilde{\Psi}^{m}_{{\rm s}}(M)$ resp.~$\widetilde{\Psi}^{m}_{{\rm s}\: r}(M)$. 

After conjugation  by the  Fourier transform $\cF$, it  coincides with the subspace of $\Psi^{m}_{{\rm s}}(\widetilde{M})$ resp.~$\Psi^{m}_{{\rm s}\: r}(\widetilde{M})$ of operators commuting with translations in $l$. 
This commutation property is preserved by inverses, parametrices etc. 
\subsubsection{Ideal of $\lambda$-dependent smoothing operators}\label{sec9.2b.2}
We define the sets $\widetilde{B}^{-\infty}_{J}(M)$ as the set of maps $\rr\ni \lambda\mapsto W(\lambda)\in B^{-\infty}_{J}(M)$ such that   \[
\|W(\lambda)\|_{B(H^{p, r }(M), H^{p-m, r}(M))}\in O(\langle\lambda\rangle^{-m}),\ \ p\in \rr,\ m\in \nn,\ r\in J.
\]

\subsubsection{Symbolic fractional powers}\label{sec9.2b.3}
For $0<\alpha<1, a, m>0$ we have:
\beq\label{e9.0}
a^{-\alpha}= C_{\alpha}^{-1}\int_{0}^{+\infty}\lambda^{-(1+\alpha m)+ m}(a+ \lambda^{m})^{-1}d\lambda,
\eeq
where $C_{\alpha}= m^{-1}\int_{0}^{+\infty}(1+ t)^{-1}t^{-\alpha}dt$. 

Let now $a\in \oper{m}$ for $m>0$ be  elliptic with $\Re \sigma_{\rm pr}(a)\geq C|\zeta|^{m}$ for $C>0$. Then  $a+ \lambda^{m}$ is elliptic in $\widetilde{\Psi}^{m}_{{\rm s}}(M)$ and  we can construct  a parametrix $(a+ \lambda^{m})^{(-1)}\in \widetilde{\Psi}^{-m}_{{\rm s}}(M)$
such that $(a+ \lambda^{m})\circ (a+ \lambda^{m})^{(-1)}= \one + \widetilde{B}^{-\infty}_{\rr}(M)$. 

If $a\in \operw{m}{r}$ then $a+ \lambda^{m}\in \widetilde{\Psi}^{m}_{{\rm s}\: r}(M)$ and $(a+ \lambda^{m})^{(-1)}\in \widetilde{\Psi}^{-m}_{{\rm s}\:r}(M)$

Following \eqref{e9.0} we define  for $0<\alpha<1$ the 'symbolic fractional power' $a^{(-\alpha)}$ by
\[
a^{(-\alpha)}\defeq C_{\alpha}^{-1}\int_{0}^{+\infty}\lambda^{-(1+\alpha m)+ m}(a+ \lambda^{m})^{(-1)}d\lambda.
\]
Similarly we set $a^{(n- \alpha)}= a^{n}\circ a^{(-\alpha)}$ for $\alpha>0, n\in \nn$ and 
we obtain that  $a^{(\alpha)}\in \oper{m\alpha}+ B^{-\infty}_{\rr}(M)$ for $\alpha\in \rr$, with $\sigma_{\rm pr}(a^{(\alpha)})= (\sigma_{\rm pr}(a))^{\alpha}$. If $a\in \operw{m}{r}$, then similarly $a^{(\alpha)}\in \operw{m\alpha}{r}+ B^{-\infty}_{\rr}(M)$.

In particular if $a\in\Diff^{2}_{\b}(M)$ is elliptic with $\Re \sigma_{\rm pr}(a)\geq C|(\sigma, \eta)|^{2}$ for $C>0$, we have
\[
a^{(-\12)}= 2\pi^{-1}\int_{0}^{+\infty}(a+ \lambda^{2})^{(-1)}d\lambda,\quad a^{(\12)}= a a^{(-\12)}, 
\]
and $(a^{(\12)})^{2}= a+ B^{-\infty}_{\rr}(M)$.

\subsection{Normal operators}\label{sec9.4}
We recall that $M_{\norm}$ and $d\mu_{\norm}$ are  defined in Subsect.~\ref{sec9.1}.

If $V\overset{\pi}{\to} M$ is a Hermitian bundle  we denote by $V_{\norm}\overset{\pi}{\to} M_{\norm}$ the same bundle equipped with the fiber Hermitian form $h_{\norm}(y)$,  see \ref{sec9.2.2}. 
 
We denote by $S^{m}(\be T^{*}M_{\norm})$ the space  of functions $a_{\norm}\in \cinf(\overline{M_{\norm}}\times \rr^{n}_{\zeta})$ {\em independent of} $x$ such that
\[
\p_{y}^{\beta}\p_{\zeta}^{\gamma}a_{\norm}(y, \zeta)\in O(\langle \zeta\rangle^{m-|\gamma|}),\ \ \beta\in \nn^{n-1},\ \gamma\in \nn^{n},
\]
and by $S^{m}_{\cl}(\be T^{*}M_{\norm})$ the corresponding subspace of classical symbols.

Symbols in $S^{m}(\be T^{*}M_{\norm})$ are quantized with the quantization map $\Op_{\norm}$ identical to $\Op_{\infty}$ in Subsect.~\ref{sec9.3} using the cylinder map $M_{\norm\:\cyl}\tosim M_{\norm}$. We have
\[
T_{s}^{-1}\Op_{\norm}(a_{\norm})T_{s}= \Op_{\norm}(a_{\norm}),\ \ a_{\norm}\in S^{m}(\be T^{*}M_{\norm}),
\]
and we denote by $\opern{m}$ the image of $S^{m}_{\rm cl}(\be T^{*}M_{\norm})$ under $\Op_{\norm}$ and the appropriate cylinder maps. Operators in $\opern{m}$ are normal models for operators in $\Psi^{m}_{{\rm s}\:1}(M)$, including those whose coefficients are only conormal at the boundary.

The weighted Sobolev spaces $H^{m, r}_{\b}(M_{\norm})$ are defined as before. We set also  
\[
a\in B^{m}_{J}(M_{\norm})\hbox{ if }a\in \bigcap_{p\in \rr, r\in J}B(H_{\b}^{p, r}(M_{\norm}); H^{p-m, r}_{\b}(M_{\norm})), 
 \]
 \[
B^{-\infty}_{J}(M_{\norm})= \bigcap_{m\in \rr}B^{m}_{J}(M_{\norm}),
 \]
 \[
 B^{m}_{J, 0^{+}}(M_{\norm})= \bigcap_{\epsilon>0}\bigcup_{\delta>0}x^{\delta}B^{m}_{J(\epsilon)}(M_{\norm})x^{\delta}.
 \]
 
\subsubsection{Normal operators}

If $P\in \Diff^{m}_{1}(M)$ then its normal operator $P_{\norm}$ acts on $M_{\norm}$. To obtain a differential operator acting on $M$, it is necessary to insert appropriate cutoff functions near $\p M$ and embeddings $M_{\infty}\hookrightarrow M_{\norm}$.  In the sequel  $\chi_{\infty}$ will denote any cutoff function in $\cinf(M)$ such that $\chi_{\infty}= 1$ near $\p M$ and $1- \chi_{\infty}\in \coinf(M)$. The embedding $M_{\infty}\hookrightarrow M_{\norm}$ will be forgotten to lighten the notation. 

We first need to define an appropriate class of normal operators. As before we fix an open interval $J$ as in \ref{sec9.2.1b}. 
\begin{definition}\label{def9.2}
 We denote by $N^{m}_{J}(M_{\norm})$ the set of dilation-invariant operators $n$ ($T_s^{-1}nT_s=n$) acting on $M_{\norm}$ of the form
 \[
 n= \tilde{n}+ r_{-\infty},\ \ \tilde{n}\in \opern{m},\ r_{-\infty}\in B^{-\infty}_{J}(M_{\norm}).
 \]
 where  for any $h\in \coinf(\rr^{+}_{x})$ with $h=1$ near $x=0$ one has:
 \[
 \ad_{h}r_{-\infty}\in B^{-\infty}_{J, 0^{+}}(M_{\norm}).
 \]
 An element of $N^{m}_{J}(M_{\norm})$ is called a {\em normal operator} of order $m$.
 
 A normal operator $n$ is {\em elliptic} if $\tilde{n}$ is elliptic, i.e.~if $|\sigma_{\rm pr}(\tilde{n})(y, \zeta)|\geq C |\zeta|^{m}$ for $y\in Y,\ \zeta\in \rr^{n}$.
 \end{definition}

 Here $\ad_{b}a= [b, a]$ is the commutator.  Clearly $\opern{m}\subset N^{m}_{\rr}(M_{\norm})$. Changing $\tilde{n}$ by an element of $\opern{-\infty}$ does not change the properties of $r_{-\infty}$.
 
 The condition on $\ad_{h}r_{-\infty}$ expresses the approximate locality of $n$. Its usefulness will be seen in the sequel.  In the lemma below $(a^{*})_{\norm}$ denotes the adjoint of $a$ in $L^{2}(M_{\norm}, d\mu_{\norm}; (\cdot| \cdot)_{V_{\norm}})$.
 \begin{lemma}\label{lemma9.1}
 \ben
 \item if $n\in N^{m}_{ J}(M_{\norm})$  then $n^{*}\in N^{m}_{-J}(M_{\norm})$.
 \item 
 If $n_{i}\in N^{m_{i}}_{ J}(M_{\norm})$ then  $n_{1}n_{2}\in N^{m_{1}+ m_{2}}_{ J}(M_{\norm})$
 \item if $n\in N^{m}_{ J}(M_{\norm})$ is  elliptic  with $n: H^{p, r}_{\b}(M_{\norm})\tosim H^{p-m, r}_{\b}(M_{\norm})$ for $p\in \rr,\ r\in J$, then $n^{-1}\in N^{-m}_{ J}(M_{\norm})$.
\item $x^{-\alpha}N^{m}_{ J}(M_{\norm})x^{\alpha}= N^{m}_{ J-\alpha}(M_{\norm})$.
\een
\end{lemma}
\proof Statement (4) is obvious. 
Let $n= \tilde{n}+ r_{-\infty} \in N^{m}_{ J}(M_{\norm})$.  
Weighted Sobolev duality gives $(B^{-\infty}_{J}(M_{\norm}))^{*}=B^{-\infty}_{-J}(M_{\norm})$, hence $(r_{-\infty})^{*}_{\norm}\in B^{-\infty}_{-J}(M_{\norm})$. By symbolic calculus we obtain then that  
$(\tilde{n})^{*}_{\norm}= \tilde{b}+ \tilde{r}_{-\infty}$, where $\tilde{b}\in \opern{m}$ and $\tilde{r}_{-\infty}\in B^{-\infty}_{\rr}(M_{\norm})$ with $\ad_{h}\tilde{r}_{-\infty}\in x^{\12}B^{-\infty}_{\rr}(M_{\norm})x^{\12}$.

It remains to prove that $\ad_{h}(r_{-\infty})^{*}_{\norm}\in B^{-\infty}_{-J,0^{+}}(M_{\norm})$. This follows from $(B^{-\infty}_{J,0^{+}}(M_{\norm}))^{*}=B^{-\infty}_{-J,0^{+}}(M_{\norm})$ and $\ad_h(r_{-\infty})^{*}_{\norm}=-(\ad_{\bar h}r_{-\infty})^{*}_{\norm}$.

Let us now prove (2). If $n_{i}= \tilde{n}_{i}+ r_{-\infty, i}\in N^{m_{i}}_{J}(M_{\norm})$ then 
\beq\label{e9.-1}
n_{1}n_{2}=\tilde{n}_{1}\tilde{n}_{2}+\tilde{n}_{1}r_{-\infty, 2}+ r_{-\infty, 1}\tilde{n}_{2}+ r_{-\infty, 1}r_{-\infty, 2}.
\eeq
By symbolic calculus we have $\tilde{n}_{1}\tilde{n}_{2}= \tilde{n}_{3}+ r_{-\infty, 3}$ for $\tilde{n}_{3}\in \opern{m_{1}+ m_{2}}$ and $r_{-\infty, 3}\in B^{-\infty}_{\rr}(M_{\norm})$, $\ad_{h}r_{-\infty, 3}\in x^{\12}B^{-\infty}_{\rr}(M_{\norm})x^{\12}\subset B^{-\infty}_{J, 0^{+}}(M_{\norm})$.
 Therefore it suffices to consider the other terms in the RHS of \eqref{e9.-1}. Each of these terms belongs to $B^{-\infty}_{J}(M_{\norm})$. Using that $\ad_{h}\tilde{n}_{i}\in x^{\12}B^{m_{i}-1}_{\rr}(M_{\norm})x^{\12}$, we obtain similarly that their commutators with $h$ belong to $B^{-\infty}_{J, 0^{+}}(M_{\norm})$. This completes the proof of (2).
 
 Let us now prove (3). Let $\tilde{n}^{(-1)}\in \opern{-m}$ a parametrix for $\tilde{n}$, so that
 $\tilde{n}\tilde{n}^{(-1)}= \one +\tilde{r}_{-\infty}$, for $\tilde{r}_{-\infty}\in B^{-\infty}_{\rr}(M_{\norm})$.  We have
$ n\tilde{n}^{(-1)}= \one + \tilde{r}_{-\infty} + r_{-\infty}\tilde{n}^{(-1)}$, hence
\[
n^{-1}= \tilde{n}^{(-1)}- n^{-1}(\tilde{r}_{-\infty}+ r_{-\infty}\tilde{n}^{(-1)})\eqdef \tilde{n}^{(-1)}+ n^{-1}\hat{r}_{-\infty}.
\]
We first note that $\hat{r}_{-\infty}$ and hence $n^{-1}\hat{r}_{-\infty}$ belong to $B^{-\infty}_{J}(M_{\norm})$.
Next we have
\beq\label{e9.4b}
\begin{aligned}
\ad_{h}\tilde{n}&\in x^{\12}B^{m-1}_{\rr}(M)x^{\12},\quad \ad_{h}\tilde{n}^{(-1)}\in x^{\12}B^{-m-1}_{\rr}(M)x^{\12},\\
\ad_{h}\tilde{r}_{-\infty}&\in x^{\12}B^{-\infty}_{\rr}(M_{\norm})x^{\12},\quad \ad_{h}r_{-\infty}\in B^{-\infty}_{J, 0^{+}}(M_{\norm}),
\end{aligned}
\eeq
which implies that $\ad_{h}\hat{r}_{-\infty}\in B^{-\infty}_{J, 0^{+}}(M_{\norm})$. We deduce also from \eqref{e9.4b} that $\ad_{h}n\in B^{m-1}_{J, 0^{+}}(M_{\norm})$. We write $\ad_{h} n^{-1}= - n^{-1}\ad_{h}n n^{-1}$. Using that  $n^{-1}\in B^{-m}_{J}(M_{\norm})$ and
 $x^{-\delta}B^{-m}_{J}(M_{\norm})x^{\delta}= B^{-m}_{J-\delta}(M_{\norm})$, as in \eqref{e9.3d}, we obtain that $\ad_{h}n^{-1}\in B^{-m-1}_{J, 0^{+}}(M_{\norm})$. This implies that $\ad_{h}(n^{-1}\hat{r}_{-\infty})\in B^{-\infty}_{J, 0^{+}}(M_{\norm})$ and completes the proof of (3). \qed

\subsection{Extended calculus}\label{sec9.4b}

We recall from \ref{sec9.2.1}  that $M$ is equipped with a density $d\mu$ as in \ref{sec9.2.1}  and that the  vector bundle $V\overset{\pi}{\to}M$ is equipped with a fiber Hermitian structure $h$  as in \ref{sec9.2.2}.

\begin{definition}\label{def9.3} 
Let $m\in \rr$ and $J\subset \rr$ an open interval as in \ref{sec9.2.1b}. We denote by $\Psi^{m}_{J}(M)$ the set of operators $a$  acting on $M$ such that there exists  $a_{\norm}\in N^{m}_{ J}(M_{\norm})$, called   the {\em normal operator }of $a$ such that 
\beq\label{bilie1}
 a= \chi_{\infty}a_{\norm}\chi_{\infty}+ x^{\12}rx^{\12}+r_{-\infty},
 \eeq
 where $r\in \Psi^{m}_{\rm s}(M)$ and $r_{-\infty}\in B^{-\infty}_{J, 0^{+}}(M)$.
 
From Def.~\ref{def9.2} we have $a_{\norm}= \tilde{a}_{\norm}+ \tilde{r}_{-\infty\norm}$ with $\tilde{a}_{\norm}\in \opern{m}$  and $\tilde{r}_{-\infty\norm}\in B^{-\infty}_{J}(M_{\norm})$. We set $\tilde{a}= \chi_{\infty}\tilde{a}_{\norm}\chi_{\infty}+ x^{\12}rx^{\12}$ so that
\beq\label{bilie2}
a= \tilde{a}+ \tilde{r}_{-\infty},\ \ \tilde{a}\in \operw{m}{1},\ \tilde{r}_{-\infty}\in B^{-\infty}_{J}(M).
\eeq
We define the {\em principal symbol} of $a$ as $\sigma_{\rm pr}(a)(x, y, \zeta)\defeq \sigma_{\rm pr}(\tilde{a})(x, y, \zeta)$ and  say that $a$ is {\em elliptic} if 
 \[
 |\sigma_{\rm pr}(\tilde{a})|(x, y, \zeta)\geq C|\zeta|^{m}\hbox{ for }(x, y, \zeta)\in M\times \rr^{n}.
 \]
\end{definition}
 
\begin{lemma}\label{lemma9.2b}
 \ben
 \item The normal operator $a_{\norm}$ of $a$ is unique.
 \item If  $a\in \Psi^{m}_{J}(M)$, then $\sigma_{\rm pr}(a_{\norm})(y, \zeta)= \sigma_{\rm pr}(a)(0, y, \zeta)$. In particular  if $a$ is elliptic then $a_{\norm}$ is elliptic.
 \item One can replace $\chi_{\infty}$ in Def.~\ref{def9.3} (2) by any similar cutoff function $\tilde{\chi}_{\infty}$.
 \item  $\operw{m}{1}\subset \Psi^{m}_{ \rr}(M)$.
 \item $x^{-\alpha}\Psi^{m}_{J}(M)x^{\alpha}= \Psi^{m}_{J-\alpha}$.
 \een
\end{lemma}
\proof  (4) and (5)  are obvious. 
 
 To prove (1) we note that if $u, v\in \coinf(M_{\norm})$, then  $T_{s}u, T_{s}v\in \coinf(M)$ for $s\gg 1$ and 
$(u|a_{\norm}v)_{M_{\norm}}= \lim_{s\to \infty}(T_{s}u |a T_{s}v)_{M}$. 

To prove (2) we use that $a= \tilde{a}+ r_{-\infty}= \chi_{\infty}a_{\norm}\chi_{\infty}+ x^{\beta}\tilde{r}x^{\beta}+ \tilde{r}_{-\infty}$, so $\sigma_{\rm pr}(a)= \sigma_{\rm pr}(\tilde{a})= \chi_{\infty}^{2}\sigma_{\rm pr}(a_{\norm})+ x^{2\beta}\sigma_{\rm pr}(\tilde{r})$. 
Finally  
 \[
 (\chi_{\infty}-\tilde{\chi}_{\infty })a_{\norm}\chi_{\infty}= (\chi_{\infty}-\tilde{\chi}_{\infty })\tilde{a}_{\norm}\chi_{\infty}+ (\chi_{\infty}-\tilde{\chi}_{\infty })r_{-\infty}\chi_{\infty}.
 \]
The first term in the RHS  belongs to $x^{N}\oper{m}\subset x^{N/2}\oper{m}x^{N/2}+ x^{N/2}B^{-\infty}_{\rr}(M)x^{N/2}$ for any $N\in \nn$. The second term belongs to $x^{N}B^{-\infty}_{J}(M)\subset B^{-\infty}_{J, 0^{+}}(M)$. We use the same argument for $\chi_{\infty}a_{\norm}(\chi_{\infty}- \tilde{\chi}_{\infty})$, which proves (3). \qed

 \begin{remark}\label{rem9.2}
 The definition of  $a_{\norm}$ implies that $(x^{r}a)_{\norm}= 0$  if $r>0$. In other words   the definition of  $a_{\norm}$ is useful only if  $a_{\norm}\neq 0$. This fact is similar to the notion of the principal symbol  of a pseudodifferential operator on a compact manifold: if $a\in \Psi^{m}(M)$ its principal symbol in    $S^{m'}_{\cl }(T^{*}M)$ is equal to $0$ for $m'>m$.
 
% Therefore we extend the definition of the normal operator to $x^{\alpha}\Psi^{m}_{b, J}(M)$ by setting:
% \[
% (x^{\alpha}a)_{\norm}\defeq x^{\alpha} a_{\norm}, \ a\in \Psi^{m}_{ J}(M).
% \]
 \end{remark}

 \subsubsection{Adjoints}
 
\begin{lemma}\label{lemma9.2a}
 If $a\in \Psi^{m}_{ J}(M)$ then $a^{*}\in \Psi^{m}_{-J}(M)$ with
 \[
( a^{*})_{\norm}= (a_{\norm})^{*},
 \]
 where $(a_{\norm})^{*}$ denotes the adjoint of $a_{\norm}$ in $L^{2}(M_{\norm}, d\mu_{\norm}; (\cdot| \cdot)_{V_{\norm}})$.
 \end{lemma}
 \proof In the proof we denote  by $(a_{\norm})^{*}_{\norm}$ resp.~$(a_{\norm})^{*}$ the adjoint of $a_{\norm}$ in $L^{2}(M_{\norm}, d\mu_{\norm}; (\cdot| \cdot)_{V_{\norm}})$ resp.~$L^{2}(M, d\mu; (\cdot| \cdot)_{V})$. We have:
 \[
 a^{*}= \chi_{\infty}(a_{\norm})^{*}\chi_{\infty}+ x^{\12}r^{*}x^{\12}+ r_{-\infty}^{*}.
 \]
Next $r_{-\infty}^{*}\in B^{-\infty}_{-J,0^{+}}(M)$ by weighted Sobolev duality, $r^{*}$ belongs to $\oper{m}+ B^{-\infty}_{\rr}(M)$ 
 and $x^{\12}B^{-\infty}_{\rr}(M)x^{\12}\subset B^{-\infty}_{-J,0^{+}}(M)$.
 
 It remains to examine  $(a_{\norm})^{*}$.  We have
$ (a_{\norm})_{V}^{*}= T(a_{\norm})_{V_{\norm}}^{*}T^{-1}$ for $T\in\mathcal{A}^{\infty}_{1}(M)$ with $T_{\norm}= \one$, so 
\[
(a_{\norm})^{*}-(a_{\norm})^{*}_{\norm}
=(T-\one)(a_{\norm})^{*}_{\norm}T^{-1}
 +(a_{\norm})^{*}_{\norm}(T^{-1}-\one).
\]
Since $T-\one,T^{-1}-\one\in x\mathcal A^\infty(M)$, symbolic calculus and $(a_{\norm})^{*}_{\norm}\in N^m_{-J}(M_{\norm})$ give
\[
\chi_{\infty}\big((a_{\norm})^{*}-(a_{\norm})^{*}_{\norm}\big)\chi_{\infty}
\in x^{\12}\oper{m}x^{\12}+xB^{-\infty}_{-J}(M)+B^{-\infty}_{-J}(M)x.
\]
By \eqref{e9.3dd}, $xB^{-\infty}_{-J}(M), B^{-\infty}_{-J}(M)x\subset B^{-\infty}_{-J,0^{+}}(M)$. We obtain finally that $a^{*}\in \Psi^{m}_{-J}(M)$ with $(a^{*})_{\norm}= (a_{\norm})^{*}_{\norm}$ as claimed. 
 \qed

\subsubsection{Composition}
 \def\thing{B^{-\infty}_{J, 0^{+}}(M)}
\begin{lemma}\label{lemma9.2}
 If $a_{i}\in \Psi^{m_{i}}_{J}(M)$, $i= 1, 2$ then $a_{1}a_{2}\in \Psi^{m_{1}+ m_{2}}_{J}(M)$ with 
 \[
 (a_{1}a_{2})_{\norm}= a_{1\norm}a_{2\norm}.
 \]
\end{lemma}
\proof
We have $a_{i}= \chi_{\infty}n_{i}\chi_{\infty}+ x^{\12}r_{i}x^{\12}+ r_{-\infty, i}$ for $n_{i}$ normal operator of order $m_{i}$, $r_{i}\in\oper{m_{i}}$ and $r_{-\infty, i}\in B^{-\infty}_{J, 0^{+}}(M)$. Hence
\[
\begin{aligned}
a_{1}a_{2}&= \chi_{\infty}n_{1}\chi_{\infty}^{2}n_{2}\chi_{\infty}\\
&\phantom{=\,}+ \chi_{\infty}n_{1}\chi_{\infty}x^{\12}r_{2}x^{\12}+ x^{\12}r_{1}x^{\12}\chi_{\infty}n_{2}\chi_{\infty}\\
&\phantom{=\,}+ \chi_{\infty}n_{1}\chi_{\infty}r_{-\infty, 2}+ r_{-\infty, 1}\chi_{\infty}n_{2}\chi_{\infty}\\
&\phantom{=\,}+ x^{\12}r_{1}xr_{2}x^{\12}+ x^{\12}r_{1}x^{\12}r_{-\infty, 2}+ r_{-\infty, 1}x^{\12}r_{2}x^{\12}+ r_{-\infty, 1}r_{-\infty, 2}\\
&=:I_{1}+ I_{2}+ I_{3}+ I_{4}.
\end{aligned}
\]
Let us first look at $I_{1}$. Since $n_{i}= \tilde{n}_{i}+ r_{-\infty, i}$, see Def.~\ref{def9.2} we have
\[
\begin{aligned}
I_{1}&=\chi_{\infty}^{3}n_{1}n_{2}\chi_{\infty}+ \chi_{\infty}n_{1}[\chi_{\infty}^{2}, n_{2}]\chi_{\infty}\\
&=\chi_{\infty}^{3}n_{1}n_{2}\chi_{\infty}+ \chi_{\infty}n_{1}[\chi_{\infty}^{2}, \tilde{n}_{2}]\chi_{\infty}+ \chi_{\infty}n_{1}[\chi_{\infty}^{2}, r_{-\infty, 2}]\chi_{\infty}\\
&=\chi_{\infty}^{3}n_{1}n_{2}\chi_{\infty}+ \chi_{\infty}\tilde{n}_{1}[\chi_{\infty}^{2}, \tilde{n}_{2}]\chi_{\infty}\\
&\phantom{=\,}+ \chi_{\infty}r_{-\infty, 1}[\chi_{\infty}^{2}, \tilde{n}_{2}]\chi_{\infty}+ \chi_{\infty}n_{1}[\chi_{\infty}^{2}, r_{-\infty, 2}]\chi_{\infty}.
\end{aligned}
\]
The term $\chi_{\infty}\tilde{n}_{1}[\chi_{\infty}^{2}, \tilde{n}_{2}]\chi_{\infty}$ belongs to $x^{N}\oper{m_{1}+ m_{2}}x^{N}+ x^{N}B^{-\infty}_{\rr}(M)x^{N}$ for any $N\in \nn$. The term $  \chi_{\infty}r_{-\infty, 1}[\chi_{\infty}^{2},  \tilde{n}_{2}]\chi_{\infty}$ belongs to $B^{-\infty}_{J}(M)x^{N}B^{m_{2}-1}_{J}(M)\subset B^{-\infty}_{J, 0^{+}}(M)$. The term $ \chi_{\infty}n_{1}[\chi_{\infty}^{2}, r_{-\infty, 2}]\chi_{\infty}$  belongs to  $B^{m_{1}}_{J}(M)B^{-\infty}_{J, 0^{+}}(M)\subset B^{-\infty}_{J, 0^{+}}(M)$ by \eqref{e9.3dd}.

Since $\chi_{0}= \chi_{\infty}^{3}- \chi_{\infty}\in \coinf(M)$ we obtain  that
 \[
 \chi_{\infty}^{3}n_{1}n_{2}\chi_{\infty}= \chi_{\infty}n_{1}n_{2}\chi_{\infty}+ x^{N}\oper{m_{1}+ m_{2}}+ x^{N}B^{-\infty}_{J}(M)
\]
hence using \eqref{e9.3d} we have finally:
\[
I_{1}= \chi_{\infty}n_{1}n_{2}\chi_{\infty}+ x^{\12}\oper{m_{1}+ m_{2}}x^{\12}+ \thing.
\]
Next we have
\[
\begin{aligned}
I_{2}&=\chi_{\infty}\tilde{n}_{1}\chi_{\infty}x^{\12}r_{2}x^{\12}+ x^{\12}r_{1}x^{\12}\chi_{\infty}\tilde{n}_{2}\chi_{\infty}\\
&\phantom{=\,}+\chi_{\infty}r_{-\infty, 1}\chi_{\infty}x^{\12}r_{2}x^{\12}+ x^{\12}r_{1}x^{\12}\chi_{\infty}r_{-\infty, 2}\chi_{\infty}.
\end{aligned}
\]
The terms in the first line in the RHS  belong to $x^{\12}\oper{m_{1}+ m_{2}}x^{\12}+ x^{\12}B^{-\infty}_{\rr}(M)x^{\12}$, and the terms in the second line to $x^{\12}B^{-\infty}_{J}(M)+ B^{-\infty}_{J}(M)x^{\12}\subset \thing$.

Since $n_{i}\in B_{J}^{m_{i}}(M)$ we obtain immediately that $I_{3}\in \thing$. The first term in $I_{4}$ belongs to $x^{\12}\oper{m_{1}+ m_{2}}x^{\12}+ x^{\12}B^{-\infty}_{\rr}(M)x^{\12}$, the other to $\thing$.
We obtain finally that
\[
a_{1}a_{2}= \chi_{\infty}n_{1}n_{2}\chi_{\infty}+ x^{\12}rx^{\12}+ r_{-\infty},
\]
for $r\in \oper{m_{1}+ m_{2}}$, $r_{-\infty}\in \thing$,  which completes the proof. \qed

\subsubsection{Inverses}
\begin{lemma}\label{lemma9.3}
 Let $a\in \Psi^{m}_{J}(M)$  be elliptic and such that 
 $a: H^{p, r}_{\b}(M)\tosim H^{p-m, r}_{\b}(M)$  and $a_{\norm}: H^{p, r}_{\b}(M_{\norm})\tosim H^{p-m, r}_{\b}(M_{\norm})$ for $p\in \rr,\ r\in J$. Then $a^{-1}\in \Psi^{-m}_{J}(M)$ with
 \[
 (a^{-1})_{\norm}= (a_{\norm})^{-1}.
 \]
  \end{lemma}
  Note that $a_{\norm}$ is elliptic by Lemma~\ref{lemma9.2b}. Since $a_{\norm}: H^{p, r}_{\b}(M_{\norm})\tosim H^{p-m, r}_{\b}(M_{\norm})$ for $p\in \rr,\ r\in J$ we know that $(a_{\norm})^{-1}$ is a normal operator by Lemma~\ref{lemma9.1} (3).
  
  \medskip

\noindent{\bf Proof of Lemma~\ref{lemma9.3}.}
We recall that $\tilde{a}\in \oper{m}$ is defined in \eqref{bilie2}.
Denoting by  $\tilde{a}^{(-1)}\in\oper{-m}$ a parametrix of $\tilde{a}$,  we have $\tilde{a}\tilde{a}^{(-1)}= \one+\tilde{r}_{-\infty}$, $a\tilde{a}^{(-1)}= \one + \tilde{r}_{-\infty, 1}$ for $\tilde{r}_{-\infty}\in B^{-\infty}_{\rr}(M), \tilde{r}_{-\infty,1}\in B^{-\infty}_{J}(M)$. Therefore
$a^{-1}= \tilde{a}^{(-1)}- a^{-1}\tilde{r}_{-\infty, 1}$. Since $a^{-1}\in B^{-m}_{J}(M)$ we obtain that 
\begin{equation}
\label{e9.6}
a^{-1}= \tilde{a}^{(-1)}+ r_{-\infty},\ \ r_{-\infty}\in B^{-\infty}_{J}(M).
\end{equation}

Next we apply Lemma~\ref{lemma9.2} to  the operators $a_{1}= a$, $a_{2}= \chi_{\infty}(a_{\norm})^{-1}\chi_{\infty}$.  We obtain that
\beq\label{e9.6b}
\begin{aligned}
a\chi_{\infty}(a_{\norm})^{-1}\chi_{\infty}
&= \chi_{\infty}^{2}- x^{\12}\tilde{r}x^{\12}- \tilde{r}_{-\infty},\\
&=1 - x^{\12}\tilde{r}_{1}x^{\12}- \tilde{r}_{-\infty},
\end{aligned}
\eeq
for $\tilde{r}_{1}\in \oper{0}$, $\tilde{r}_{-\infty}\in \thing$. Therefore using \eqref{e9.6}:
\[
\begin{aligned}
a^{-1}&= \chi_{\infty}(a_{\norm})^{-1}\chi_{\infty}+ a^{-1}x^{\12}\tilde{r}_{1}x^{\12}+ a^{-1}\tilde{r}_{-\infty}\\
&= \chi_{\infty}(a_{\norm})^{-1}\chi_{\infty}+ \tilde{a}^{(-1)}x^{\12}\tilde{r}_{1}x^{\12}\\
&\phantom{=\,}+r_{-\infty}x^{\12}\tilde{r}_{1}x^{\12}+ a^{-1}\tilde{r}_{-\infty},
\end{aligned}
\]
where $r_{-\infty}\in B^{-\infty}_{J}(M)$. Next $r_{-\infty}x^{\12}\tilde{r}_{1}x^{\12}\in B^{-\infty}_{J}(M)x^{\12}\subset \thing$ and $a^{-1}\tilde{r}_{-\infty}\in B^{-m}_{J}(M)\thing\subset \thing$ by \eqref{e9.3dd}.
Finally $\tilde{a}^{(-1)}x^{\12}\tilde{r}_{1}x^{\12}= x^{\12}\tilde{r}_{2}x^{\12}+ x^{\12} r_{-\infty}x^{\12}$, for $\tilde{r}_{2}\in \oper{-m}$ and $r_{-\infty}\in B^{-\infty}_{\rr}(M)$.  Therefore we have
\[
a^{-1}= \chi_{\infty}(a_{\norm})^{-1}\chi_{\infty}+ x^{\12}rx^{\12}+ r_{-\infty},\ \ r\in \oper{-m},\ r_{-\infty}\in \thing,
\]
which  completes the proof of the lemma. \qed

\subsection{Fredholm property  and essential spectrum}\label{sec9.5}
We recall that the notation $V$ is  defined in \ref{sec9.2.2}.

\subsubsection{Closed and self-adjoint operators}\label{sec9.5.1}
\begin{lemma}\label{lemma9.1b}
 Assume that $a\in \Psi^{m}_{J}(M; V)$ ($m>0$, $0\in J$) is {\em elliptic}, i.e.\ $\sigma_{\rm pr}(a)(x, y, \zeta)$ is invertible for $\zeta\neq 0$. Then $a$ with domain $H^{m}_{\b}(M; V)$ is closed. If  $a$  is formally self-adjoint on $\coinf(M; V)$ then it is self-adjoint on $H^{m}_{\b}(M; V)$.
\end{lemma}
\proof Let $u_{n}\in H^{m}_{\b}(M; V)$ be such that $u_{n}\to u$  and $au_{n}\to v$ in $L^{2}_{\b}(M; V)$.  As in the proof of Lemma~\ref{lemma9.3}, let   $\tilde{a}^{(-1)}\in  \Psi^{-m}_{{\rm s}}(M; V)$ with  $\tilde{a}^{(-1)}\circ a= \one + r_{-\infty}$, where $r_{-\infty}\in B^{-\infty}_{J}(M; V)$. We obtain that $\tilde{a}^{(-1)}v= u+ r_{-\infty}u$, hence $u\in  H^{m}_{\b}(M; V)$ and $au= v$.

If $a$ is symmetric on $\coinf(M; V)$ we obtain first that $a$ is also symmetric on $H^{m}_{\b}(M; V)$. To prove that $a$ is self-adjoint on $H^{m}_{\b}(M; V)$ we show that for $\lambda\gg 1$ $(a\pm \i\lambda^{m}): H^{m}_{\b}(M; V)\to L^{2}_{\b}(M; V)$ is surjective. 
 By the  $\Psi$DO calculus on $M_{\rm cyl}$ with large parameter as in Subsect.~\ref{sec9.2b}, we can find $a^{(-1)}(\lambda)\in \operbun{-m}$ such that $(a\pm\i\lambda^{m})a^{(-1)}(\lambda)=\one+r(\lambda)$ where $\|r(\lambda)\|\in o(1)$ in $B(L^{2}(M; V))$.  We obtain that  $(a\pm \i\lambda^{m}): H^{m}_{\b}(M; V)\to L^{2}(M; V)$ is surjective for $\lambda\gg 1$. \qed

\subsubsection{Fredholm property}\label{sec9.5.2}

\begin{lemma}\label{lemma9.4}
 Let $a\in \Psi^{m}_{J}(M; V)$ be  elliptic with $a_{\norm}: H^{p, r}_{\b}(M_{\norm})\tosim H^{p-m, r}_{\b}(M_{\norm})$ for $p\in \rr,\ r\in J$. Then $a: H^{p, r}_{\b}(M; V)\to H^{p-m, r}_{\b}(M; V)$ is Fredholm for $p\in \rr,\ r\in J$.
\end{lemma}
\proof   By \eqref{e9.6} and \eqref{e9.6b} we have
\[
a\tilde{a}^{(-1)}= \one - r_{-\infty},\quad a \chi_{\infty}(a_{\norm})^{-1}\chi_{\infty}= \one- x^{\12}\tilde{r}_{1}x^{\12}- \tilde{r}_{-\infty},
\]
where $r_{-\infty}\in B^{-\infty}_{J}(M)$, $\tilde{r}_{1}\in \oper{0}$ and $\tilde{r}_{-\infty}\in \thing$. Therefore
\[
a\left(\tilde{a}^{(-1)}+ \chi_{\infty}(a_{\norm})^{-1}\chi_{\infty}r_{-\infty}\right)= \one- x^{\12}\tilde{r}_{1}x^{\12}
 r_{-\infty}- \tilde{r}_{-\infty}r_{-\infty}= \one - \hat{r}_{-\infty}
\]
for $\hat{r}_{-\infty}\in \thing$.  The error term $\hat{r}_{-\infty}$ is compact in $H^{p, r}_{\b}(M)$ for all $p\in\rr,\ r\in J$. We can similarly construct a right pseudo-inverse of $a$ and obtain that $a: H_{\b}^{p, r}(M)\to H^{p-m, r}_{\b}(M)$ is Fredholm for $p\in \rr,\ r\in J$.  \qed

\subsubsection{Essential spectrum}\label{sec9.5.3}
If $P: \Dom P\to \cH$ is a closed densely defined operator, we say that $z\not\in \sigma_{\rm ess}(P)$ if $P-z: \Dom P\to \cH$ is Fredholm. 
From  Lemma~\ref{lemma9.4} we obtain immediately the following lemma. 
\begin{lemma}\label{lemma9.5}
 Let $a\in \Psi^{m}_{J}(M;V)$, with $m>0$ and $0\in J$, be elliptic. By Lemma~\ref{lemma9.1b} $a$ is closed on $\Dom a= H^{m}_{\b}(M; V)$ and $a_{\norm}$ is closed on $H^{m}_{\b}(M_{\norm}; V_{\norm})$. Then $\sigma_{\rm ess}(a)\subset  \sigma(a_{\norm})$.
 \end{lemma}
 
\begin{lemma}\label{lemma9.6}
 Let $a\in \Psi^{m}_{J}(M; V)$, with $m>0$ and $J=-J$, be elliptic with $a_{\norm}: H^{p, r}_{\b}(M_{\norm}; V_{\norm})\allowbreak\tosim H^{p-m, r}_{\b}(M_{\norm}; V_{\norm})$ for $p\in \rr,\ r\in J$ and $\Re a_{\norm}\geq C>0$ on $\coinf(M_{\norm}; V_{\norm})$ for the Hilbertian scalar product $(\cdot| \cdot)_{V_{\norm}(M_{\norm})}$.
 Assume also that, for some $c>0$,
 \[
 \Re\sigma_{\rm pr}(a)(x,y,\zeta)\geq c|\zeta|^m\one_V
 \quad\hbox{for }\zeta\neq0.
 \]
 Then there exists $s_{-\infty}\in B^{-\infty}_{\rr}(M)$ of finite rank with a smooth distributional kernel supported away from $\p M$ such that
 \[
 a+ s_{-\infty}: H^{p, r}_{\b}(M; V)\tosim   H^{p-m, r}_{\b}(M; V)\hbox{ for }p\in \rr,\ r\in J. 
 \]
 \end{lemma}
\proof By Lemma~\ref{lemma9.1b} we  know that $a$ is closed on $L^{2}(M; V)$ with domain $H^{m}_{\b}(M; V)$.  Let us consider $\Re a= \12 (a+ a^{*})$, which is elliptic and bounded below by the assumption on the principal symbol, and self-adjoint on $L^{2}(M; V)$ with domain $H^{m}_{\b}(M; V)$. By Lemma~\ref{lemma9.2a} $(\Re a)_{\norm}= \Re(a_{\norm})$ hence by Lemma~\ref{lemma9.5} $\sigma_{\rm ess}(\Re a)\subset [C, +\infty[$. 

We choose $\chi\in \coinf(]-\infty, C[)$ such that $\Re a+ \chi(\Re a)\geq 3C/4$,  for the Hilbertian scalar product in $L^{2}(M;V)$. The operator $\chi(\Re a)$ is of finite rank  and we  can choose a cutoff function $\psi\in \coinf(M)$ equal to $0$ near $\p M$ such that $\| \chi(\Re a)- \psi \chi(\Re a)\psi\|_{B(L^{2})}\ll 1$. Setting $s_{-\infty}= \psi \chi(\Re a)\psi$ we obtain that 
\beq\label{e9.7}
\Re a + s_{-\infty}\geq C/2. 
\eeq
The distributional kernel is supported away from $\p M$ and is smooth by elliptic regularity. In particular $s_{-\infty}\in B^{-\infty}_{\rr}(M)$ and since the kernel of $s_{-\infty}$ vanishes near $\p M$ we have $(a+ s_{-\infty})_{\norm}= a_{\norm}$. Therefore by   Lemma~\ref{lemma9.4} $a+ s_{-\infty}: H^{p, r}_{\b}(M; V)\to  H^{p-m, r}_{\b}(M; V)$ is Fredholm for $p\in \rr,\ r\in J$.  

In the sequel we denote $H^{p, r}_{\b}(M; V)$ simply by $H^{p, r}_{\b}$.
By elliptic regularity, if $u\in H^{p,r}_{\b}$ and $(a+s_{-\infty})u=0$ or $(a+s_{-\infty})^{*}u=0$, then $u$ belongs to the operator domain $H^{m,r}_{\b}$, and the usual bootstrapping gives the corresponding higher regularity. Thus the kernels can be tested in the domain spaces.

At weight zero, \eqref{e9.7} gives
\[
        0=\Re ((a+s_{-\infty})u|u)\geq C\|u\|^2/2
\]
for any $u\in\Ker(a+s_{-\infty})\subset H^{m,0}_{\b}$. The same argument applies to the Hilbert adjoint, since
$\Re(a+s_{-\infty})^{*}=\Re(a+s_{-\infty})$. Hence
\beq\label{bilie3}
        \Ker(a+s_{-\infty})_{|H^{p,0}_{\b}}
        =\Ker(a+s_{-\infty})^{*}_{|H^{p,0}_{\b}}
        =\{0\},\ \ p\in\rr .
\eeq
If $r\geq 0$, then $H^{p,r}_{\b}\subset H^{p,0}_{\b}$, so the kernels vanish at weight $r$ as well.

It remains to consider $r<0$.  By Def.~\ref{def9.3} we can write
\[
        a+s_{-\infty}= \chi_{\infty}a_{\norm}\chi_{\infty}
        +x^{\12}r_{\rm ps}x^{\12}+ r_{\rm sm},
\]
where $r_{\rm ps}\in\oper{m}$ and $r_{\rm sm}\in\thing$.  Fix $\epsilon>0$ such that $r\in J(\epsilon)=J'$.  By the definition of $\thing$, $r_{\rm sm}=x^{\delta}b_{-\infty}x^{\delta}$ for some $\delta>0$ and some $b_{-\infty}\in B^{-\infty}_{J'}(M)$.
If $u\in\Ker(a+s_{-\infty})_{|H^{p,r}_{\b}}$, then
\[
        x^{\12}r_{\rm ps}x^{\12}u\in H^{p-m,r+1}_{\b},
        \quad
        r_{\rm sm}u\in H^{\infty,\min(r+\delta,0)}_{\b}.
\]
The equation therefore implies, modulo compactly supported terms which are harmless for weights, that
$a_{\norm}\chi_{\infty}u\in H^{p-m,\min(r+\delta,0)}_{\b}$.  Since
$a_{\norm}:H^{p,s}_{\b}\to H^{p-m,s}_{\b}$ is invertible for all $s\in J$, we obtain
$\chi_{\infty}u\in H^{p,\min(r+\delta,0)}_{\b}$, and hence
$u\in H^{p,\min(r+\delta,0)}_{\b}$.  Iterating this improvement finitely many times gives
$u\in H^{p,0}_{\b}$, and then $u=0$ by \eqref{bilie3}.  The same argument applies to
$(a+s_{-\infty})^{*}$, using $(a+s_{-\infty})^{*}\in\Psi^m_J(M;V)$ and the same normal invertibility.  Thus both the kernel and cokernel of $a+s_{-\infty}$ vanish at every weight $r\in J$.  Since the operator is Fredholm, it is invertible.  This completes the proof of the lemma. \qed

\section{Approximate stability lemma}\label{sec:app-ode}\init
The reduction of the metric in Sect.~\ref{sec2.2} uses the flow of a normal vector field.  The following local lemma provides the stability at $\rho=0$ needed there: a conormal $\b$-vector field vanishing to positive order at the boundary produces a flow which preserves the boundary defining function and boundary coordinates up to the corresponding order.

Let $M$ be a smooth $(n-1)$-dimensional manifold with compact boundary $\p M=Y$, and let $\rho$ be a boundary defining function.  If $I\subset\rr$ is an interval containing $0$, set
\[
        \cA^\infty(I\times M)=C^\infty(I;\cA^\infty(M)),
\]
and use the same convention with $\cA^\infty$ replaced by $\cA^{\infty,r}$, or by $\cA^\infty_r$ when $r>0$.

\begin{lemma}\label{thm:ODE}
Let
\[
        v\in C^\infty(I;\cA^{\infty,\epsilon}(M;{}^{\b}TM)),
        \ \ \epsilon>0,
\]
be a time-dependent $\b$-vector field.  Let $\widetilde x(t,x)$ be the flow of $v$ with initial condition $\widetilde x(0,x)=x$.  After possibly shrinking $I$ to an interval $I'$ containing $0$, the following statements hold.
\begin{enumerate}[label={\rm(\arabic*)},leftmargin=7ex]
\item The flow $\widetilde x:I'\times M\to M$ exists and is unique.
\item If $\widetilde\rho:=\rho\circ\widetilde x$, then
\[
        \widetilde\rho\in \rho(1+\cA^{\infty,\epsilon}(I'\times M)).
\]
\item Fix a collar $[0,\delta[_\rho\times Y$ and local coordinates $(\omega_1,\ldots,\omega_{n-2})$ on $Y$, extended to the collar.  On a smaller coordinate neighborhood,
\[
        \widetilde\omega_i-\omega_i\in\cA^{\infty,\epsilon},
        \quad
        \widetilde\omega_i:=\omega_i\circ\widetilde x .
\]
\item Let $m\in\rr$, let $a_0\in\rho^m C^\infty(I'\times M)$, and assume
\[
        a-a_0\in\cA^{\infty,m+\epsilon}(I'\times M).
\]
Then, with $\widetilde a(t,x):=a(t,\widetilde x(t,x))$, one has
\[
        \widetilde a-a_0\in\cA^{\infty,m+\epsilon}(I'\times M).
\]
\end{enumerate}
\end{lemma}

\proof
The existence statement is standard.  Choose $N,k\in\mathbb N$ with $N\epsilon>k$ and introduce the smooth structure with boundary defining function $\check\rho=\rho^{1/N}$.  Since $\check\rho\partial_{\check\rho}=N\rho\partial_\rho$, the vector field $v$ is $C^k$ as a $\b$-vector field for this smooth structure.  Hence its flow exists, is unique, and is $C^k$ in $(t,x)$ after shrinking $I$ if necessary.  Since $v$ is tangent to $\partial M$, the flow preserves the boundary.  As $\widetilde\rho=\rho$ at $t=0$, we also have
\beq\label{eq:rho-comparable}
        c_0\rho\leq\widetilde\rho(t,\rho,\omega)\leq C_0\rho
\eeq
for $t\in I'$ and suitable constants $0<c_0<C_0$.

We work in a collar chart and write
\beq\label{eq:vf-errors}
        v=v_\rho\partial_\rho+\sum_{i=1}^{n-2}v_i\partial_{\omega_i},
        \ \ v_\rho\in\cA^{\infty,1+\epsilon},
        \ v_i\in\cA^{\infty,\epsilon}.
\eeq
The flow satisfies
\beq\label{eq:rho-omega-int}
\begin{aligned}
        \widetilde\rho(t)-\rho
        &=\int_0^t v_\rho(\tau,\widetilde\rho(\tau),\widetilde\omega(\tau))\,d\tau,\\
        \widetilde\omega_i(t)-\omega_i
        &=\int_0^t v_i(\tau,\widetilde\rho(\tau),\widetilde\omega(\tau))\,d\tau .
\end{aligned}
\eeq
Using \eqref{eq:rho-comparable} in \eqref{eq:vf-errors}, the integrands are respectively $O(\rho^{1+\epsilon})$ and $O(\rho^\epsilon)$.  Hence
\[
        \widetilde\rho-\rho=O(\rho^{1+\epsilon}),
        \quad
        \widetilde\omega_i-\omega_i=O(\rho^\epsilon).
\]

For the conormal estimates set
\[
        D_{\b}^{k,l,\alpha}:=\partial_t^k(\rho\partial_\rho)^l\partial_\omega^\alpha.
\]
We claim that for all $k,l,\alpha$,
\beq\label{eq:all-conormal-flow-est}
        D_{\b}^{k,l,\alpha}(\widetilde\rho-\rho)=O(\rho^{1+\epsilon}),
        \quad
        D_{\b}^{k,l,\alpha}(\widetilde\omega_i-\omega_i)=O(\rho^\epsilon).
\eeq
The proof is by induction on $k+l+|\alpha|$.  If $k>0$, differentiate the equations
$\partial_t\widetilde\rho=v_\rho(t,\widetilde\rho,\widetilde\omega)$ and
$\partial_t\widetilde\omega_i=v_i(t,\widetilde\rho,\widetilde\omega)$ and use the chain rule: each $\partial_\rho$ falling on a coefficient costs one power of $\rho^{-1}$, while the corresponding derivative of $\widetilde\rho$ contributes one power of $\rho$.

It remains to consider $k=0$.  Let $D=(\rho\partial_\rho)^l\partial_\omega^\alpha$ with $l+|\alpha|>0$.  After applying $D$ to \eqref{eq:rho-omega-int}, all terms are controlled by the induction hypothesis except the linear terms containing $D\widetilde\rho$ or $D\widetilde\omega_i$.  If
\[
        X(t):=|D(\widetilde\rho-\rho)|+
        \rho\sum_i|D(\widetilde\omega_i-\omega_i)|,
\]
then these terms give
\[
        X(t)\leq C\rho^{1+\epsilon}+C\rho^\epsilon\int_0^t X(\tau)\,d\tau .
\]
Gronwall's inequality yields $X(t)\leq C\rho^{1+\epsilon}$, proving \eqref{eq:all-conormal-flow-est}.

For the composition statement, factor out the conormal weight $\rho^m$ and write
\[
        a_0=\rho^m a_\infty,
        \quad
        a=\rho^m(a_\infty+\rho^\epsilon a'),
        \ \ a_\infty\in C^\infty,
        \ a'\in\cA^\infty .
\]
If $m$ is not a non-negative integer this is understood as a conormal factorization after removing the explicit factor $\rho^m$.
By \eqref{eq:rho-comparable} and \eqref{eq:all-conormal-flow-est},
\[
        \widetilde\rho^m-\rho^m\in\rho^m\cA^{\infty,\epsilon},
        \quad
        a_\infty(t,\widetilde\rho,\widetilde\omega)-a_\infty(t,\rho,\omega)
        \in\cA^{\infty,\epsilon}.
\]
Therefore $\widetilde a-a_0\in\cA^{\infty,m+\epsilon}(I'\times M)$, and the same chain-rule estimate gives the bounds after any number of $\b$-derivatives.\qed

\begin{corollary}\label{cor:ODE-small}
Let $v\in C^\infty(I;\rho\cA^\infty(M;{}^{\b}TM))$ and let $\widetilde x(t,\cdot)$ be its flow with $\widetilde x(0,\cdot)=\one$.  Then, after possibly shrinking $I$, in any product collar $[0,\delta[_\rho\times Y$ one has
\[
        \rho\circ\widetilde x=\rho+\rho^2\cA^\infty(I\times M),
        \quad
        \omega_i\circ\widetilde x=\omega_i+\rho\cA^\infty(I\times M).
\]
\end{corollary}
\proof
Apply Lemma~\ref{thm:ODE} with $\epsilon=1$.\qed

\section{Index of notation}\init\label{sec:notation-index}
In this appendix we collect the most important notation used in the main part of the paper.

{\small
\setlength{\tabcolsep}{3pt}
\renewcommand{\arraystretch}{1.22}
\begin{longtable}{>{\raggedright\arraybackslash}p{0.15\textwidth} >{\raggedright\arraybackslash}p{0.33\textwidth} >{\raggedright\arraybackslash}p{0.36\textwidth} >{\raggedleft\arraybackslash}p{0.10\textwidth}}
\hline
\emph{notation} & \emph{formula} & \emph{description} & \emph{}\\
\hline
$\cA^\infty(M)$, $\cA^{\infty,r}(M)$ & $(x\partial_x)^\alpha\partial_y^\beta u\in L^\infty$ over $M_\infty$  $\cA^{\infty,r}=x^r\cA^\infty$ & conormal functions and weighted conormal classes near a boundary & \defpage{not:conormal-classes}\\
$\be TM$, $\be T^*M$ & frames $(x\partial_x,\partial_y)$ and $(dx/x,dy)$ & $\b$-tangent and $\b$-cotangent bundles & \defpage{not:b-tensor-bundles}\\
$S^m(\rr^n)$, $S^m_1(\rr^n)$ & $\p_x^\alpha u\in O(\langle x\rangle^{m-|\alpha|})$ $S^m_1=S^m_{\cl}+S^{m-1}$ & symbols on Euclidean space and their classical leading parts & \defpage{secnot.1.1}\\
$\Diff^m(\rr^n)$, $\Diff^m_1(\rr^n)$ & $\sum_{|\alpha|+|\beta|\le m} a_{\alpha\beta}\partial_r^\alpha(r^{-1}\partial_\omega)^\beta$ & differential operators with prescribed orders of decay & \defpage{secnot.1.1b}\\
$\Diff^m_\b(M)$, $P_\norm$ & $P=\sum_{k+|\alpha|\le m} a_{k\alpha}D_y^\alpha(xD_x)^k$; $P_\norm$ freezes $a_{k\alpha}$ at $x=0$ & $\b$-differential operators and normal operators & \defpage{not:b-diff}, \appdefpage{sec9.4}\\
$H^{m,r}_\b(M)$ & $H^{m,r}_\b=x^rH^m_\b$ & weighted $\b$-Sobolev spaces & \defpage{sec9.2.3}\\
$V_k$ & $V_k=\cc\otimes^k_{\rm s}T^*M$ & symmetric complex $(0,k)$-tensor bundle & \defpage{sec1.1.2}\\
$(\cdot|\cdot)_{V_k}$ & $(u|v)_{V_k}=k!\,\bar u\cdot(\rg^{\otimes k})^{-1}v$ & fiber Hermitian form on $V_k$ & \defpage{eq:uvk}\\
$I$ & $I=\one-\frac14|\rg)(\rg|$ & trace reversal on two-tensors & \defpage{blurck}\\
$(\cdot|\cdot)_{I,V_2}$ & $(u|v)_{I,V_2}=(u|Iv)_{V_2}$ & physical Hermitian form & \defpage{not:physical-hermitian-forms}\\
$d$, $\delta$ & $(dw)_{\mu\nu}=\frac12(\nabla_\mu w_\nu+\nabla_\nu w_\mu)$; $\delta=d^*$ for $(\cdot|\cdot)_{V_k(M)}$ & symmetric differential and divergence & \defpage{sec1.2}\\
$D_1$, $D_2$ & $D_1=-\square_1$, $D_2=-\square_2+2\Riem_\rg$ & wave/Lichnerowicz operators & \defpage{e0.3bip}\\
$P$ & $P=D_2-I\circ d\circ\delta$ & linearized Einstein operator & \defpage{sec1.2}\\
$K$, $K^\star$ & $K=I\circ d$, $K^\star=\delta$ & pure gauge operator and harmonic gauge operator & \defpage{sec1.2}\\
$G_k$ & $G_k=G_{k{\rm ret}}-G_{k{\rm adv}}$ & causal propagator for $D_k$ & \defpage{sec1.2}\\
$\varrho_k$ & $\varrho_ku=\binom{u|_\Sigma}{\i^{-1}\nabla_\nu u|_\Sigma}$ & Lorentzian Cauchy data map & \defpage{e00.3a}\\
$\cU_k$ & $D_k\cU_k=0$, $\varrho_k\cU_k=\one$ & Cauchy evolution for $D_k$ & \defpage{e00.3}\\
$K_\Sigma$, $K_\Sigma^\dag$, $I_\Sigma$ & $K_\Sigma=\varrho_2K\cU_1$, $K_\Sigma^\dag=\varrho_1K^\star\cU_2$, $I_\Sigma=\varrho_2I\cU_2$ & gauge and trace-reversal maps on Cauchy data & \defpage{sec1.6.2}\\
$\rq_k$ & Green form definition in~\eqref{lub} & Cauchy charge for $D_k$ & \defpage{lub}\\
$\boldsymbol\Lambda_2^\pm$, $\boldsymbol\lambda_2^\pm$ & $\boldsymbol\Lambda_2^\pm=(\varrho_2G_2)^*\boldsymbol\lambda_2^\pm(\varrho_2G_2)$ & spacetime and Cauchy-surface covariances & \defpage{e00.2b}\\
$c_2^\pm$ & $\boldsymbol\lambda_2^\pm=\pm\rq_{I,2}\circ c_2^\pm$ & Cauchy-surface Hadamard projectors for two-tensors & \defpage{e00.4}\\
$\cN$, $\cN^\pm$ & $\cN=\{(x,\xi):\xi\cdot\rg^{-1}(x)\xi=0\}$; signs by time orientation & characteristic set and positive/negative energy shells & \defpage{sec1.7b}\\
$\reta$, $\reta_\b$ & $\reta=-r^2dt^2+\boldsymbol\delta(t)$; $\reta_\b=\rho^2\reta$ & Minkowski metric and its $\b$-rescaling in Beig--Schmidt coordinates & \defpage{e2.mink-hyp}\\
$\boldsymbol\delta(t)$, $\boldsymbol\delta_\b(t)$ & $\boldsymbol\delta=dr^2+r^2\cosh^2(t)h_{\bS^{n-2}}$; $\boldsymbol\delta_\b=(d\rho/\rho)^2+\cosh^2(t)h_{\bS^{n-2}}$ & spatial model metrics & \defpage{e2.mink-b-hyp}\\
$r$, $\rho$, $t$, $\omega$ & $x_0=r\sinh t$, $\rx=r\cosh t\,\omega$, $\rho=r^{-1}$ & Beig--Schmidt coordinates near spacelike infinity & \defpage{sec2.2}\\
$\chi$, $c$ & $\chi^*\rg=c^{-2}(-r^2dt^2+\rh(t))$ & diffeomorphism and conformal factor reducing the background metric & \defpage{eq:metric-reduced}\\
$\rg_\b$ & $\rg_\b=c^2\rho^2\rg=-dt^2+\rh_\b(t)$ & reduced $\b$-rescaled background metric & \defpage{e2.bmetric-reduced}\\
$\rh_\b(t)$ & $\rh_\b=\rho^2\rh(t)=\boldsymbol\delta_\b(t)+\rho\cA^\infty$ & reduced spatial $\b$-metric & \defpage{e2.metric-reduced-asymptotics}\\
$\tau_k$ & $\tau_k=\tau^{\otimes k}$, $\tau=\left(\begin{smallmatrix}-1&0\\0&1\end{smallmatrix}\right)$ & involution relating Hermitian and Hilbertian adjoints & \defpage{e2.12b}\\
$\rq_{k\b}$, $q_k$ & $\bar f\cdot\rq_{k\b}f=(f_1|f_0)_{V_k(\Sigma)}+(f_0|f_1)_{V_k(\Sigma)}$; $q_k=\mat{0}{\tau_k}{\tau_k}{0}$ & Lorentzian Cauchy charge and charge matrix & \defpage{e2.lorentzian-charge}\\
$\varrho_{k\b}$ & $\varrho_{k\b}u=({u|_{t=0}},{\i^{-1}\nabla_tu|_{t=0}})$ & Lorentzian $\b$-Cauchy trace for $D_{k\b}$ & \defpage{e2.10bc}\\
$I_\b$, $I_{\b\Sigma}$ & $I_\b=I$, $I_{\b\Sigma}=\mat{I_\b|_\Sigma}{0}{0}{I_\b|_\Sigma}$ & $\b$-trace reversal and its action on Cauchy data & \defpage{e2.10bf}\\
$\rq_{I,2\b}$ & $\rq_{I,2\b}=\rq_{2\b}\circ I_{\b\Sigma}$ & physical two-tensor Cauchy charge & \defpage{e2.physical-charge}\\
$\tilde\rq_{k\b}$, $\tilde\rq_{k\b\norm}$ & $\bar f\cdot\tilde\rq_{k\b}f=(f_1|f_0)_{\tV_k(\Sigma)}+(f_0|f_1)_{\tV_k(\Sigma)}$ & Euclidean charges, including their normal versions & \defpage{e2.euclidean-charge}\\
$U_\b$, $D_{k\b}$ & $U_\b=(c\rho)^{-n/2}$; $D_{k\b}=(c\rho)^{-1}U_\b D_kU_\b^{-1}(c\rho)^{-1}$ & conformal weight and reduced $\b$-operators & \defpage{e6.-1}\\
$\nabla_a$, $R_{k\b a}$ & $\nabla_a=\partial_{e_{a\b}}+R_{k\b a}$, $R_{k\b a}\in\cA_1^\infty$ & connection after conformal and $\b$-reduction & \defpage{e6.new15f}\\
$a_{k\b}(t)$ & $D_{k\b}=-(\nabla_t)^\star\nabla_t+a_{k\b}(t)$ & spatial part of the reduced operator & \defpage{prop3.normal-split}\\
$D_{k\b\norm}$, $a_{k\b\norm}$ & $D_{k\b\norm}=D_{k\b}^{\reta}$; $a_{k\b\norm}=a_{k\b}^{\reta}$ & normal model and its spatial operator & \defpage{prop10.1}\\
$d_\b$, $K_\b$ & $d_\b=(c\rho)U_\b dU_\b^{-1}(c\rho)^{-1}$; $K_\b=I_\b\circ d_\b$ & $\b$-gauge differential and gauge operator & \defpage{sec7.e2a}\\
$T_\b$ & $T_\b=(c\rho)^{-2}K_\b(c\rho)^2$ & conjugated gauge intertwiner in $D_{2\b}K_\b=T_\b D_{1\b}$ & \defpage{e7.0}\\
$d_{\b\Sigma}$, $K_{\b\Sigma}$ & $d_{\b\Sigma}\varrho_{1\b}w=\varrho_{2\b}d_\b w$; $K_{\b\Sigma}=I_{\b\Sigma}d_{\b\Sigma}$ & Cauchy-surface gauge maps & \defpage{e7.new2}\\
$K_{\b\Sigma\norm}$ & $K_{\b\Sigma\norm}\rho^{-1}=(K_{\b\Sigma}\rho^{-1})_\norm$ & normal Cauchy-surface gauge operator with shifted weight & \defpage{e7.new4}\\
$\treta$ & $\treta=r^2ds^2+dr^2+r^2\cos^2(s)h_{\bS^{n-2}}$ & Wick-rotated Minkowski metric & \defpage{e10.n-22}\\
$\cF_k$, $\cF_{k\Sigma}$ & $\cF_{k\Sigma}=\mat{\cF_k}{0}{0}{\cF_k}$; $dt\mapsto\i ds$ & Wick rotation transport between Lorentzian and Riemannian Cauchy data & \defpage{e10.n-7}\\
$\tD_{k\b\norm}$ & $\tD_{k\b\norm}=D_{k\b}^{\treta}$ & Wick-rotated normal operator & \defpage{sec10.2}\\
$\kappa_k$ & lift of $(\rho,s,\omega)\mapsto(\rho,-s,\omega)$ to $V_k$ & imaginary time reflection on tensor fields & \defpage{sec10.3b}\\
$\trho^\pm$, $\tsig$ & $\trho^\pm u=\binom{u(0^\pm)}{-\widetilde\nabla_su(0^\pm)}$, $\tsig=\mat{0}{-\one}{\one}{0}$ & Riemannian Cauchy traces and boundary symplectic matrix & \defpage{e10.normal-trace}\\
$\widetilde n_\norm^\pm$ & $\widetilde n_\norm^\pm v=-\widetilde\nabla_su^\pm(0^\pm)$ & Riemannian Dirichlet-to-Neumann maps & \defpage{e10.fixed-dtn}\\
$\tc_{k\norm}^\pm$ & $\tc_{k\norm}^\pm=\mp\trho^\pm(\tD_{k\b\norm})^{-1}(\trho)^*\tsig$ & Wick-rotated normal Calder\'on projectors & \defpage{sec10.5}\\
$c_{k\norm}^\pm$ & $c_{k\norm}^\pm=\cF_{k\Sigma}^{-1}\tc_{k\norm}^\pm\cF_{k\Sigma}$ & Lorentzian versions of the normal Calder\'on projectors & \defpage{sec10.8}\\
$\cH^{m,r}_\b(\Sigma_\norm)$ & $H^{m,r}_\b(\Sigma_\norm;V_k)\oplus H^{m-1,r}_\b(\Sigma_\norm;V_k)$ & normal Cauchy-data Sobolev space & \defpage{e01.16e}\\
$J$ & $J=\open{-\frac{n-2}{2},\frac{n-2}{2}}$ & weight interval for the normal calculus and projectors & \defpage{prop10.2}\\
$b^+$, $b^-$ & $b^+=b(0)$, $b^-=-b^\star(0)$ modulo residual terms & positive and negative symbolic factorization branches & \defpage{sec8.3}\\
$c^\pm$ & matrix formula~\eqref{e8.9bb} built from $b^\pm$ & Hadamard projectors for the reduced operators & \defpage{e8.9bb}\\
$\cH^{m,r}_\b(\Sigma)$ & $H^{m,r}_\b(\Sigma;V)\oplus H^{m-1,r}_\b(\Sigma;V)$ & weighted Cauchy data Sobolev space & \defpage{e8.9b}\\
$S^m_{\cl}(\be T^*M)$ & classical symbols with conormal $x$-dependence & principal symbol class for the extended $\b$-calculus & \appdefpage{sec9.2.5}\\
$\Psi^m_{{\rm s}\:1}(M)$ & $\Op(S^m_{\cl\:1}(\be T^*M))$ & pseudodifferential operators with a normal symbol & \appdefpage{def9.0}\\
$B_J^m(M)$ & $\bigcap_{p\in\rr,\,r\in J}B(H^{p,r}_\b,H^{p-m,r}_\b)$ & uniform weighted Sobolev mapping class & \appdefpage{def9.1b}\\
$B_{J,0^+}^{-\infty}(M)$ & $\bigcap_{\epsilon>0}\bigcup_{\delta>0}x^\delta B_{J(\epsilon)}^{-\infty}x^\delta$ & smoothing operators improving decay at the boundary & \appdefpage{def9.1c}\\
$N_J^m(M_\norm)$ & $n=\widetilde n+r_{-\infty}$, $T_s^{-1}nT_s=n$ & dilation-invariant normal operators & \appdefpage{def9.2}\\
$\Psi_J^m(M)$ & $a=\chi_\infty a_\norm\chi_\infty+x^{1/2}rx^{1/2}+r_{-\infty}$ & extended $\b$-calculus with prescribed normal operator & \appdefpage{def9.3}\\
$\hat\ell$ & $\hat\ell u_2=\binom{c^{-2}(n-2)^{-1}(\rg(0)|u_2)_{V_2}}{2u_{2t\Sigma}}$ & trace/synchronous constraint entering $R_\Sigma$ & \defpage{sec11.1}\\
$R_\Sigma$ & $R_\Sigma=\binom{\hat\ell\pi_0c_2^+}{\hat\ell\pi_0c_2^-}$ & microlocal TT-synchronous gauge map & \defpage{e10.3}\\
$M_\Sigma^\pm$ & $\hat\ell\pi_0K_{\b\Sigma}c_1^\pm=\rho M_\Sigma^\pm\pi_0c_1^\pm$ & first order elliptic boundary operator in  gauge-fixing equation & \defpage{e11.1b}\\
$T$ (regular) & $T=\one-K_{\b\Sigma}(R_\Sigma K_{\b\Sigma})^{-1}R_\Sigma$ & projection imposing the microlocal TT-synchronous gauge in the invertible case & \defpage{e11.4}\\
$\pi$ & $\pi f_2=\sum_k\bar v_k\cdot\rq_{I,2\b}f_2\,w_k$ & finite-rank projection used in the singular case & \defpage{e10.8}\\
$T_{\rm reg}$ & $T_{\rm reg}=\one-K_{\b\Sigma}(R_\Sigma K_{\b\Sigma})^{-1}R_\Sigma$ on the regular complement & gauge-fixing projection on the finite-codimension regular subspace & \defpage{e10.20}\\
$T$ (singular) & $T=(\one-\pi_2)\circ T_{\rm reg}\circ(\one-\pi)$ & projection used in the singular gauge-fixing case & \defpage{e10.21}\\
$\boldsymbol\lambda_{2\b}^\pm$ (regular) & $\boldsymbol\lambda_{2\b}^\pm=\pm T^*\circ\rq_{I,2\b}c_2^\pm\circ T$ & Cauchy surface covariances in the regular case & \defpage{e12.1}\\
$\boldsymbol\nu$, $\nu$ & $\boldsymbol\nu=\rq_{I,2\b}-(\one-\tilde\pi)^*\rq_{I,2\b}(\one-\tilde\pi)$ & finite-rank correction form in the singular case & \defpage{e11.12}\\
$\boldsymbol\lambda_{2\b}^\pm$ (singular) & see formula~\eqref{e11.10}  & Cauchy surface covariances in the singular case & \defpage{e11.10}\\
$\mathcal T_{k}$ & $\mathcal T_k\binom{f_0}{f_1}=\binom{c\rho U_\b f_0}{U_\b f_1}$ & transition from original Cauchy data to $\b$-Cauchy data & \defpage{e100}\\
$\boldsymbol\lambda_2^\pm$ & $\boldsymbol\lambda_2^\pm=\mathcal T_2^*\boldsymbol\lambda_{2\b}^\pm\mathcal T_2$ & physical Cauchy surface covariances & \defpage{e101}\medskip\\ 
\hline\end{longtable}
}

%
%\bibliographystyle{amsplain}
%\bibliography{merde.bib}
% \end{document}

 \end{document}